\documentclass[journal]{IEEEtran}
\usepackage[T1]{fontenc}% optional T1 font encoding

\def\argmax{\mathop{\rm \arg\!\max}}
\def\argmin{\mathop{\rm \arg\!\min}}
\usepackage{graphicx}
\usepackage[noadjust]{cite}
\usepackage{mcite}
\usepackage{amsfonts,helvet}
\usepackage{fancyhdr}
\usepackage{threeparttable}
\usepackage{epsf,epsfig}
\usepackage{amsthm}
\usepackage{amsmath}
\usepackage{siunitx}
\usepackage{amssymb}
\usepackage{dsfont}
\usepackage{subfigure}
\usepackage{color}
\usepackage{algorithm}
\usepackage{algpseudocode}
\usepackage{algcompatible}
\usepackage{enumerate}
\usepackage{hyperref}
\usepackage{cancel}
\usepackage{bbm}
\usepackage{dsfont}
\usepackage[subnum]{cases}
\usepackage{adjustbox}
\usepackage{enumitem}
\usepackage{dblfloatfix} 

\newcommand{\fig}[1]{Fig.\ \ref{#1}}
\newtheorem{theorem}{Theorem}

\newtheorem{corollary}{Corollary}

\newtheorem{lemma}{Lemma}

\newtheorem{proposition}{Proposition}
\newtheorem{remark}{Remark}

\newmuskip\pFqmuskip

\usepackage{array}

\makeatletter
\newcommand{\thickhline}{%
    \noalign {\ifnum 0=`}\fi \hrule height 1pt
    \futurelet \reserved@a \@xhline
}
\newcolumntype{"}{@{\hskip\tabcolsep\vrule width 1pt\hskip\tabcolsep}}
\makeatother

\IEEEoverridecommandlockouts

\title{
Resolution-Adaptive Hybrid MIMO Architectures \\for Millimeter Wave Communications}
\author{
Jinseok Choi, Brian L. Evans, and Alan Gatherer  \thanks{
J. Choi and B. L. Evans are with the Wireless Networking and Communication Group(WNCG), Department of Electrical and Computer Engineering, The Universityof Texas at Austin, Austin, TX 78701 USA. (e-mail: jinseokchoi89@utexas.edu, bevans@ece.utexas.edu).

A. Gatherer is with Wireless Access Lab., Huawei Technologies, Legacy Dr, Plano, TX  75024 USA. (e-mail: alan.gatherer@huawei.com).

The authors at The University of Texas at Austin were supported by gift funding from Huawei Technologies.
}
}
\begin{document}
\maketitle

%%%%%%%%%%%%%%%%%%%%%%%%%%%%
\begin{abstract}

%Hybrid analog-digital beamforming architectures with low-resolution analog-to-digital converters (ADCs) reduce hardware cost and power consumption in multiple-input multiple-output (MIMO) millimeter wave (mmWave) communication systems. 
In this paper, we propose a hybrid analog-digital beamforming architecture with resolution-adaptive ADCs for millimeter wave (mmWave) receivers with large antenna arrays. 
{\color {black} We adopt array response vectors for the analog combiners and derive ADC bit-allocation (BA) solutions in closed form. 
The BA solutions reveal that the optimal number of ADC bits is logarithmically proportional to the RF chain's signal-to-noise ratio raised to the $1/3$ power. 
Using the solutions, two proposed BA algorithms minimize the mean square quantization error of received analog signals under a total ADC power constraint. 
%It is beneficial to assign more bits to the ADC with a larger channel gain on the corresponding radio frequency (RF) chain, and 
Contributions of this paper include 
1) ADC bit-allocation algorithms to improve communication performance of a hybrid MIMO receiver, 2) approximation of the capacity with the BA algorithm as a function of channels, and 3) a worst-case analysis of the ergodic rate of the proposed MIMO receiver that quantifies system tradeoffs and serves as the lower bound. }% 1) an ADC bit allocation algorithm to improve communication performance of a hybrid MIMO receiver, 2) a revised ADC bit allocation algorithm that is robust to additive noise, and 3) a worst-case analysis of the ergodic rate of the proposed MIMO receiver that quantifies system tradeoffs and serves as the lower bound. 
Simulation results demonstrate that the BA algorithms outperform a fixed-ADC approach in both spectral and energy efficiency, and validate the capacity and ergodic rate formula.
For a power constraint equivalent to that of fixed 4-bit ADCs, the revised BA algorithm makes the quantization error negligible while achieving $22$\% better energy efficiency.
Having negligible quantization error allows existing state-of-the-art digital beamformers to be readily applied to the proposed system.
\end{abstract}
%%%%%%%%%%%%%%%%%%%%%%%%%%%%

%%%%%%%%%%%%%%%%%%%%%%%%%%%%
\section{Introduction}
\label{sec:intro}
%%%%%%%%%%%%%%%%%%%%%%%%%%%%

% Millimerter-wave communicatoin introduction
%\IEEEPARstart{M}{illimeter} wave communication has drawn extensive attention as a promising technology for next-generation cellular systems~\cite{pi2011introduction,andrews2014will,boccardi2014five}, and evinced its feasibility \cite{rappaport2013millimeter}.
Moving to a millimeter wave spectrum in range of 30\textendash 300 GHz enables the utilization of multi-gigahertz bandwidth and offers an order of magnitude increase in achievable rate \cite{pi2012millimeter,swindlehurst2014millimeter,bai2015coverage}.
The small wavelength allows a large number of antennas to be packed into tranceivers with very small antenna spacing.
Leveraging the large antenna arrays, mmWave systems can manipulate directional beamforming to produce high beamforming gain, which helps overcome large free-space pathloss of mmWave signals.
% and maintains a reasonable level of received signal-to-noise ratio (SNR).
%The advantages of remarkably wide bandwidth and large antenna arrays have encouraged wireless researchers to perform comprehensive studies to resolve practical challenges in the realization of mmWave communications \cite{niu2015survey, heath2016overview}. 
%Unlike the traditional MIMO communication that operates sub-3 GHz with a small number of antennas, 
Problems with hardware cost and power consumption, however, arise from deploying large antenna arrays coupled with power-demanding ADCs. 
%due to a large signal bandwidth and a high number of bits/sample in mmWave communications. 
To overcome these challenges, hybrid analog-digital beamforming architectures \cite{han2015large} that attempt to reduce the burden of fully digital beamforming, and receivers with low-resolution ADCs \cite{mo2015capacity} have attracted the most interest in recent years.
To take advantage of the two considered architectures, we propose a hybrid massive-MIMO architecture with resolution-adaptive ADCs for mmWave communications. 
%We develop ADC bit allocation (BA) algorithms to maximize the communication performance with limited ADC power consumption.
%We also perform a worst-case analysis for the proposed architecture by 
%deriving the achievable rate of the system where ADCs have infinite switching period, which converges to the conventional ADC scenario.

%%%%%%%%%%%%%%%%%%%%%%%%%%
\subsection{Prior Work}
\label{subsec:prior}
Hybrid architectures employ a less number of RF chains than the number of antennas to reduce power consumption and system complexity.
An analog beamformer is the pivotal component that enables the hybrid structure to reduce the number of RF chains  \cite{el2014spatially,alkhateeb2014channel}.
% What is different from conventional system with respect to beamforming design?
% Since analog beamforming is commonly considered to use sets of phase shifters \cite{zhang2005variable, venkateswaran2010analog}, the design of an analog precoder and combiner is limited by its constant amplitude \cite{el2014spatially}, which leads to separate analog and digital beamformer design.
%, and the state of the art hybrid beamforming design algorithms have been proposed with the goal of achieving spectral efficiency close to that of the system with fully digital beamformers \cite{el2012capacity,el2012low, alkhateeb2013hybrid, el2014spatially, liang2014low, alkhateeb2015limited, sohrabi2015hybrid}.
% Condition for optimal performance
%In \cite{sohrabi2015hybrid}, it was shown that the number of RF chains are required to be at least twice the number of data streams to realize the performance of fully digital beamforming.
% which corresponds to the requirement for single stream case in \cite{zhang2005variable}.
% array response Vector is a good solution
An analog beamformer is often designed by selecting array response vectors corresponding to the dominant channel eigenmodes \cite{el2012capacity,el2012low, alkhateeb2013hybrid, el2014spatially, liang2014low, alkhateeb2015limited}. 
Indeed, it was shown that the optimal RF precoder and combiner converge to array response vectors in dominant eigenmodes \cite{el2012capacity}.
%Since the beam steering alone cannot perfectly detect the dominant channel eigenmodes, 
Motivated by this, orthogonal matching pursuit (OMP) was used to develop beamformer design algorithms \cite{el2012low,alkhateeb2013hybrid,el2014spatially}.
%which composes RF beamformer with the array response vectors by estimating the dominant eigenmodes.
% Additional approach
%In addition, low-complexity hybrid precoding algorithms in multi-user MIMO downlink systems were proposed by considering zero-forcing precoding \cite{liang2014low} and limited feedback \cite{alkhateeb2015limited}.
% Limitation (Full-resolution)
Although the hybrid beamforming approaches in \cite{el2014spatially,alkhateeb2014channel,el2012capacity,el2012low, alkhateeb2013hybrid, liang2014low, alkhateeb2015limited, zhang2005variable, venkateswaran2010analog} delivered remarkable achievements in the development of the low-power and low-complexity architecture with large antenna arrays, the hybrid architectures still assume high-resolution ADCs that consume a high power at receivers.
% \cite{murmann2015adc}.

% Low-Resolution ADCs
Since power consumption of ADCs scales exponentially in terms of the number of quantization bits \cite{lee2008analog}, employing low-resolution ADCs can be indispensable to reduce hardware cost and power consumption in the large antenna array regime.
Consequently, low-resolution ADC architectures have been investigated 
\cite{mezghani2007ultra,risi2014massive,jacobsson2015one,mo2014channel,wang2014multiuser,wen2016bayes,mo2016channel,mezghani2012capacity,li2016channel,orhan2015low,fan2015uplink,zhang2016spectral}.
%without any analog combining 
% Showing only $\frac{2}{\pi}$ ($1.96$ dB) loss of mutual information in the low SNR regime \cite{mezghani2007ultra}, the utilization of 1-bit ADCs at the receiver demonstrated its cost and energy efficiency.
%, especially with the large antenna arrays \cite{risi2014massive}.
It was revealed that least-squares channel estimation and maximum-ratio combining (MRC) with 1-bit ADCs are sufficient to support multi-user operation with quadrature-phase-shift-keying \cite{risi2014massive}, which is known to be optimal for 1-bit ADC systems \cite{mezghani2007ultra, mo2015capacity}.
%, or even with higher-order modulation \cite{jacobsson2015one}. 
Deploying large antenna arrays provided an opportunity to use message-passing and expectation-maximization algorithms for symbol detection and channel estimation with low complexity \cite{mo2014channel,wang2014multiuser,wen2016bayes, mo2016channel}.
To examine the effect of quantization in achievable rate, the Bussgang decomposition \cite{mezghani2012capacity,li2016channel} was utilized for linear expressions of quantization operation.
The analysis in \cite{mezghani2012capacity} revealed that noise correlation can reduce the capacity loss to less than $\frac{2}{\pi}$ at low signal-to-noise ratio (SNR).
A lower bound for the achievable rate of the 1-bit ADC massive MIMO system was derived \cite{li2016channel}, using MRC detection with a linear minimum mean square error (MMSE) channel estimator.
%Although the AQNM provides an approximated linear quantization model, it showed a reasonable degree of accuracy in the low and medium SNR regimes \cite{orhan2015low}.
Offering an analytical tractability, the additive quantization noise model (AQNM) \cite{orhan2015low,fan2015uplink,zhang2016spectral,zhang2017performance} were adopted to derive the achievable rate of massive MIMO systems with low-resolution ADCs using MRC in Rayleigh \cite{fan2015uplink} and Rician fading channels \cite{zhang2016spectral}.

%Hybrid Low resolution => Constribution (done)
%Hybrid in mmWave => steering okay we fix F_RF = A (system model)

%Jaobsson15 => MRC okay why we choose to analyze MRC (not needed?)

%Mixed ADC => switching => Constribution
%` Works => WB is pretty much the same => Contribution
%conjecture that BA would have similar performance in WB as

%While previous studies have derived meaningful discussions,
%regarding an energy-efficient massive MIMO design,
The considered architectures in the previous studies, however, present two extreme points: (1) less number of RF chains with high-resolution ADCs and (2) low-resolution ADCs with full number of RF chains.
% in terms of the number of RF chains and quantization bits.
%the hybrid receivers consider a small number of RF chains with full-resolution ADCs, while the low-resolution ADC receivers assume no reduction on the number of RF chains. 
%There is prior work \cite{mo2016achievable} that studied a general version of these two extreme points, where hybrid architecture with low-resolution ADCs. In \cite{mo2016achievable}, it was shown that hybrid architecture with low-resolution ADCs achieves high energy efficiency.
One prior study with less extremity \cite{mo2016achievable} focused on a generalized system consisting of less number of RF chains with low-resolution ADCs.
%: a hybrid architecture with low-resolution ADCs. 
In \cite{mo2016achievable}, the spectral efficiency was analyzed under a constant channel assumption.
It is also assumed that each ADC's resolution is predetermined regardless of channel gain on each RF chain.
In another line of research, mixed-ADC architectures were proposed \cite{zhang2017performance, liang2016mixed,zhang2016mixed}.
{\color{black} In \cite{zhang2017performance}, performance analysis of mixed-ADC systems where receivers use a combination of low-resolution and high-resolution ADCs showed that the architecture can achieve a better energy-rate tradeoff compared to systems either with infinite-resolution ADCs or low-resolution ADCs.}
In \cite{liang2016mixed,zhang2016mixed} each antenna uses different ADC resolution depending on its channel gain. 
This system has explicit benefits compared to fixed low-resolution ADC systems such as increase of channel estimation accuracy and spectral efficiency.
In \cite{liang2016mixed,zhang2016mixed}, however, they force antennas to select between $1$-bit ADC and $\infty$-bit ADC, which is far from an energy-efficient architecture mainly because the total ADC power consumption can be dominated by only a few high-resolution ADCs.
%This is mainly because the ADC power consumption exponentially increases with the quantization bits and therefore, the total ADC power consumption can be dominated by only a few high-resolution ADCs. 
Moreover, it assumes full number of RF chains, which leads to dissipation of energy.
% in mmWave channels due to the channel sparsity.
For these reasons, an adaptive ADC design for a hybrid beamforming architecture is still questionable. 
%This paper proposes a possible architecture design foIn \cite{zhang2017performance}, performance of mixed-ADC systems was  r such a question.

%%%%%%%%%%%%%%%%%%%%%%%%%%
\subsection{Contributions}
\label{subsec:contributtion}
%%%%%%%%%%%%%%%%%%%%%%%%%%

% SYSTEM PROPOSITION
{\color{black} Our main contribution is the proposition of a hybrid beamforming MIMO architecture with resolution-adaptive ADCs to offer a potential energy-efficient mmWave receiver architecture. 
Under this architecture proposition, we investigate the architecture as follows: we first $(i)$ develop two bit-allocation (BA) algorithms to exploit the flexible ADC architecture and derive a capacity expression for a given channel realization.
Due to the intractable ergodic rate analysis with BA, we then $(ii)$ perform the analysis without BA, offering the baseline performance of the proposed receiver architecture.}
% In this paper, we propose a hybrid beamforming architecture with resolution-adaptive ADCs for a MIMO mmWave receiver. 
The proposed architecture is distinguishable from many other systems because it not only consists of a lower number RF chains and low-resolution ADCs \cite{mo2016achievable} but also adapts the ADCs resolutions \cite{liang2016mixed,zhang2016mixed}.
%While only the first case was considered in , the authors in  merely assumed the second case.
%, which allows binary ADC resolution switching between 1-bit and $\infty$-bit. 
In the context of mmWave communications, we design the analog combiner to be a set of array response vectors to aggregate channel gains in the angular domain. 
Such design approach is beneficial as the sparse nature of mmWave channels in the angular domain allows the number of RF chains to be less than the number of antennas. 
%In addition, the optimal hybrid combiner becomes the array response vectors in the large antenna array regime \cite{el2012capacity}.
%Regarding a digital combiner, 
Leaving the design issue of digital combiners, we focus on the quantization problem for the proposed system.

% BIT ALLOCATION
Given the different channel gains on RF chains, the system performance can be improved by leveraging the flexible ADC architecture. 
 {\color{black} To this end, as an extension of our work \cite{choi2016adc}, we derive a close-form BA solution for a minimum mean square quantization error (MMSQE) problem subject to a constraint on the total ADC power.
Using the solution, we develop a BA algorithm that determines ADC resolutions depending on angular domain channel gains.
The derived solution provides an explicit relationship between the number of quantization bits and channel environment.
% and ADC resolution switching is not limited to binary switching between $1$ bit and $\infty$ bits as in \cite{liang2016mixed,zhang2016mixed}.} 
%the total ADC power constraint and channel environment such as channel gains in the angular domain, transmit power and noise power.
%We note that the prior work \cite{liang2016mixed,zhang2016mixed} %investigated the similar ADC architecture which switches ADC resolution  depending on system parameters.
%also investigated the dynamic ADC structure where the used ADCs' quantization bits are changed as a function of system parameter.
%This work is relevant to our paper in a sense that  
%We note that the primary limitations of previous work \cite{liang2016mixed,zhang2016mixed} are that the resolution switching scheme does not provide an explicit relationship between system parameters and channel environment, and that switching is limited to binary switching between $1$ bit and $\infty$ bits.
%, which is not optimal in the power consumption and performance trade-off. 
%To overcome these drawbacks, we derive a close-form BA solution for a mean square quantization error (MSQE) problem subject to the constrained total ADC power.
 %and the proposed BA algorithms are designed to maximize the proposed system performance by minimizing total quantization error under a power constraint.
%%%%%
%Specifically, we derive a close-form BA solution for a mean square quantization error (MSQE) problem subject to the constrained total ADC power. 
%%%%%
One major finding from the solution is that the optimal number of ADC bits is logarithmically proportional to the corresponding RF chain's SNR raised to the $1/3$ power. 
This result quantifies the conclusion made in \cite{liang2016mixed} that allocating more bits to the RF chain with stronger channel gain is beneficial.
%{\color{black} Although a similar conclusion that stronger channels should be assigned high-resolution ADCs was drawn for mixed-ADC systems in \cite{liang2016mixed}, our close-form BA solution mathematically quantifies the near optimal number of bits that ADCs on each RF chain should be allocated with.}
%, since it is more advantageous to extract the channel gain from each RF chain. 
%In particular, it is shown that the optimal number of ADC bits is logarithmically proportional to the corresponding RF chain's SNR raised to the $1/3$ power.
%One finding from the obtained MMSQE-BA solution is that the optimal resolution bis scales logarithmically with the $1/3$ order of the SNR. This is mainly because the ADC power consumption exponentially increases with the ADC quantization bits, so that XXX
%This will be explained in more detail in the later section.
%Interestingly, the MMSQE-BA illustrates that the optimal number of quantization bits scales logarithmically with the $1/3$ order of the SNR due to the non-linearity of the ADC power consumption in term of the number of bits.
} 
We also derive a solution for a revised MMSQE problem to modify the proposed BA method to be robust to noise.
{\color{black}  We show that the revised MMSQE problem is equivalent to maximizing generalized mutual information (GMI) in the low SNR regime.
Applying the solution to a capacity, we approximate the capacity with the revised MMSQE-BA algorithm as a function of channels.}
%and thus fits better in mmWave communications. 
Simulation results disclose that the BA algorithms achieve a higher capacity and sum rate than the conventional fixed-ADC system where all ADCs have same resolution.
% the revised BA method provides consistent improvement regardless of noise power.
In particular, the revised BA algorithm provides the sum rate close to the infinite-resolution ADC system while achieving higher energy efficiency than using fixed ADCs in the low-resolution regime.
%Such small distortion allows the existing state-of-the-art digital beamforming techniques to be applied to the proposed system.
% This implies that the revised BA algorithm makes quantization distortion negligible with small power consumption, which allows existing digital beamformers for hybrid mmWave receivers to be readily applied to the proposed system.

% IMPLEMENTATION ISSUE
Regarding the implementation issue of the BA algorithms, the best scenario is to operate the resolution switching at the time-scale of the channel coherence time. 
This is because the proposed BA algorithms allocate different quantization bits to each ADC depending on the channel gain on each RF chain. 
Accordingly, if the switch is able to operate at the channel coherence time, the proposed architecture is able to adapt to channel fluctuations.
Such coherence time switching in mmWave channels, however, may not be feasible due to the very short coherence time of mmWave channels \cite{khan2012millimeter}. 
Consequently, the switching period may need to be the multiples of the coherence time. 
In this case, switching at the time-scale of slowly changing channel characteristics such as large-scale fading and angle of arrival (AoA) marginally degrades the performance of the BA algorithms.
Then, the worst-case scenario is not to exploit the flexibility of ADC resolutions, which is equivalent to have an infinitely long switching period, and indeed converges to fixed-ADC architectures.
%{\color {black} 
%Simulation results for these scenarios demonstrate that the coherence time switching improves the achievable rate of the system compared to the fixed-ADC architecture and the slow switching also increases the rate with marginal loss compared to the coherence time switching.}

 %Notably, the revised BA method provides consistent improvement over the entire SNR range, which is not the case for the other method.

% WORST-CASE ANALYSIS
{\color{black} To provide deeper insight for the proposed system, we further perform an ergodic rate analysis.
As mentioned, due to the intractability of the anaylsis with the BA algorithms, we derive an approximation of the ergodic rate for the considered architecture without applying BA---the worst-case analysis---for analytical tractability.}
% we look into the considered architecture without using BA---that is, the worst-case analysis---i.e., fixed-ADC architecture, for analytical tractability;
%Although the performance analysis with the BA solution is not tractable, the worst-case analysis is.
%Unfortunately, the performance analysis applying the BA solution is not tractable, leaving us with the worst-case analysis.
%The best scenario of the proposed architecture with the BA algorithm is to operate the resolution switching at the time-scale of the channel coherence time. 
%This is because the proposed BA algorithms allocate different quantization bits to each ADC depending on the channel gain on each RF chain. 
%Accordingly, if the switch is able to operate at the channel coherence time, the proposed architecture is able to adapt to channel fluctuations.
%Such coherence time switching in mmWave channels, however, may not be feasible due to the very short coherence time of mmWave channels \cite{khan2012millimeter}. 
%Consequently, the switching period may need to be the multiples of the coherence time. 
%In this case, switching at the time-scale of slowly changing channel characteristics such as large-scale fading and angle of arrival (AoA) marginally degrades the performance of the BA algorithms.
%Then, the worst-case scenario is not to exploit the flexibility of ADC resolutions, which is equivalent to have an infinitely long switching period, and converges to fixed-ADC architectures with analog beamforming.
%we derive an approximation of the ergodic rate in mmWave fading channels using MRC.
Although the analysis focuses on the worst-case scenario, we shed light on the importance of the derived rate for the following reasons:
\begin{itemize}[leftmargin=*]
\item The obtained achievable rate can serve as the lower bound of the proposed architecture.
Hence, we conjecture that the proposed system can achieve a higher ergodic rate than the derived rate by leveraging the flexible ADCs.
\item As a function of system parameters, the tractable rate provides a broad insight for the considered system.
%: the number of quantization bits, antennas, RF chains, users and channel propagation paths, and transmit power, 
We observe that the achievable rates for the BA algorithms and for the fixed ADC show similar trends. 
In this regard, the derived rate provides general tradeoffs of the proposed architecture in terms of system parameters including quantization effect.
\item The analysis in \cite{mo2016achievable} considered a quasi-static channel.
This setting, however, ignores the transmission of a coded packet over different fading realizations so that rate adaptation cannot be applied over multiple fading realization.
%precluding rate adaptation techniques from being applied over multiple fading realizations. 
Especially, the quasi-static setting is not adequate in mmWave channels with the short coherence time \cite{khan2012millimeter}. 
Arguably, our ergodic rate analysis offers more realistic evaluation in contemporary wireless systems that transmit a coded packet over multiple fading realizations \cite{lozano2012yesterday}.
% as it considers mmWave fading channels in the ergodic sense
\end{itemize}

 %than that of comparable microwave systems \cite{khan2012millimeter}.
%The analysis for the generalized hybrid architecture with low-resolution ADCs was performed in \cite{mo2016achievable}, showing the performance comparable to that of fully digital or hybrid architectures with high-resolution ADCs in the low-to-medium SNR range. 
%The channel model in \cite{mo2016achievable}, however, is a quasi-static channel which is constant and not appropriate for mmWave channels with shorter channel coherence time than that of comparable microwave systems \cite{khan2012millimeter}.
%We show that numerical studies validate the accuracy of the analytical expression. 

{\it Notation}: $\bf{A}$ is a matrix and $\bf{a}$ is a column vector. 
$\mathbf{A}^{H}$ and $\mathbf{A}^\intercal$  denote conjugate transpose and transpose. 
$[{\bf A}]_{i,:}$ and $ \mathbf{a}_i$ indicate the $i$th row and column vector of $\bf A$. 
We denote $a_{i,j}$ as the $\{i,j\}$th element of $\bf A$ and $a_{i}$ as the $i$th element of $\bf a$. 
$\mathcal{CN}(\mu, \sigma^2)$ is a complex Gaussian distribution with mean $\mu$ and variance $\sigma^2$. 
$\mathbb{E}[\cdot]$ and ${\rm Var}[\cdot]$ represent expectation and variance operator, respectively.
The cross-correlation matrix is denoted as ${\bf R}_{\bf xy} = \mathbb{E}[{\bf x}{\bf y}^H]$.
The diagonal matrix $\rm diag(\bf A)$ has $\{a_{i,i}\}$ at its diagonal entries, and $\rm diag (\bf a)$ or $\rm diag(\bf a^\intercal)$ has $\{a_i\}$ at its diagonal entries. 
${\bf I}_N$ is a $N\times N$ identity matrix and $\|\bf A\|$ represents L2 norm. 
Inequality between vectors, e.g. $\bf a < b$, is element-wise inequality, and ${\rm tr}(\cdot)$ is a trace operator. 
We denote $(a)_{b}^{+} = \max(a,b)$ and $(a)_{b}^{-} = \min(a,b)$.

%FIGURE 
\begin{figure}[!t]\centering
\includegraphics[scale = 0.43]{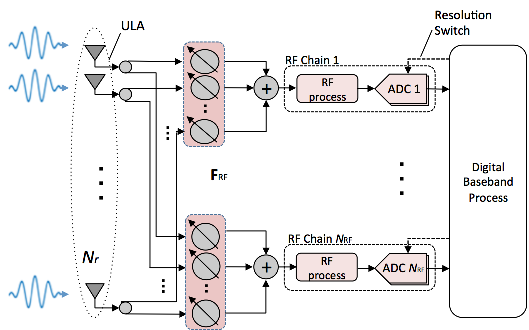}
\caption{A hybrid beamforming architecture with resolution-adaptive ADCs.} 
\label{fig:system}
\end{figure}

%%%%%%%%%%%%%%%%%%%%%%%%%%%%
\section{System Model}
\label{sec:sys_model}
%%%%%%%%%%%%%%%%%%%%%%%%%%%%

\subsection{Network and Signal Models}

We consider single-cell MIMO uplink networks in which $N_{u}$ users with a single transmit antenna are served by a base station (BS) with $N_r$ antennas. 
We assume large antenna arrays at the BS ($N_r \gg N_u$).
%For a energy-efficient system design based on a special property of mmWave channels known as sparsity, 
The hybrid architecture with low-resolution ADCs is employed at the BS.
We focus on uniform linear array (ULA) and assume that there are $N_{\rm RF}$ RF chains connected to $N_{\rm RF}$ pairs of ADCs.
{\color{black} Employing adaptive ADCs such as flash ADCs, we consider the proposed system to be able to switch quantization resolution.
Indeed, many power and resolution adaptive flash ADCs have been fabricated \cite{yoo2002power,nahata2004high,rajashekar2008design}, and flash ADCs are the most suitable ADCs for applications requiring very large bandwidth with moderate resolution \cite{le2005analog}.}
%, we investigate two cases: (i) ADC resolution is fixed as a typical communication system i.e., the switching period is infinite, and (ii) ADC resolution is adaptive as the switches update at every channel coherence time.

%, and hence, we consider resolution adaptive flash ADCs with resolution switches at the BS.

%The switches update at every channel coherence time.
%+ For 200 MHz signal bandwidth and 1-4 bit resolution, comparators consume most power
%+ Resolution can be reduced by disabling comparators, which would scale power consumption exponentially in the number of bits of resolution
%+ 
%Flash ADCs could be designed for basestations to change the resolution according to channel gains and achieve significant power savings while having minimal error (as shown in our ICASSP paper)

Assuming a narrowband channel, the received baseband analog signal ${\bf r} \in \mathbb{C}^{N_r}$ at the BS is expressed as
\begin{align}
\label{eq:rx_signal}
{\bf r} = \sqrt{p_u}{\bf Hs} + \tilde {\bf n}
\end{align} 
where $p_u$ is the average transmit power of users, $\bf H$ represents the $N_r \times N_u$ channel matrix between the BS and users, $\bf s$ indicates the $N_u \times 1$ vector of symbols transmitted by $N_u$ users and
$\tilde{\bf n} \in \mathbb{C}^{N_r}$ is the additive white Gaussian noise which follows complex Gaussian distribution $\tilde{\bf n}  \sim \mathcal{CN}(\mathbf{0}, \mathbf{I}_{N_r})$. 
%A user symbol has zero mean and unit variance.
We further consider that the transmitted signal vector ${\bf s} \sim \mathcal{CN}(\mathbf{0}, \mathbf{I}_{N_r})$ is Gaussian distributed with a zero mean and unit variance.
%, i.e., $\mathbb{E}\big[s_i\big] = 0$ and $\mathbb{E}\big[|s_i|^2\big] = 1$.
We also assume that the channel $\bf H$ is perfectly known at the BS.
%\footnote{Efficient estimation algorithms have been proposed for mmWave channels \cite{} by exploiting the sparse nature of mmWave channels. In the hybrid receiver structure where $N_{\rm RF} < N_r$, state-of-art channel estimators such as bisectional approach \cite{}, MG-OMP \cite{}, and DGMP \cite{} validated the estimation performance.  }

An analog beamformer ${\bf W_{\rm RF}}  \in \mathbb{C}^{N_r \times N_{\rm RF} }$ is applied to ${\bf r}$ and constrained to satisfy $[{\bf W_{\rm RF} }{\bf W}^H_{\rm RF}]_{i,i} = 1/N_r$, i.e., all element of $\bf W_{\rm RF}$ have equal norm of $1/\sqrt{N_r}$.
% EQUATION
\begin{align}
\label{eq:rx_rf_signal} 
\mathbf{ y}
={\bf W}^H_{\rm RF}\,\mathbf{r} 
= \sqrt{p_u}\mathbf{W}^H_{\rm RF}\,\mathbf{Hs} + \mathbf{W}^H_{\rm RF}\,\tilde{\bf n}.
\end{align} 
We consider that the number of RF chains is less than the number of antennas ($N_{\rm RF} < N_r$), alleviating the power consumption and complexity at the BS. 
Each beamforming output $y_i$ is connected to an ADC pair as shown in \fig{fig:system}.
At each ADC, either a real or imaginary component of the complex signal $y_i$ is quantized.

%where ${\bf H} \in \mathbb{C}^{N\times M}$ is the channel matrix, ${\bf x} \in \mathbb{C}^{M} $ is the vector of $M$ user symbols, $p_u$ represents the average transmitted power of each user and ${\bf n} \in \mathbb{C}^{N}$ is the additive white Gaussian noise which follows complex Gaussian distribution $\mathcal{CN}(\mathbf{0}, \mathbf{I}_N)$. 
%A user symbol has zero mean and unit variance.
%We assume $\mathbb{E}[x_i] = 0$ and $\mathbb{E}[|x_i|^2] = 1$.

%%%%%%%%%%%%%%%%%%%%%%%%%%%
\subsection{Channel Model}
\label{subsec:channel}
%%%%%%%%%%%%%%%%%%%%%%%%%%%

In this paper, we consider mmWave channels. 
Since mmWave channels are expected to have limited scattering~\cite{sayeed2007maximizing,el2014spatially,heath2016overview}, each user channel is the sum of contributions of $L$ scatterings and $L \ll N_r$.
%We consider $L > 1$ for multiple scatterers.
Adopting a geometric channel model, the $k$th user channel with $L_k$ scatterers that contribute to $L_k$ propagation paths is expressed as
\begin{align}
\label{eq:channel_geo}
{\bf h}_k = \sqrt{\gamma_k}\sum_{\ell = 1}^{L_k}g^k_\ell {\bf a}(\theta^k_\ell) \in \mathbb{C}^{N_r}
\end{align}
where $\gamma_k$ denotes the large-scale fading gain that includes geometric attenuation, shadow fading and noise power between the BS and $k$th user, $g^k_{\ell}$ is the complex gain of the $\ell$th path for the $k$th user and ${\bf a}(\theta^k_{\ell})$ is the BS antenna array response vector corresponding to the azimuth AoA of the $\ell$th path for the $k$th user $\theta^k_{\ell} \in [-\pi/2,\pi/2]$.
Each complex path gain $g^k_{\ell} \sim \mathcal{CN}(0, 1)$ is assumed to be an independent and identically distributed (IID) complex Gaussian random variable.
% , which follows $\mathcal{CN}(0, 1)$. 
We also assume that the number of propagation paths $L_k$ is distributed as $L_k \sim \max \big\{Poisson(\lambda_{\rm p}), 1\big\}$ \cite{akdeniz2014millimeter} for $k = 1,\cdots,N_u$. 
We call $\lambda_{\rm p} \in \mathbb{R}$ as the near average number of propagation paths.
%We only consider the azimuth and ignore the zenith, which implies that the scattering happens in azimuth.

Under the ULA assumption, the array response vector is 
\begin{align}
\nonumber
% {\bf a}(\theta) = \frac{1}{\sqrt{N_r}}\Big[1,e^{-j \frac{2\pi d}{\lambda}\sin(\theta)},e^{-j \frac{4\pi d}{\lambda}\sin(\theta)},\dots,e^{-j \frac{2(N_r-1)\pi d}{\lambda}\sin(\theta)}\Big]^\intercal 
{\bf a}(\theta) = \frac{1}{\sqrt{N_r}}\Big[1,e^{-j 2\pi\vartheta},e^{-j 4\pi \vartheta},\dots,e^{-j 2(N_r-1)\pi\vartheta}\Big]^\intercal 
\end{align}
where $\vartheta$ is the normalized spatial angle that is $\vartheta = \frac{d}{\lambda}\sin(\theta)$, $\lambda$ is a signal wave length, and $d$ is the distance between antenna elements.
Considering the uniformly-spaced spatial angle, i.e., $\vartheta_i = \frac{d}{\lambda}\sin(\theta_i) = (i-1)/N_r$, the matrix of the array response vectors
%$\frac{d}{\lambda}\sin(\theta_i) = (i-1)/N_r$, the matrix 
%To represent $\bf H$ in beamspace, we consider uniformly spaced angles $\vartheta_i = (i-1)/N_r$, $i = 1, \cdots, N_r$, and define the $N_r \times N_r$ matrix with corresponding steering vectors of $\theta_i = \arcsin(\lambda\vartheta_i/d)$ as
%\begin{align}
$\mathbf{A}=\big[{\bf a}(\theta_1),\cdots,{\bf a}(\theta_{N_r})\big]$
%\end{align}
becomes a unitary discrete Fourier transform matrix; ${\bf A}^{H}{\bf A} = {\bf A}{\bf A}^{H} = \bf I$.
%The vector ${\bf a}(\theta_i)$ is the array response vector where $\theta_i \in [0,2\pi]$ is the azimuth angles of arrival.
Then, adopting the virtual channel representation~\cite{sayeed2002deconstructing,mendez2016hybrid,heath2016overview}, the channel vector ${\bf h}_k$ in \eqref{eq:channel_geo} can be modeled as
% EQUATION
\begin{align} 
\nonumber
{\bf h}_k &= {\bf A}{ \tilde{ \bf h}}_{{\rm b},k}=\sum_{i = 1}^{N_r}{ \tilde h}_{{\rm b},(i,k)}\,{\bf a}(\theta_i)
\end{align}
where $\tilde{\bf h}_{{\rm b},k} \in \mathbb{C}^{N_r}$ is the beamspace channel of the $k$th user, i.e., $\tilde{\bf h}_{{\rm b},k}$ has $L_k$ nonzero elements that contain the complex gains $\sim \mathcal{CN}(0,1)$ and the large-scale fading gain $\sqrt{\gamma_k}$.
We denote the beamspace channel matrix as $\tilde{\bf H}_{\rm b} = [\tilde{\bf h}_{{\rm b},1},\cdots,  \tilde{\bf h}_{{\rm b},N_u}]$ and it can be decomposed into $\tilde{\bf  H}_{\rm b} = \tilde{\bf  G}{\bf D}_\gamma^{1/2}$ where $\tilde{\bf  G} \in \mathbb{C}^{N_r \times N_u} $ is the sparse matrix of complex path gains and ${\bf D}_\gamma = {\rm diag}(\gamma_1,\cdots,\gamma_{N_u})$.
Accordingly, the beamspace channel of the $k$th user is expressed as $ \tilde{\bf h}_{{\rm b},k} = \sqrt{\gamma_k}\tilde {\bf g}_k$.
Finally, the channel matrix $\bf H$ is expressed as
\begin{align}
 \label{eq:channel}
{\bf H} =  {\bf A} \tilde {\bf H}_{\rm b} = {\bf A}\tilde {\bf G }{\bf D}_\gamma^{1/2}.
\end{align}
We assume that the analog beamformer is composed of the array response vectors corresponding to the $N_{\rm RF}$ largest channel eigenmodes \cite{el2012capacity}, i.e., ${\bf W_{\rm RF}} = {\bf A}_{\rm RF}$
% \in \mathbb{C}^{N_r \times N_{\rm RF}}$ 
where ${\bf  A}_{\rm RF}$ is a ${N_r \times N_{\rm RF}}$ sub-matrix of $\bf A$.
% (the matrix $\bf \tilde A$ consists of $N_{\rm RF}$ columns of $\bf A$).
We further assume that the array response vectors in ${\bf A}_{\rm RF}$ capture all channel propagation paths from $N_u$ users \cite{kim2015virtual}.
 % $N_{\rm RF} $ is considered to be large enough to capture all non-zero channel path gains from $N_u$ users.
%Since AoA is the long-term statistics of a channel, 
%Consequently, $\bf A$ includes ${\bf a}(\theta_i)$ for all $L$ dominant AoAs of $N_u$ users. 
%through AoA search training at the BS.
%This analog beamforming allows the BS to see the beamspace received signals at $N_{\rm RF}$ arrival angles. 
Then, the received signal after the analog beamforming in \eqref{eq:rx_rf_signal} reduces to
% EQUATION
\begin{align} 
\label{eq:y}
{\bf {y}}  = \sqrt{p_u}{\bf  A}_{\rm RF}^H{\bf Hs} +  {\bf A}_{\rm RF}^H \tilde{\bf n} =\sqrt{p_u} {\bf  H_{\rm b}}{\bf s} + {\bf n}
\end{align}
where $  {\bf n} = {\bf A}_{\rm RF}^H\tilde{\bf n} \sim \mathcal{CN}(\mathbf{0}, \mathbf{I}_{N_{\rm RF}})$ as $\bf A$ is unitary. 
Note that ${\bf H}_{\rm b}$
% \in \mathbb{C}^{N_{\rm RF} \times N_u}$
is the ${N_{\rm RF} \times N_u}$ sub-matrix of the beamspace channel matrix $\tilde {\bf H}_{\rm b}$ and contains $\sum_{k=1}^{N_u}L_k$ propagation path gains: 
\begin{align}
\label{eq:RFchannel}
%\nonumber
{\bf H_{\rm b} = GD}_\gamma^{1/2}
\end{align}
where ${\bf G }$\
% in \mathbb{C}^{N_{\rm RF} \times N_u} $ 
is the ${N_{\rm RF} \times N_u}$ sub-matrix of the complex gain matrix $\tilde{\bf G}$, corresponding to ${\bf A}_{\rm RF}$.

%, which contains complex gain vector of the $i$th user to the $N$ antennas at the BS, 
%Note that ${{\bf h}_{\rm b}}_i$ has $p$ channel gains corresponding to $L$ scatterings and the other $(N-L)$ channel gains are assumed to be much smaller than the $p$ channel gains.
%Under the ULA assumption with uniformly spaced spatial angles $\vartheta_i = i/N$, $\bf A$ becomes the Fourier transform matrix~\cite{heath2016overview}; ${\bf a}(\theta_i) = \frac{1}{\sqrt{N}}\left[1,e^{-j2\pi\vartheta_i},\cdots,,e^{-j2\pi\vartheta_i(N-1)}\right]^\intercal$ with $\theta_i = \arcsin(\frac{\lambda\vartheta_i}{d})$ where $\lambda$ is the signal wavelength and $d$ is the distance between antenna elements.

%%%%%%%%%%%%%%%%%%%%%%%%%%%
\subsection{Quantization Model}
\label{subsec:quantization}
%%%%%%%%%%%%%%%%%%%%%%%%%%%

%In this paper, instead of uniform bit ADCs where the quantization bit of the $i$th ADC $b_i = b$, $i = 1, \cdots, N_{\rm RF}$, 
We consider that each of the $i$th ADC pair has $b_i$ quantization bits and adopt the AQNM~\cite{fletcher2007robust, orhan2015low} as the quantization model to obtain a linearized quantization expression. 
The AQNM is accurate enough in low and medium SNR ranges \cite{orhan2015low}.
%, which corresponds to the mmWave communication operating SNR regime.  
After quantizing $\mathbf{y}$, we have the quantized signal vector
% EQUATION
\begin{align} 
\nonumber
\mathbf{y}_{\rm q}&=\mathcal{ Q}(\mathbf{ y}) = \mathbf{D}_{\alpha} \,\mathbf{ y}+ \mathbf{n}_{\rm q}
\\ \label{eq:AQNM2}
& = \sqrt{p_u}{\bf D_\alpha  H_{\rm b}}{\bf s + D_\alpha} {\bf n} + {\bf n}_{\rm q}
% = \sqrt{p_u}\mathbf{W}_{\alpha}\mathbf{F}^H_{\rm RF}\,\mathbf{Hs} + \mathbf{W}_{\alpha} \mathbf{F}^H_{\rm RF}\,\mathbf{n} + \mathbf{n}_{\rm q}
\end{align} 
where $\mathcal{Q}(\cdot)$ is an element-wise quantizer function separately applied to the real and imaginary parts and $\mathbf{D}_\alpha$ is a diagonal matrix with quantization gains $\mathbf{D}_\alpha =  {\rm diag}(\alpha_1,\cdots, \alpha_{N_{\rm RF}})$.
The quantization gain $\alpha_i$ is a function of the number of quantization bits and defined as $\alpha_i = 1- \beta_i$ where $\beta_i$ is a normalized quantization error. 
Assuming the non-linear scalar MMSE quantizer and Gaussian transmit symbols, it can be approximated for $b_i > 5$ as $\beta_i = \frac{\mathbb{E}[|{y}_i - {y}_{{\rm q}i}|^2]}{\mathbb{E}[|{y}_i|^2]}  \approx \frac{\pi\sqrt{3}}{2} 2^{-2b_i}$.
%\begin{align}
%    \nonumber
%    \beta_i = \frac{\mathbb{E}[|{y}_i - {y}_{{\rm q}i}|^2]}{\mathbb{E}[|{y}_i|^2]}  \approx \frac{\pi\sqrt{3}}{2} 2^{-2b_i}
%\end{align}
%The value of the quantization gain $\alpha_i$ depends on the number of quantization bits and $\alpha_i = 1- \beta_i$ where  $\beta_i$ is defined as $\beta_i = \frac{\mathbb{E}[|{y}_i - {y}_{{\rm q}i}|^2]}{\mathbb{E}[|{y}_i|^2]} $.
% where $ {y}_{{\rm q},i}$ is the quantized output for ${y}_i$.
The values of $\beta_i$ are listed in Table \ref{tb:beta} for $b_i \leq 5$.
% , and it is approximated to $\beta_i = \frac{\pi\sqrt{3}}{2} 2^{-2b_i}$ for $b_i > 5$. 
Note that $b_i$ is the number of quantization bits for each real and imaginary part of $ y_i$.
The quantization noise $\mathbf{n}_{\rm q}$ is an additive noise which is uncorrelated with $\bf  y$ and follows the complex Gaussian distribution with zero mean.
%\cite{fan2015uplink,zhang2016spectral}.
For a fixed channel realization $ {\bf H}_{\rm b}$, the covariance matrix of $\bf n_{\rm q}$ is
% EQUATION
\begin{align}
%\label{eq:cov2}
\nonumber
\mathbf{R}_{\mathbf{n}_{\rm q}\mathbf{n}_{\rm q}}= {\bf D}_\alpha {\bf D}_\beta \,{\rm diag}(p_u{\bf H_{\rm b}}{\bf H}_{\rm b}^H + {\mathbf{I}_{N_{\rm RF}}})
\end{align}
where $\mathbf{D}_\beta =  {\rm diag}(\beta_1,\cdots, \beta_{N_{\rm RF}})$.

Assuming sampling at the Nyquist rate, the ADC power consumption is modeled as \cite{orhan2015low} 
\begin{align}
\label{eq:ADCpower}
P_{\rm ADC}(b)= c\,f_s\,2^{b}
\end{align}
where $c$ is the energy consumption per conversion step (conv-step), called Walden's figure-of-merit, $f_s$ is the sampling rate and $b$ is the number of quantization bits.
This model illustrates that the ADC power consumption scales exponentially in the number of quantization bits $b$.
% considering the non-linear scalar MMSE quantizer. 

% TABLE
\begin{table}[!t]
\centering
\caption{The Values of $\beta$ for Different Quantization Bits $b$ }\label{tb:beta}
\begin{tabular}{ l c c c c c }
  \thickhline
$ b$  & 1 & 2 & 3 & 4 & 5\\
  \hline
 $\beta$   & 0.3634 & 0.1175 & 0.03454 & 0.009497 & 0.002499 \\
  \thickhline
\end{tabular}
\end{table}

%%%%%%%%%%%%%%%%%%%%%%%%%%%%
\section{ADC Bit Allocation Algorithms}
\label{sec:main2}
%%%%%%%%%%%%%%%%%%%%%%%%%%%%

In this section, we propose BA algorithms to improve the performance of the proposed system by leveraging the flexibility of ADC resolutions.
{\color{black} Note that we assume perfect knowledge of the channel state information (CSI) at the BS. 
The rationale behind this is that efficient algorithms have been proposed for mmWave channel estimation \cite{mo2014channel, mo2016channel,alkhateeb2014channel,lee2014exploiting,gao2016channel} by exploiting the sparse nature of mmWave channels. 
In the hybrid receiver structure with $N_{\rm RF} < N_r$, state-of-the-art mmWave channel estimators 
such as bisectional approach \cite{alkhateeb2014channel}, modified OMP \cite{lee2014exploiting}, and distributed grid message passing \cite{gao2016channel} validated the estimation performance.
Assuming the use of high-resolution ADCs for a channel estimation phase, such estimation algorithms can be adopted in the considered system.}
%The BA algorithms minimize the total quantization error depending on channel gains under constrained total ADC power consumption.
%We first propose a BA algorithm with the proper quantization modeling shown in \eqref{eq:AQNM2}, and in order to improve the performance of the BA algorithm, we revised the MMSQE-BA algorithm by applying the linear quantization modeling only to the desired signal term ${\bf x}$, which may be inaccurate since the actual quantization involves the noise term in \eqref{eq:y}. 
%Simulation results, however, validate the revised BA algorithm

%%%%%%%%%%%%%%%%%%%%%%%%%%%%%%%%
\subsection{MMSQE Bit Allocation}
\label{subsec:BA_Power}
%%%%%%%%%%%%%%%%%%%%%%%%%%%%%%%%

We adopt the MSQE $\mathcal{E}(b)  = \mathbb{E}\big[| y_- {y_{{\rm q}}}|^2\big]$ for  $\bf y$ in \eqref{eq:y} as a distortion measure.
% it is expected that there exists an optimal allocation of quantization bits. %as the quantization error from the signal with a small gain causes minor effect. 
%Assuming the non-linear scalar MMSE quantizer, 
%Recall that $\beta_i = \frac{\mathbb{E}[|{y}_i - {y}_{{\rm q}i}|^2]}{\mathbb{E}[|{y}_i|^2]} $, then the MSQE of $y_i$ with $b_i$ quantization bits is expressed as $\mathcal{E}_{i}(b_i) = \sigma_{y_i}^2 \beta_i$ where $\sigma_{y_i}^2 = \mathbb{E}\big[|y_i|^2\big]$. 
Assuming the MMSE quantizer and Gaussian transmit symbols, the MSQE of $y_i$ with $b_i$ quantization bits for $b_i > 5 $ is modeled as \cite{orhan2015low}
% EQUATION
\begin{align}
\label{eq:msqe}
\mathcal{E}_{y_i}(b_i) 
& = \frac{\pi\sqrt{3}}{2}\sigma_{y_i}^2\,2^{-2b_i}
\end{align}
where $\sigma^2_{y_i} = p_u\|[{\bf H}_{\rm b}]_{i,:}\|^2+1$.
Using \eqref{eq:msqe} for any quantization bits,\footnote{{\color{black} 
Although \eqref{eq:msqe} holds for $b_i > 5$, it can be validated by the performance of our algorithms that \eqref{eq:msqe} can provide a good approximation when formulating optimization problem even for a small number of quantization bits.}}
we formulate the MMSQE problem through some relaxations.
Then, the solution of the MMSQE problem minimizes the total quantization error by adapting quantization bits under constrained total ADC power consumption.
%In this section, we solve the BA problem by minimizing the total MSQE subject to the power constraint.
%Although \eqref{eq:msqe} holds for $b_i > 5$, 

To avoid integer programming, we relax the integer variables ${\bf b} \in \mathbb{Z}^{N_{\rm RF}}_+ $ to the real numbers ${\bf b} \in \mathbb{R}^{N_{\rm RF}} $ to find a closed-form solution.
We also consider \eqref{eq:msqe} to hold for $b_i \in \mathbb{R}$.
Despite the fact that the ADC power consumption with $b$ bits $P_{\rm ADC}(b) = 0 \text{ for } b \leq 0$, we assume $P_{\rm ADC}(b) = cf_s2^b$ in \eqref{eq:ADCpower} to hold for $b \in \mathbb{R}$.
Under the constraint of the total ADC power of the conventional fixed-ADC system in which all $N_{\rm RF}$ ADCs are equipped with $\bar b$ bits, the relaxed MMSQE problem is formulated as
% EQUATION
\begin{gather}
\label{eq:opt_power}
\hat {\mathbf{{b}}}^{\rm } 
= \argmin_{\mathbf{b}=[b_1,\cdots,b_{N_{\rm RF}}]^\intercal} \sum_{i=1}^{N_{\rm RF}}\mathcal{E}_{y_i}(b_i)
\\ \nonumber
\text{s.t.} \quad \sum_{i=1}^{N_{\rm RF}} P_{\rm ADC}(b_i) \leq N_{\rm RF}P_{\rm ADC}(\bar b),\ {\bf b}\in\mathbb{R}^{N_{\rm RF}}.
\end{gather}
%where $\bar b$ is the number of quantization bits for the fixed-ADC system.
{\color{black} 
Here, $\bar {b}$ is the number of ADC bits for a fixed-ADC system, which we use to give a reference total ADC power in the constraint for the above MMSQE optimization problem.}
Proposition \ref{thm:BA_power} provides the MMSQE-BA solution in a closed form by solving the Karush-Kuhn-Tucker (KKT) conditions for \eqref{eq:opt_power}, which is different from the previously proposed greedy BA approach under a bit constraint in \cite{choi2016space}.
%Let $h_i = \|[{\bf H}_{\rm b}]_{i,:}\|^2$ be the aggregated beamspace channel gains for the $i$th RF chain.  
% Propositionclosed-form
\begin{proposition} 
	\label{thm:BA_power}
	For the relaxed MMSQE problem in \eqref{eq:opt_power}, the optimal number of quantization bits which minimizes the total MSQE is derived as 
	%Under the relaxed conditions: real number quantization bits $b_i \in \mathbb{R}$, the ADC power consumption \eqref{eq:ADCpower} and MSQE \eqref{eq:msqe} for any $b_i$, the optimal number of quantization bits which minimizes the total MSQE for~\eqref{eq:AQNM2} is 
	
	\vspace{-1em}
	\small
	\begin{align}
		\label{eq:opt_BA}
		%\hat  b^{\rm}_i = \bar b - \log_2\left(\frac{1}{N_{\rm RF}}\sum_{j = 1}^{N_{\rm RF}}\left(\frac{1+{\rm SNR}^{\rm rf}_j}{1+{\rm SNR}^{\rm rf}_i}\right)^{\frac{1}{3}}\right),\, i = 1,\cdots,N_{\rm RF}\\
		\hat  b^{\rm}_i = \bar b + \log_2\left(\frac{N_{\rm RF}\big(1+{\rm SNR}^{\rm rf}_i\big)^{\frac{1}{3}}}{\sum_{j = 1}^{N_{\rm RF}} \big(1+{\rm SNR}^{\rm rf}_j\big)^{\frac{1}{3}}}\right),\  i = 1,\cdots,N_{\rm RF}
	\end{align}
	\normalsize
	where ${\rm SNR}_i^{\rm rf} = p_u\|[{\bf H}_{\rm b}]_{i,:}\|^2$.
	%$[{\bf G}]_{j,:}[{\bf G}]_{j,:}^H$.
\end{proposition}
\begin{proof}
	See Appendix~\ref{appx:BA_power}.
\end{proof}

In Proposition~\ref{thm:BA_power}, ${\rm SNR}_i^{\rm rf}$ indicates the SNR of the $i$th received signal after analog beamforming $y_i$, which illustrates that the MMSQE-BA \eqref{eq:opt_BA} depends on the channel gain of $\bf y$. 
We remark that the MMSQE-BA has the power of $1/3$ which comes from the relationship between the MSQE $\mathcal{E}_{y_i}(b_i)$ and the ADC power $P_{\rm ADC}(b_i)$ in terms of $b_i$.
Proposition~\ref{thm:BA_power} indicates that the optimal number of the $i$th ADC bits $\hat b_i$ increases logarithmically with $\big(1+\text{SNR}_i^{\rm RF}\big)^{1/3}$ and decreases logarithmically with the sum of  $\big(1+\text{SNR}_j^{\rm RF}\big)^{1/3}$ for $j = 1,\cdots, N_{\rm RF}$.
Accordingly, the ADC pair with the relatively larger aggregated channel gain $\|[{\bf H}_{\rm b}]_{i,:}\|^2$ needs to have more quantization bits to minimize the total quantization distortion. 
Note that since the slowly changing channel characteristics such as large-scale fading and AoA mostly determines the channel gains and sparsity, they are the dominant factors for the BA solution in Proposition \ref{thm:BA_power}.
%Accordingly, performing BA based on \eqref{eq:opt_BA} at the 
%In~\eqref{eq:opt_BA}, the log term can be negative so that $ \hat b^{\rm}_i > \bar b_i$, and this occurs when ${\rm SNR}^{\rm rf}_i$ is large.

% ALGORITHM
\begin{algorithm}[!b]
\caption{MMSQE-BA Algorithm}
\label{Power}
\begin{enumerate}
\item Set power constraint $P_{\rm max} = N_{\rm RF} P_{\rm ADC}( \bar{b})$ using~\eqref{eq:ADCpower}
\item Set $\mathbb S = \{1 \ldots N_{\rm RF}\}$ and $P_{\rm total} = 0$
%\item Zero out $N$-element vectors ${\bf b}$, ${\bf p}$ and ${\bf T}$
\item {\bf for} $i = 1 \ldots N_{\rm RF} $
  \begin{enumerate}
  \item Compute ${\hat b}_i$ using~\eqref{eq:opt_BA} and $b_i = \max(0, \lceil {\hat b}_i \rceil)$
  %\item $b_i = \max(0, \lceil \hat{b}_i \rceil)$
  \item {\bf if} ($b_i = 0$), $\mathbb S = \mathbb S - \{i\}$
  \item {\bf else} $p_i = P_{\rm ADC}( b_i )$ and $P_{\rm total} = P_{\rm total} + p_i$ 
\begin{enumerate}
\item [$\circ$]{\bf if} ($\hat b_i \in \mathbb Z$), $\mathbb S = \mathbb S - \{i\}$
\end{enumerate}
  \end{enumerate}
\item {\bf if} $P_{\rm total} \le P_{\rm max}$, {\bf  return b}
\item {\bf for} $i \in \mathbb S$, compute $T_i = T(i)$ using~\eqref{eq:tradeoff}
\item {\bf while} $P_{\rm total} > P_{\rm max}$ \label{while}
  \begin{enumerate}
  \item $i^{*} = {\rm argmin}_{i \in \mathbb S} T_i$
  \item $b_{i^{*}} = {b}_{i^{*}} - 1$ and $\mathbb S = \mathbb S - \{i^{*}\}$
 % \item $\mathbb S = \mathbb S - \{i^{*}\}$
  \item $P_{\rm total} = P_{\rm total} - p_{i^*} + P_{\rm ADC}(b_{i^*})$
  \end{enumerate}
\item {\bf return b}
\end{enumerate}
\end{algorithm}
%%%%%%%%%%%%%%%%%%%%%%%%%%%%%%%%%%%%%%%%%%%%%%%%

Since $\hat b_i$ in \eqref{eq:opt_BA} is a real number solution, we need to map it into non-negative integers.
Although the nearest integer mapping can be applied to the solution, it ignores the tradeoff between power consumption and quantization error and can violate the power constraint after the mapping. 
As an alternative, we propose a tradeoff mapping that is power efficient.
First, the negative quantization bits ($\hat b^{\rm}_i < 0$) are mapped to zero, i.e., the ADC pairs with $\hat b_i \leq 0$ are deactivated. 
Note that this mapping does not violate the actual power constraint as $P_{\rm ADC}(b) = 0 \text{ for } b \leq 0$.
Next, we map positive non-integer quantization bits ($\hat b_i >0, \hat b_i \notin \mathbb{Z}$) to $\lceil  \hat b^{\rm }_i \rceil$. 
If the power constraint is violated, i.e., $\sum_{i\in\mathbb{S}_+}P_{\rm ADC}(\lceil \hat b_i \rceil) > N_{\rm RF} P_{\rm ADC}(\bar b)$ where $\mathbb{S}_+ = \{ i\, |\, \hat b_i >0 \}$, we need to map the subset of the positive non-integer quantization bits to $\lfloor \hat b^{\rm}_i \rfloor$ instead of $\lceil  \hat b^{\rm }_i \rceil$.
Notice that the $\lfloor \hat b^{\rm}_i \rfloor$-mapping reduces the power consumption while increasing the MSQE.
In this regard, we need to find the best subset to perform power-efficient $\lfloor \hat b^{\rm}_i \rfloor$-mapping.

%Due to nonlinearity of the quantization error and ADC power with respect to the quantization bits $b$, 
To determine the best subset of the positive non-integer quantization bits for $\lceil \hat b_i \rceil$, we propose a tradeoff function %to achieve efficient  $\lfloor \hat b^{\rm}_i \rfloor$-mapping for $ \hat b_i >0,\ \hat b_i \notin \mathbb{Z}$:
%EQUATION

\vspace{-1em}
\small
 \begin{align}
	 \label{eq:tradeoff}
	 %= \left|\frac{\Delta {\mathcal{E}}_{i}(\hat b^{\rm }_i)}{\Delta P_{\rm ADC}(\hat b^{\rm}_i)}\right|
	 T(i) =\left|\frac{\mathcal{E}_{i}(\hat b^{\rm}_i) -  \mathcal{E}_{i}(\lfloor \hat b^{\rm }_i \rfloor) }{P_{\rm ADC}( \hat b^{\rm }_i)-P_{\rm ADC}(\lfloor \hat  b^{\rm}_i \rfloor)}\right| 
	\to \frac{2^{-2\lfloor \hat b^{\rm}_i\rfloor}-2^{-2\hat b^{\rm }_i}}{2^{ \hat b^{\rm}_i}-2^{\lfloor \hat b^{\rm }_i\rfloor}}\sigma^2_{ y_i}.
 \end{align}
\normalsize
The proposed function in \eqref{eq:tradeoff} represents the MSQE increase per unit power savings after mapping $\hat b^{\rm}_i$ to $\lfloor \hat b^{\rm}_i \rfloor$.
For the $\lfloor \hat b^{\rm}_i \rfloor$-mapping, $ \hat b^{\rm}_i$ with the smallest $T(i)$ is re-mapped to $\lfloor \hat b_i \rfloor$ from $\lceil \hat b_i \rceil$ to achieve the best tradeoff of quantization error vs. power consumption.
This repeats for $\hat b^{\rm}_i$ with the next smallest $T(i)$ until the power constraint is satisfied.
%For implementation, we use a relative tradeoff function obtained from~\eqref{eq:tradeoff}
%EQUATION
%\begin{align}
%\label{eq:Trel}
%T_{\rm rel}(i) = \frac{2^{-2\lfloor \hat b^{\rm}_i\rfloor}-2^{-2\hat b^{\rm }_i}}{2^{ \hat b^{\rm}_i}-2^{\lfloor \hat b^{\rm }_i\rfloor}}\sigma^2_{ y_i}.
%\end{align}
Algorithm~\ref{Power} shows the proposed MMSQE-BA algorithm. 
%Although the while-loop at line~\ref{while} occurs $\mathcal O(N_{\rm RF})$ times, %the computational complexity is $\mathcal O(N_{\rm RF}^2)$ including the effort of sorting $T_{\rm rel}(i)$.
%this can be easily reduced to $\mathcal O(\log_2 N_{\rm RF})$ by using binary search.
The while-loop at line~\ref{while} will always end as this mapping algorithm can always satisfy the power constraint from the following reasons: (i) for $ \hat b^{\rm}_i <0$, the 0-bit mapping does not increase power, and (ii) for $\hat b_i >0$, the total ADC power consumption always becomes $\sum_{i\in \mathbb{S}_+}P_{\rm ADC}(\lfloor  \hat b^{\rm}_i \rfloor) \leq \sum_{i\in \mathbb{S}_+}P_{\rm ADC}( \hat b^{\rm }_i)$.

Note that the constant term in $1+ {\rm SNR}_i^{\rm rf}$ of \eqref{eq:opt_BA} comes from the additive noise $ \bf n$ in \eqref{eq:y}.
Due to this noise term, the MMSQE-BA $\hat b_i$ would be almost the same for all ADCs when the transmit power $p_u$ is small.
In other words, in the low SNR regime, the noise term in $\hat b_i$ becomes dominant $( 1\gg { \rm SNR}_i^{\rm rf},\ i = 1, \cdots, N_{\rm RF})$.
This leads to $\hat b_i \approx \bar b$ for $i= 1,\cdots, N_{\rm RF}$.
The intuition behind this is that since we minimize the total MSQE of $\bf y$, which always includes the noise, the MMSQE-BA $\hat b_i$ minimizes mostly the quantization error of the noise in the low SNR regime, not the desired signal. 
Consequently, uniform bit allocation ($\hat b_i = \bar b$) across all the ADCs is likely to appear in the low SNR regime. 
In this perspective, the MMSQE-BA becomes more effective as the SNR increases while providing similar performance as fixed-ADCs in the low SNR regime.
In Section \ref{subsec:BA_Power_Modified}, we revise the MMSQE-BA to overcome such noise-dependency.
% so that it becomes robust to noise.

%%%%%%%%%%%%%%%%%%%%%%%%%%%%%%%
\subsection{Revised MMSQE Bit Allocation}
\label{subsec:BA_Power_Modified}
%%%%%%%%%%%%%%%%%%%%%%%%%%%%%%%

The MMSQE-BA \eqref{eq:opt_BA} is dependent to the additive noise as it minimizes the quantization error of $y_i$, not solely the desired signal. Accordingly, the MMSQE-BA is less effective in the low SNR regime. 
To address this problem, we modify the previous MMSQE problem \eqref{eq:opt_power} by considering to minimize the quantization error of only the desired signal. 
% by taking advantage of the CSIR.
We ignore the additive noise ${\bf n}$ in $\bf y$ and consider the quantization of the desired signal ${\bf x} =\sqrt{p_u} {\bf  H_{\rm b}}\bf s$ at the ADCs. 
According to the AQNM, we can model the quantization of ${\bf x}$ as
\begin{align}
\nonumber
{\bf x}_{\rm q} =  \sqrt{p_u}{\bf D}_{\alpha}{\bf H}_{\rm b}{\bf s} + \hat {\bf n}_{\rm q}
\end{align}
where $ \hat {\bf n}_{\rm q}$ is the additive quantization noise uncorrelated with $\bf x_{\rm q}$.
% and its covariance becomes
%\begin{align}
%\label{eq:cov_rev}
%\mathbf{R}_{\tilde {\mathbf{n}}_{\rm q}}= {\bf W}_\alpha {\bf W}_\beta \,{\rm diag}(p_u{\bf H_{\rm b}}{\bf H}_{\rm b}^H).
%\end{align}
The corresponding MSQE for the $i$th signal $x_i$ becomes
\begin{align}
\label{eq:msqe_rev}
{\mathcal{E}}_{x_i}(b_i) 
&= \mathbb{E}\Big[|{x}_i - {x}_{{\rm q}i}|^2\Big]= \frac{\pi\sqrt{3}}{2}\sigma_{x_i}^2\,2^{-2b_i}
\end{align}
where $\sigma^2_{x_i} = p_u\|[{\bf H}_{\rm b}]_{i,:}\|^2$. 
Using \eqref{eq:msqe_rev}, we formulate the revised MMSQE problem as
% EQUATION
\begin{gather}
\label{eq:opt_power_rev}
\hat {\mathbf{{b}}}^{rev } 
= \argmin_{\mathbf{b}=[b_1,\cdots,b_{N_{\rm RF}}]^\intercal} \sum_{i=1}^{N_{\rm RF}}\mathcal{E}_{x_i}(b_i)
\\ \nonumber
\text{s.t.} \quad \sum_{i=1}^{N_{\rm RF}} P_{\rm ADC}(b_i) \leq N_{\rm RF}P_{\rm ADC}(\bar b),\ {\bf b}\in\mathbb{R}^{N_{\rm RF}}.
\end{gather}
{\color {black} 
Note that while the MMSQE-BA algorithm in Section \ref{subsec:BA_Power} is developed with the proper AQNM quantization modeling \eqref{eq:AQNM2}, 
%The algorithm minimizes the quantization error of received signals under constrained total ADC power consumption.
the revised MMSQE-BA (revMMSQE-BA) algorithm will be developed based on the quantization modeling only for the desired signal term in \eqref{eq:y}.
Consequently, this modeling approach may be inaccurate since the actual quantization process involves noise.
Adopting the GMI which serves a lower bound on the channel capacity \cite{zhang2012general, zhang2016remark}, however, we show that \eqref{eq:opt_power_rev} is equivalent to maximizing the GMI in the low SNR regime.
Under the assumptions of IID Gaussian signaling $s_i \sim \mathcal{CN}(0, 1)$ and applying a linear combiner ${\bf W}$ to the quantized signal ${\bf y}_{\rm q}$ with nearest-neighbor decoding, the GMI of user $n$ \cite{liang2016mixed} is expressed as 
\begin{align}
   	\label{eq:GMI}
    	I_{n}^{\rm GMI}({\bf w}_n, {\bf b}) = \log_2\bigg(1+\frac{\kappa({\bf w}_n,\ {\bf b})}	{1-\kappa({\bf w}_n,\ {\bf b})}\bigg)
\end{align}
where
\begin{align}
	\label{eq:kappa}
	\kappa({\bf w}_n,\ {\bf b}) = \frac{\big|\mathbb{E}[{\bf w}^H_n{\bf y}_{\rm q} \sqrt{p_u}s_n]\big|^2}{p_u\mathbb{E}[|{\bf w}_n^H{\bf y}_{\rm q}|^2]} = \frac{{\bf w}_n^H {\bf R}_{{\bf y}_{\rm q}s_n} {\bf R}^H_{{\bf y}_{\rm q}s_n}{\bf w}_n}{{\bf w}_n^H {\bf R}_{{\bf y}_{\rm q}{\bf y}_{\rm q}}{\bf w}_n}
\end{align}

\begin{proposition}
    \label{pr:maxGMI}
    Using the IID Gaussian signaling and linear combiner ${\bf W}$ to the quantized signal ${\bf y}_{\rm q}$ with nearest-neighbor decoding, the revised MMSQE problem \eqref{eq:opt_power_rev} is equivalent to \eqref{eq:maxGMI} in the low SNR regime.
\begin{gather}
        \label{eq:maxGMI}
        \hat {\mathbf{{b}}}^{GMI} 
        = \argmax_{{\bf w}_n,\ {\bf b}} \sum_{n=1}^{N_{u}} I_{n}^{\rm GMI}({\bf w}_n, {\bf b}) 
        \\ \nonumber
        \text{s.t.} \quad \sum_{i=1}^{N_{\rm RF}} P_{\rm ADC}(b_i) \leq N_{\rm RF}P_{\rm ADC}(\bar b), \ {\bf b} \in\mathbb{R}^{N_{\rm RF}}.
    \end{gather}
\end{proposition}
\begin{proof}
See Appendix \ref{appx:maxGMI}.
\end{proof}
}
Now, we solve \eqref{eq:opt_power_rev} and derive the revMMSQE-BA solution $\hat {\bf b}^{rev}$ in the following proposition.
% PROPOSITION
\begin{proposition} 
	\label{thm:BA_power_rev}
	%Under the relaxed conditions: real number quantization bits $b_i \in \mathbb{R}$, the ADC power consumption \eqref{eq:ADCpower} and MSQE \eqref{eq:msqe} for any $b_i$, 
	Assuming $\|[{\bf H}_{\rm b}]_{i,:}\| \neq 0$ for $i = 1,\cdots,N_{\rm RF}$, the optimal number of quantization bits which minimizes the total MSQE of desired signals $\bf x$ for the revised MMSQE problem \eqref{eq:opt_power_rev} is 
	% for the relaxed MMSQE problem with \eqref{eq:msqe_rev} is 
	%Under relaxed quantization bits $b_i \in \mathbb{R}$ and relaxed ADC power consumption model $P_{\rm ADC}(i)= cW2^{b_i}$, $i =  1,\cdots, N_r$, the global optimal number of quantization bits which minimizes the total MSQE of desired signals is 
	
	\vspace{-1em}
	\small
	\begin{align}
		\label{eq:opt_BA_rev}
		%  \hat b^{\rm rev}_i = \bar b - \log_2\left(\frac{1}{N_{\rm RF}}\sum_{j = 1}^{N_{\rm RF}}\frac{\|[{\bf H}_{\rm b}]_{j,:}\|^{\frac{2}{3}}}{\|[{\bf 		H}_{\rm b}]_{i,:}\|^{\frac{2}{3}}}\right), \quad i = 1,\cdots, N_{\rm RF}.
		\hat b^{rev}_i = \bar b + \log_2\left(\frac{{N_{\rm RF}}\|[{\bf H}_{\rm b}]_{i,:}\|^{\frac{2}{3}}}{\sum_{j = 1}^{N_{\rm RF}}\|[{\bf H}_{\rm b}]_{j,:}\|^{\frac{2}{3}}}\right), \  i = 1,\cdots, N_{\rm RF}.
	\end{align}
	\normalsize
\end{proposition}

\begin{proof}
Replacing $\sigma^2_{y_i}$ with $ \sigma^2_{x_i}$ ($c_i =\sigma^2_{x_i}$) in \eqref{eq:trans_problem} and following the same steps in the proof of Proposition \ref{thm:BA_power} in Appendix \ref{appx:BA_power}, we obtain \eqref{eq:BA_mid}.
Then, \eqref{eq:opt_BA_rev} is obtained by putting $z_i = 2^{-2b_i}$, $\bar z = 2^{-2\bar b}$ and $c_i = \sigma^2_{x_i}$ into \eqref{eq:BA_mid}.
\end{proof}

{\color{black}
\begin{corollary}
    The revMMSQE-BA solution $\hat {\bf b}^{rev}$ maximizes the GMI in the low SNR regime and minimizes the quantization error of the beam-domain received signal ${\bf y}$ in the high SNR. 
\end{corollary}
\begin{proof}
    When the SNR is low, Proposition \ref{pr:maxGMI} holds. 
    For the high SNR, the MMSQE-BA solution reduces to the revMMSQE-BA solution, $\hat{\bf b} \to \hat{\bf b}^{rev}$, as ${\rm SNR}^{rf}_i \gg 1$. 
\end{proof}
}
%The revMMSQE-BA \eqref{eq:opt_BA_rev} in Proposition \ref{thm:BA_power_rev} does not include any noise-related term so that it performs similarly to the MMSQE-BA in Proposition \ref{thm:BA_power} operating in the high SNR regime (${\rm SNR}^{\rm rf} \gg 1$).
%In this regard, the revMMSQE-BA provides noise-robust BA performance.
Accordingly, even in the low SNR regime, we can selectively assign ADC bits to maximize GMI, which can be considered as maximizing achievable rate.
In this regard, the revMMSQE-BA provides noise-robust BA performance.
%Moreover, since zero bit is likely to be allocated for the ADCs with negligibly small aggregated channel gains $\|[{\bf H}_{\rm b}]_{i,:}\|$ under a harsh power constraint, noise reduction can be achieved in the low SNR regime.
Similar non-negative integer mapping can be performed by replacing $\sigma^2_{y_i}$ in \eqref{eq:tradeoff} with $\sigma^2_{x_i}$.
%Similar non-negative integer mapping can be performed by replacing $T_{\rm rel}(i)$ in \eqref{eq:Trel} with
%\begin{align}
%\label{eq:T_rev}
%\tilde {T}(i) = \left |\frac{\Delta{\tilde{\mathcal{E}}}(\hat b^{\rm rev}_i)}{\Delta P_{\rm ADC}(\hat b^{\rm rev}_i)}\right|,
%\end{align}
%and the relative trade-off function of $\tilde {T}(i)$ is given as  
%\begin{align}
%\label{eq:Trel_rev}
%\nonumber
%\tilde T_{\rm rel}(i) = \frac{2^{-2\lfloor \hat b^{\rm}_i\rfloor}-2^{-2\hat b^{\rm }_i}}{2^{ \hat b^{\rm}_i}-2^{\lfloor \hat b^{\rm }_i\rfloor}}\sigma^2_{x_i}.
%\end{align}

%In Section \ref{sec:Num}, we evaluate the proposed BA algorithms to show the performance improvement compared to the conventional fixed-ADC approach by considering two different switching periods: the channel coherence time and the time-scale of slowly changing channel characteristics.
% such as the large-scale fading and AoA.

{\color {black} 
%%%%%%%%%%%%%%%%%%%%%%%%%%%
\subsection{Capacity Anaylsis with Bit Allocation}
\label{subsec:capacity}
%%%%%%%%%%%%%%%%%%%%%%%%%%%

In this subsection, we analyze the capacity under the considered system model \eqref{eq:AQNM2} for given $({\bf b}, {\bf H}_{\rm b})$ when the SNR is low. %and (ii) ${\bf H}_{\rm b}{\bf H}_{\rm b}^H$ becomes a near diagonal matrix.
Let $\eta = {\bf D}_{\alpha} {\bf n} + {\bf n}_{\rm q}$, then the capacity is expressed as
\begin{align}
	\label{eq:capacity}
	C\big({\bf b}, {\bf H}_{\rm b}\big) = \log_2\bigg|{\bf I}_{N_{\rm RF}} + p_u {\bf R}_{\eta\eta}^{-1}{\bf D}_{\alpha}{\bf H}_{\rm b}{\bf H}_{\rm b}^H{\bf D}_{\alpha}^H \bigg|
\end{align}
\normalsize
where ${\bf R}_{\eta\eta} = {\bf D}_{\alpha}{\bf D}_{\alpha}^H + {\bf R}_{{\bf n}_{\rm q}{\bf n}_{\rm q}}$.

\begin{lemma}
	\label{lem:capacity_lowSNR}
	For a given ADC bit allocation ${\bf b}$, the capacity of \eqref{eq:AQNM2} in the low SNR regime is approximated as
	\begin{align}
		\label{eq:capacity_lowSNR}
		C_{low}\big({\bf b}, {\bf H}_{\rm b}\big) = \log_2\Bigg(1+\sum_{i = 1}^{N_{\rm RF}} \frac{p_u \alpha_i \|[{\bf H}_{\rm b}]_{i,:}\|^2}{1+p_u(1-\alpha_i)\|[{\bf H}_{\rm b}]_{i,:}\|^2}\Bigg).
	\end{align}
\end{lemma}
\begin{proof}
%Let ${\pmb \eta} = {\bf W}_{\alpha}{\bf n} + {\bf n}_{\rm q}$.
In the low SNR regime, the capacity \eqref{eq:capacity} can be approximated as

\vspace{-1 em}
\small
\begin{align}
	\nonumber
	&C\big({\bf b}, {\bf H}_{\rm b}\big)  \approx \log_2\bigg(1+ p_u {\rm tr}\bigg( {\bf R}_{\eta\eta}^{-1}{\bf D}_{\alpha}{\bf H}_{\rm b}{\bf H}_{\rm b}^H{\bf D}_{\alpha}^H \bigg)\bigg)
	\\ \nonumber
	& = \log_2\bigg(1+p_u {\rm tr}\bigg(\Big[{\bf I}_{N_{\rm RF}} + p_u{\bf D}_{\beta}{\rm diag}\big({\bf H}_{\rm b}{\bf H}_{\rm b}^H\big)\Big]^{-1}{\bf D}_{\alpha}{\bf H}_{\rm b}{\bf H}_{\rm b}^H\bigg)\bigg)
	\\ \nonumber
	& = \log_2\bigg(1 + p_u\sum_{i = 1}^{N_{\rm RF}}\Big(1+p_u\beta_i\|[{\bf H}_{\rm b}]_{i,:}\|^2\Big)^{-1}\alpha_i\|{\bf H}_{\rm b}]_{i,:}\|^2\bigg).
\end{align}
\normalsize
This completes the proof for Lemma \ref{lem:capacity_lowSNR}.
\end{proof}

%%%%%%%%%%%%%%%%%%%%%%%%%%%%%%%%%%%%%%%
\begin{figure*}[!t]
% 		\rule{\textwidth}{0.7pt}\\
 		\begin{align}
			\label{eq:capacity_lowSNR_BA}
			&\tilde C^{RBA}_{low}\big({\bf H}_{\rm b} \big) =\log_2\left(1+\sum_{i = 1}^{N_{\rm RF}} \frac{p_u\bigg(1-\pi\sqrt{3}\,2^{-(2{\bar b}+1)}\Big(N_{\rm RF}^{-2}\|[{\bf H}_{\rm b}]_{i,:}\|^{-\frac{4}{3}} \Big\{\sum_{j = 1}^{N_{\rm RF}}\|[{\bf H}_{\rm b}]_{j,:}\|^{\frac{2}{3}}\Big\}^2\Big)_{2^{2\bar b}}^-\bigg) \big\|[{\bf H}_{\rm b}]_{i,:}\big\|^2}{1+p_u\pi\sqrt{3}\,2^{-(2{\bar b}+1)}\Big(N_{\rm RF}^{-2}\|[{\bf H}_{\rm b}]_{i,:}\|^{-\frac{4}{3}}\Big\{ \sum_{j = 1}^{N_{\rm RF}}\|[{\bf H}_{\rm b}]_{j,:}\|^{\frac{2}{3}}\Big\}^2\Big)_{2^{2\bar b}}^- \,\big\|[{\bf H}_{\rm b}]_{i,:}\big\|^2}\right).  
		\end{align} 
		\vspace{-1 em}
		\\	\rule{\textwidth}{0.7pt}
\end{figure*}
%%%%%%%%%%%%%%%%%%%%%%%%%%%%%%%%%%%%%%%

Lemma \ref{lem:capacity_lowSNR} gives the same intuition as the BA solutions \eqref{eq:opt_BA}, \eqref{eq:opt_BA_rev} that to maximize the capacity, we need to assign more bits to the RF chain with larger channel gains in the low SNR regime. 
We further derive an approximation of the capacity with the proposed BA algorithms by applying a BA solution to \eqref{eq:capacity_lowSNR}.
In particular, we consider the case in which the revMMSQE-BA algorithm is applied to the resolution-adaptive ADC architecture since it is more effective in the low SNR regime.
%%%%%%%%%%%%%%%%%%%%%%%%%%%%%%%%%%%%%%%
%	\begin{figure*}[!b]
%	%	\rule{\textwidth}{0.7pt}
%		\begin{align}
%			\label{eq:capacity_diag_BA}
%			\tilde C^{RBA}_{diag}\big({\bf b}, {\bf H}_{\rm b} \big) =\sum_{i = 1}^{N_{\rm RF}}\log_2\left(1+ \frac{p_u\Big(1-\pi\sqrt{3}\,2^{-({\bar b}+1)}\Big( N_{\rm RF}^{-1}\|[{\bf H}_{\rm b}]_{i,:}\|^{\frac{2}{3}}\sum_{j = 1}^{N_{\rm RF}}\|[{\bf H}_{\rm b}]_{j,:}\|^{\frac{2}{3}}\Big)_{2^{\bar b}}^-\Big) \|[{\bf H}_{\rm b}]_{i,:}\|^2}{1+p_u\pi\sqrt{3}\,2^{-({\bar b}+1)}\Big(N_{\rm RF}^{-1}\|[{\bf H}_{\rm b}]_{i,:}\|^{\frac{2}{3}} \sum_{j = 1}^{N_{\rm RF}}\|[{\bf H}_{\rm b}]_{j,:}\|^{\frac{2}{3}}\Big)_{2^{\bar b}}^- \,\|[{\bf H}_{\rm b}]_{i,:}\|^2}\right).  
%		\end{align} 
%	\end{figure*}
%%%%%%%%%%%%%%%%%%%%%%%%%%%%%%%%%%%%%%%

\begin{proposition}
	\label{pr:capacity_lowSNR_BA}
	For the low SNR, the capacity under the proposed resolution-adaptive ADC architecture with the revMMSQE-BA algorithm, $C^{RBA}_{low}\big({\bf b},{\bf H}_{\rm b}\big)$, can be approximated as \eqref{eq:capacity_lowSNR_BA}.
\end{proposition}
\begin{proof}
	Forcing non-negativity to the revMMSQE-BA solution \eqref{eq:opt_BA_rev} as 
	$b_i = (\hat b^{rev}_i)^+$ where $(a)^+ = \max(a,0)$, we apply $b_i = (\hat b^{rev}_i)^+$ to \eqref{eq:capacity_lowSNR}.
	Then, the capacity $C^{RBA}_{low}\big({\bf b},{\bf H}_{\rm b}\big)$ can be approximated as
	
	\vspace{-1em}
	\small
	\begin{align}
		\nonumber
		& C^{RBA}_{low}\big({\bf b},{\bf H}_{\rm b}\big) \approx C_{low}\big((\hat{\bf b}^{rev})^+, {\bf H}_{\rm b}\big)
		%= \log_2\Bigg(1+\sum_{i = 1}^{N_{\rm RF}} \frac{p_u \alpha_i \|[{\bf H}_{\rm b}]_{i,:}\|^2}{1+p_u(1-\alpha_i)\|[{\bf H}_{\rm b}]_{i,:}\|^2}\Bigg).
		\\ \nonumber
		& \stackrel{(a)} \approx \log_2\Bigg(1+\sum_{i=1}^{N_{\rm RF}} \frac{p_u\Big(1-\frac{\pi\sqrt{3}}{2}2^{ -2(\hat b^{rev}_i)^+}\Big)^+\|[{\bf H}_{\rm b}]_{i,:}\|^2}{1+p_u\frac{\pi\sqrt{3}}{2}2^{ -2(\hat b^{rev}_i)^+}\|[{\bf H}_{\rm b}]_{i,:}\|^2}\Bigg)
		\\ \label{eq:capacity_lowSNR_BA_pf}
		& \stackrel{(b)}{\approx} \log_2\Bigg(1+\sum_{i=1}^{N_{\rm RF}} \frac{p_u\Big(1-\frac{\pi\sqrt{3}}{2}2^{ -2(\hat b^{rev}_i)^+}\Big)\|[{\bf H}_{\rm b}]_{i,:}\|^2}{1+p_u\frac{\pi\sqrt{3}}{2}2^{ -2(\hat b^{rev}_i)^+}\|[{\bf H}_{\rm b}]_{i,:}\|^2}\Bigg)
	\end{align}
	\normalsize
	where (a) is from the approximation of $\alpha_i$ and (b) comes from removing the non-negativity condition of $\alpha_i$. 
	Since $p_u$ and $\|[{\bf H}_{\rm b}]_{i,:}\|$, which corresponds to $\alpha_i < 0 $ are small, the error from the approximation (b) can be negligible.
	Rearranging \eqref{eq:capacity_lowSNR_BA_pf}, we derive \eqref{eq:capacity_lowSNR_BA}.
\end{proof}
%\begin{proposition}
%	\label{pr:capacity_diag_BA}
%	When rows of ${\bf H}_{\rm b}$ are nearly orthogonal to each other, the capacity of the proposed resolution-adaptive ADC architecture \eqref{eq:AQNM2} with the revMMSQE-BA algorithm can be approximated as \eqref{eq:capacity_diag_BA}.
%\end{proposition}
%\begin{proof}
%	Similarly to the proof of Proposition \ref{pr:capacity_lowSNR_BA}, applying $b_i = (\hat b^{rev}_i)^+$ to \eqref{eq:capacity_diag}, removing the non-negativity condition of $\alpha_i$, and rearranging the resulting capacity expression, we derive \eqref{eq:capacity_diag_BA}.
%\end{proof}
Since the revMMSQE-BA solution $\hat {\bf b}^{rev}$ is the function of ${\bf H}_{\rm b}$, $\tilde C^{RBA}_{low}$ in \eqref{eq:capacity_lowSNR_BA} is only a function of channels and captures the capacity that the proposed flexible ADC architecture can achieve adaptively for a given channel by using the revMMSQE-BA algorithm.
}

Now, regarding the implementation issue of the algorithm, we remark the following ADC resolution switching scenarios.
%. as only a function of ${\bf H}_{\rm b}$ in given system environment.

\begin{remark}
Resolution switching at every channel coherence time allows the proposed architecture to adapt to different channel fading realizations, implying that it is the best switching scenario.
Such coherence time switching in mmWave channels, however, may not be feasible due to the very short coherence time of mmWave channels \cite{khan2012millimeter}.
Consequently, the switching period needs to be the multiples of the coherence time. 
In this case, switching at the time-scale of slowly changing channel characteristics %such as large-scale fading and AoA 
marginally degrades the performance of the algorithms.
Then, the worst-case scenario is not to exploit the flexibility of ADC resolutions, which is equivalent to have a infinitely long switching period, and converges to the fixed-ADC system with analog beamforming.
\end{remark}

In the next section, using a practical receiver, e.g., MRC, we analyze the worst-case scenario in terms of an \emph{ergodic achievable rate} due to the intractability of the analysis with the BA solutions. 
% to provide deeper insight for the proposed system.
% Although the ergodic rate analysis with the BA solutions is not tractable, the worst-case analysis is.
The derived ergodic rate of the proposed system for the worst-case scenario offers the insight of the system performance as a function of the system parameters.

\section{Worst-Case Analysis}
\label{sec:main}
%%%%%%%%%%%%%%%%%%%%%%%%%%%%

%%%%%%%%%%%%%%%%%%%%%%%%%%%%
%\subsection{Ergrodic Achievable Rate}
%\label{sec:Rate}
%%%%%%%%%%%%%%%%%%%%%%%%%%%%

%Using MRC equalizer at the BS, we first derive the average achievable rate of $n$th user of the typical system where all ADCs have the same resolution, which can be also considered as the optimal case when resolution switches have an infinite update period.
We derive the ergodic achievable rate of user $n$ for the hybrid beamforming architecture with fixed-ADCs over mmWave channels.
The number of quantization bits in \eqref{eq:AQNM2} is considered to be the same, i.e., $b_i = b$, and thus, $\alpha_i = \alpha$ for $i = 1,\cdots, N_{\rm RF}$.
% without resolution switching, which can be considered as the case where resolution switches have an infinite update period.
%In such environment, we derive the achievable rate of \eqref{eq:AQNM2} with homogeneous ADC configuration, i.e., $b_i = b$, $i = 1,\cdots, N_{\rm RF}$.
%since homogeneous ADC bit configuration is optimal to maximize the achievable rate with infinite update period in an ergodic sense due to the randomness of mmWave channels.
%Thus, we derive the achievable rate with homogeneous ADC bit configuration $b_i = b$, $i = 1,\cdots, N_{\rm RF}$.
%Hence, the quantization gain $\alpha_i$ in \eqref{eq:AQNM2} is also $\alpha_i = \alpha$, $i=1,\cdots,N_{\rm RF}$.
% to analyze the achievable rate of the quantization system \eqref{eq:AQNM2}.
Using MRC, the quantized signal vector is
\begin{align}
\nonumber
{\bf y}_{\rm q}^{\rm mrc} =  {\bf H}_{\rm b}^H{\bf y}_{\rm q}
 = \alpha \sqrt{p_u}{\bf H}_{\rm b}^H{\bf H_{\rm b}}{\bf s} + \alpha {\bf H}_{\rm b}^H {\bf n} +  {\bf H}_{\rm b}^H{\bf n}_{\rm q}
\end{align}
%After quantizing the received signals~\eqref{eq:AQNM2},  MRC equalizer is applied. 
and the ${n}$th element of ${\bf y}_{\rm q}^{\rm mrc} $ for the user $n$ is expressed as
% EQUATION
\begin{align}
\nonumber
%{\bf g}_{m^*}^H{\bf y}_{\rm q} = \sqrt{p_u}{\bf g}_{m^*}^H {\bf W}_\alpha {\bf g}_{m^*} x_{m^*} + 
%\sqrt{p_u}\sum_{\substack{m = 1\\ m \neq m^*}}^{M} {\bf g}_{m^*}^H {\bf W}_\alpha {\bf g}_m x_m + 
%{\bf g}_{m^*}^H {\bf W}_\alpha \tilde {\bf n} + {\bf g}_{m^*}^H {\bf n}_{\rm q}
y_{{\rm q},n}^{\rm mrc} =& \alpha  \sqrt{p_u}{\bf h}_{{\rm b},n}^H {\bf h}_{{\rm b},n} s_{n}
 \\ \label{eq:MRC}
&+ 
\alpha  \sqrt{p_u}\sum_{\substack{k = 1\\ k \neq n}}^{N_u}{\bf h}_{{\rm b},n}^H {\bf h}_{{\rm b},k} s_k + 
\alpha {\bf h}_{{\rm b},n}^H   {\bf n} + {\bf h}_{{\rm b},n}^H  {\bf n}_{\rm q}.
\end{align}
Since ${\bf h}_{{\rm b},n} = \sqrt{\gamma_n}{\bf g}_n$ from \eqref{eq:RFchannel}, the desired signal power in \eqref{eq:MRC} becomes $p_u\alpha^2 \gamma_n^2 \|{\bf g}_{n}\|^4$ and the noise-plus-interference power is given by
\begin{align}
	%\label{eq:IplusN}
	\nonumber
	\Psi_{\bf G} =& p_u\alpha^2 \gamma_n\sum\limits_{\substack{k = 1\\ k \neq n}}^{N_u}\gamma_k|{\bf g}_{n}^H {\bf g}_k|^2 + \alpha^2 \gamma_n \|{\bf g}_{n}\|^2 
	\\ \nonumber
	&+\alpha (1-\alpha)\gamma_n {\bf g}_{n}^H {\rm diag}(p_u{\bf G} {\bf D}_\gamma {\bf G}^H + {\bf I}_{N_{\rm RF}}) {\bf g}_{n}.
\end{align}
% EQUATION
%\begin{align}
%\label{eq:SQNR}
%{\rm SINR}_{n}^{\rm rf, mrc} = \frac{p_u\alpha^2 \gamma_n^2 \|{\bf g}_{n}\|^4}{p_u\alpha^2 \gamma_n\sum\limits_{\substack{k = 1\\ k \neq n}}^{N_u}\gamma_k|{\bf g}_{n}^H {\bf g}_k|^2 + \alpha^2 \gamma_n \|{\bf g}_{n}\| + \alpha (1-\alpha)\gamma_n {\bf g}_{n}^H {\rm diag}(p_u{\bf H}_{\rm b}^H {\bf H}_{\rm b}+{\bf I}_{N_{\rm RF}}) {\bf g}_{n}}.
%{\bf W}_\alpha{\bf W}_\beta {\rm diag}(p_u{\bf GG}^H + {\bf I})
%\end{align}
Simplifying the ratio of the two power terms, the achievable rate of the $n$th user is expressed as
% EQUATION
\begin{align}
\label{eq:AchievableRate}
R_{n} = \mathbb{E} \Bigg[ \log_2\left(1+ \frac{ p_u\alpha \gamma_n \|{\bf g}_{n}\|^4}{\tilde\Psi_{\bf G}} \right)\Bigg]
\end{align}
where
\begin{align}
\nonumber
\tilde \Psi_{\bf G} = &p_u\alpha \sum\limits_{\substack{k = 1\\ k \neq n}}^{N_u}\gamma_k\big|{\bf g}_{n}^H {\bf g}_k\big|^2 + \alpha \|{\bf g}_{n}\|^2 
\\ \nonumber
&+  (1-\alpha) {\bf g}_{n}^H {\rm diag}\big(p_u{\bf G} {\bf D}_\gamma {\bf G}^H + {\bf I}_{N_{\rm RF}}\big) {\bf g}_{n}.
\end{align}
Considering large antenna arrays at the receiver, we use Lemma \ref{lem1} to characterize the achievable rate \eqref{eq:AchievableRate}.

% LEMMA
\begin{lemma}
\label{lem1}
Considering large antenna arrays at the BS, the uplink ergodic achievable rate \eqref{eq:AchievableRate} for the user $n$ can be approximated as
\begin{align}
\label{eq:rate_apprx}
\tilde R_{n} = \log_2\left(1+ \frac{p_u\alpha \gamma_n \mathbb{E}\big[\|{\bf g}_{n}\|^4\big]}{\mathbb{E}\big[\tilde \Psi_{\bf G}\big]} \right)
\end{align}
where 
\begin{align}
\nonumber
\mathbb{E}\big[\tilde \Psi_{\bf G}\big]& = \mathbb{E}\Bigg[p_u\alpha \sum\limits_{\substack{k = 1\\ k \neq n}}^{N_u}\gamma_k\big|{\bf g}_{n}^H {\bf g}_k\big|^2 + \alpha \|{\bf g}_{n}\|^2 
\\ \label{eq:E_psi}
&+  (1-\alpha) {\bf g}_{n}^H {\rm diag}\big(p_u{\bf G} {\bf D}_\gamma {\bf G}^H + {\bf I}_{N_{\rm RF}}\big) {\bf g}_{n}\Bigg].
\end{align}
\end{lemma}
\begin{proof}
We apply Lemma 1 in \cite{zhang2014power} to \eqref{eq:AchievableRate}.
\end{proof}

According to Lemma \ref{lem1} in \cite{zhang2014power}, the approximation in \eqref{eq:rate_apprx} becomes more accurate as the number of the BS antennas increases. 
Thus, this approximation will be particularly accurate in systems with the large number of antennas. 
% LEMMA
%\begin{lemma}[\cite{zhang2014power}, Lemma 1]
%\label{lem1}
%If $X = \sum_{i=1}^{t_1}X_i$ and $Y = \sum_{j=1}^{t_2}Y_j$ are both sums of nonnegative random variables $X_i$ and $Y_j$, then we get the following approximation
% EQUATION
%\begin{align}
%\label{eq:lem1}
%\mathbb{E}\left[\log_2\left(1+\frac{X}{Y}\right)\right]
%\approx \log_2\left(1+\frac{\mathbb{E}[X]}{\mathbb{E}[Y]}\right).
%\end{align}
%\end{lemma}
%\begin{proof}
%See Lemma 1 in \cite{zhang2014power}.
%\end{proof}
%
%Note that Lemma \ref{lem1} does not require the random variables $X$ and $Y$ to be independent and the approximation in \eqref{eq:lem1} becomes more accurate as $t_1$ and $t_2$ increase. 
%Accordingly, this approximation will be particularly accurate in the systems with the large number of antennas. 
%
%We derive the probability of both two different channel path gains not to be zero in Lemma \ref{lem2}, which is used to characterize the desired signal power when using the MRC equalizer.
%
% LEMMA
%\begin{lemma}
%\label{lem2}
%Under the channel model described in Section \ref{subsec:channel}, the probability that both two different channel path gains of the user channel ${\bf g} \in \mathbb{C}^{N_r}$ are not zero is 
%\begin{align}
%\label{eq:lem2}
%\Pr\Big(g_{i} \neq 0, g_{i+d} \neq 0\Big) =
%\begin{cases}
%    \frac{L(L-1)(Q-d_{\rm min})}{N_rQ(Q-1)}, &\quad \text{if } d\in \mathcal{D}_i\\
%    0 &\quad \text{else.}
%\end{cases}
%\end{align}
%where $d_{\rm min} = \min(d,N_r - d)$ and $\mathcal{D}_i = \{d \in \mathbb{N} |\, d \leq N_r-i, d_{\rm min} < Q\}$, $i = 1,\cdots, N_r$.
%\end{lemma}
Using Lemma \ref{lem1}, we derive the closed-form approximation of \eqref{eq:AchievableRate} as a function of system parameters: the transmit power,
the number of BS antennas, RF chains, users and quantization bits, and the near average number of propagation paths. %, which evaluates the achievable rate of the mmWave massive MIMO systems.
\begin{theorem}
\label{thm:SE_uniform} 
The uplink ergodic achievable rate of the user $n$ in the considered system with fixed ADCs is derived in a closed-form approximation as
% EQUATION
\begin{align}
\label{eq:SE_uniform} 
\tilde R_{n} = \log_2\left( 1+\frac{p_u \gamma_n \alpha \big(\lambda_{\rm p}^2  + 2\lambda_{\rm p} + 2e^{-\lambda_{\rm p}}\big)}{\eta}\right)
%%
%\tilde R_{n} = \log_2\left( 1+\frac{p_u \alpha \gamma_n(\lambda_{\rm p}^2  + 2\lambda_{\rm p} + 2e^{-\lambda_{\rm p}} )}{\frac{p_u}{N_{\rm RF}}(\lambda_{\rm p} + e^{-\lambda_{\rm p}})^2  \sum_{\substack{k=1\\k\neq n}}^{N_u}\gamma_k + 2(1-\alpha)p_u\gamma_n(\lambda_{\rm p}+ e^{-\lambda_{\rm p}}) + (\lambda_{\rm p} + e^{-\lambda_{\rm p}})}\right)
%%
%\tilde R_{n}^{\rm rf, mrc}= \log_2\left(1+\frac{p_u\gamma_n\alpha(L+1)}{ p_u\Big(\frac{L}{N_{\rm RF}}\mathcal{S}_{\gamma_{k \setminus n}} + 2\gamma_n\beta\Big)+ {1}} \right).
\end{align}
where 
\footnotesize
\begin{align}
\nonumber
\eta =& \Big(\lambda_{\rm p}+ e^{-\lambda_{\rm p}}\Big)
 \Bigg(1 + 2p_u\gamma_n (1-\alpha) + \big(\lambda_{\rm p}+ e^{-\lambda_{\rm p}}\big) \frac{p_u}{N_{\rm RF}} \sum_{\substack{k=1\\k\neq n}}^{N_u}\gamma_k  \Bigg).
\end{align}
\end{theorem}
\normalsize
\begin{proof} 
See Appendix \ref{appx:SE_uniform}
\end{proof}
%We remark that \eqref{eq:SE_same} is not a function of $Q$. 
%Thus, the same result as \eqref{eq:SE_same} can be derived by assuming uniform scattering gain over $[0,2\pi]$ for the $\alpha_i = \alpha$ case.
%This implies that under our system model, spatial angle spread has no influence on the uplink achievable rate when all ADCs are equipped with the same number of bits.
%We also note that the derived achievable rate \eqref{eq:SE_variable} in Theorem \ref{SE_variable} is for the system without ADC switches, which does not adaptively link the RF preprocessed signals to the ADC pairs based on channel gains.  
Note that since the obtained ergodic rate in Theorem \ref{thm:SE_uniform} is from the worst-case scenario, it can serve as the lower bound of the proposed architecture.
This further implies that the proposed system can achieve higher ergodic rate than the derived rate by leveraging the flexibility of ADC resolutions.
In addition, the derived ergodic rate explains general tradeoffs of the proposed system thanks to its tractability as a function of the system parameters.
% Compared to previous work \cite{ngo2013energy,zhang2014power,fan2015uplink}, the derived ergodic rate further explains the effect of the near average number of propagation paths $\lambda_{\rm p}$ and the number of RF chains for mmWave hybrid massive MIMO systems.
In contrast to the prior work \cite{mo2016achievable} which assumes the quasi-static setting, the achievable rate in Theorem \ref{thm:SE_uniform} considers mmWave fading channels in the ergodic sense.
Accordingly, the derived ergodic rate measures the achievable rates by adopting the rate to the different fading realizations and thus offers more realistic evaluation than the quasi-static analysis in contemporary wireless systems. 
%Ignoring the transmission of a coded packet over different fading realizations, the quasi-static setting has lost much of its relevance to the contemporary wireless communications where the transmission of a coded packet over multiple fading realizations is common  \cite{lozano2012yesterday}.
%Especially, the quasi-static setting is not adequate more in mmWave channels with the much shorter channel coherence time. 
%Accordingly, the derived ergodic rate which measures the achievable rates adopting the rate to the different fading realizations can offer more realistic evaluation than the quasi-static analysis in contemporary wireless systems.
 %{\bf Edit} In addition, Theorem \ref{thm:SE_uniform} embraces \cite{ngo2013energy}, \cite{zhang2014power} and \cite{fan2015uplink} as special cases by adjusting the parameters.

We derive Corollary \ref{col:path_large} for simplifying the ergodic rate in \eqref{eq:SE_uniform} when the near average number of propagation paths $\lambda_p$ is moderate or large, and further provide remarks on the derived rate in behalf of profound understanding.

\begin{corollary}
\label{col:path_large}
When the near average number of propagation paths $\lambda_{\rm p}$ is moderate or large, \eqref{eq:SE_uniform} can be approximated as
\begin{align}
\label{eq:path_large}
%\tilde R_n' = \log_2\left(1 + \frac{p_u\alpha\gamma_n(\lambda_{\rm p} + 2) }{2p_u\gamma_n(1-\alpha) + \frac{p_u}{N_{\rm RF}}  \lambda_{\rm p} \sum_{\substack{k=1\\k\neq n}}^{N_u}\gamma_k + 1}\right).
\tilde R_n^{\dagger} = \log_2\left(1 + \frac{p_u \gamma_n \alpha (\lambda_{\rm p} + 2) }{1 + p_u\Big(2\gamma_n(1-\alpha) + \frac{\lambda_{\rm p}}{N_{\rm RF}} \sum_{\substack{k=1\\k\neq n}}^{N_u}\gamma_k\Big) }\right).
\end{align}
\begin{proof}
When $\lambda_{\rm p}$ is moderate or large enough, we can approximate $\lambda_{\rm p} + e^{-\lambda_{\rm p}} \approx \lambda_{\rm p}$. 
%since $\lambda_{\rm p} \gg e^{-\lambda_{\rm p}}$. 
Hence, we have the approximation \eqref{eq:path_large} by replacing $\lambda_{\rm p} + e^{-\lambda_{\rm p}}$ with $\lambda_{\rm p}$ in \eqref{eq:SE_uniform}.
\end{proof}
\end{corollary}
\begin{remark}
\label{remark:inf_b}
For fixed $\lambda_p$, \eqref{eq:SE_uniform} with infinite-resolution ADCs ($b \to \infty$) reduces to 
\begin{align}
\label{eq:infinite_b}
\tilde R_{n} \to \log_2\left(1+\frac{p_u\gamma_n(\lambda_{\rm p}^2 + 2\lambda_{\rm p} + 2e^{-\lambda_{\rm p}})}{1 + (\lambda_{\rm p} + e^{-\lambda_{p}})\frac{p_u}{N_{\rm RF}}\sum_{\substack{k=1\\k\neq n}}^{N_u}\gamma_k} \right).
%\hat R_n^{\rm rf, mrc} \to \log_2\left(1+\frac{p_u\bar p_c(L+1)}{1+p_u\bar p_c\frac{L(N_u-1)}{N_r}}\right),\quad as\ b \to \infty.
\end{align}
%Note that \eqref{eq:infinite_b} depends on the near average number of propagation paths $\lambda_{\rm p}$ and the number of RF chains $N_{\rm RF}$. 
It is clear from \eqref{eq:infinite_b} that the uplink rate can be improved by using more RF chains (larger $N_{\rm RF}$), which reduces the user interference. 
%With the fixed $\lambda_{p}$, however, \eqref{eq:infinite_b} cannot be infinitely improved by increasing $N_{\rm RF}$.
Let $N_{\rm RF} = \tau N_r$ where $ 0 <\tau < 1 $, then for the fixed $\lambda_{p}$, the full-resolution rate \eqref{eq:infinite_b} increases to
\begin{align}
\nonumber
\tilde R_{n} \to \log_2\Big(1+ p_u\gamma_n\big(\lambda_{\rm p}^2 + 2\lambda_{\rm p} + 2e^{-\lambda_{\rm p}}\big)  \Big),\  \text{as } N_r \to \infty.
\end{align}
%Furthermore, if $L = N_r$, \eqref{eq:infinite_b} agrees with the achievable rate of massive MIMO systems in Rayleigh channels derived in \cite{zhang2014power}.
\end{remark}

% REMARK
\begin{remark}
\label{remark:inf_pu}
When using MRC,
%, both the desired signal power and the inter-user interference power increase with the transmit power $p_u$.
the uplink user rate transfers to the interference-limited regime from the noise-limited regime as $p_u$ increases.
Consequently, for fixed $\lambda_p$, \eqref{eq:SE_uniform} with the infinite transmit power ($p_u \to \infty$), converges to $\tilde R_{n} \to$
%For the fixed $\lambda_p$, when the transmit power tends to infinity ($p_u \to \infty$), \eqref{eq:SE_uniform} converges to

\small
\begin{align}
\label{eq:infinite_pu}
\log_2\left(1+\frac{\gamma_n \alpha (\lambda_{\rm p}^2  + 2\lambda_{\rm p} + 2e^{-\lambda_{\rm p}})}{\big(\lambda_{\rm p}+ e^{-\lambda_{\rm p}}\big)\bigg(2\gamma_n (1-\alpha) + \frac{(\lambda_{\rm p}+ e^{-\lambda_{\rm p}})}{N_{\rm RF}}  \sum_{\substack{k=1\\k\neq n}}^{N_u}\gamma_k \bigg)} \right). 
%
%\hat R_n^{\rm rf, mrc} \to \log_2\left(1+\frac{p_u\bar p_c(L+1)}{1+p_u\bar p_c\frac{L(N_u-1)}{N_r}}\right),\quad as\ b \to \infty.
\end{align}
\normalsize
The interference power can be eliminated by using an infinite number of antennas with $N_{\rm RF} = \tau N_r$ where $0<\tau<1$. 
In this case, \eqref{eq:infinite_pu} approaches to 
\begin{align}
\label{eq:infinite_pu_Nr}
\tilde R_{n} \to \log_2\left(1+\frac{\alpha (\lambda_{\rm p}^2  + 2\lambda_{\rm p} + 2e^{-\lambda_{\rm p}})}{2 (1-\alpha) \big(\lambda_{\rm p}+ e^{-\lambda_{\rm p}}\big) } \right),\  \text{as } N_r \to \infty.
\end{align}  
The result \eqref{eq:infinite_pu_Nr} shows that even the infinite transmit power ($p_u \to \infty$) and the infinite number of BS antennas ($N_r \to \infty$) cannot fully compensate for the degradation caused by the quantization distortion when mmWave channels have a fixed number of propagation paths independent to $N_r$.
\end{remark}

We now consider that $\lambda_{\rm p}$ is an increasing function of $N_r$ \cite{raghavan2011sublinear} since larger antenna arrays with a fixed antenna spacing capture more physical paths due to larger array aperture. 
Then, Corollary \ref{col:path_large} holds in the large antenna array regime. 

% REMARK
\begin{remark}
\label{remark:3}
Without loss of generality, we assume $\lambda_{\rm p} = \epsilon N_r$ where $0 < \epsilon < 1$. 
Considering large antenna arrays with $N_{\rm RF} = \tau N_r$, \eqref{eq:path_large} becomes 
\begin{align}
\label{eq:rm3}
\tilde R_n^{\dagger} = \log_2\left(1 + \frac{p_u \gamma_n \alpha (\epsilon N_r + 2) }{1 + p_u\Big(2\gamma_n(1-\alpha) + \epsilon/\tau \, \sum_{\substack{k=1\\k\neq n}}^{N_u}\gamma_k\Big) }\right).
\end{align}
The achievable rate \eqref{eq:rm3} increases to infinity as $N_r \to \infty$  for any quantization bits $b$, which is not the case for the fixed $\lambda_{\rm p}$ as previously shown in Remark \ref{remark:inf_b} and \ref{remark:inf_pu}.
With finite $N_r$, however, \eqref{eq:rm3} cannot increase to infinity but converges to
\begin{align}
\label{eq:rm3-2}
 \tilde R_n^{\dagger} \to \log_2\left(1 + \frac{\gamma_n \alpha (\epsilon N_r + 2) }{2\gamma_n(1-\alpha) + \frac{\epsilon}{\tau}\sum_{\substack{k=1\\k\neq n}}^{N_u}\gamma_k}\right), \  \text{as }  p_u \to \infty.
\end{align}
It is observed that the convergence in \eqref{eq:rm3-2} is from the limited number of propagation paths. %and quantization corresponds to the intuition in \cite{mo2015capacity} for the mmWave channel with 1-bit quantization case.
\end{remark}

% REMARK
\begin{remark}
\label{remark:L_scale}
Assuming that the transmit power inversely scales with the number of RF chains that is proportional to the number of BS antennas, i.e., $p_u = E_s/N_{\rm RF} = E_s/(\tau N_r)$, the rate in \eqref{eq:path_large} with fixed $E_s$ and $\lambda_{\rm p} = \epsilon N_r$ reduces to
\begin{align}
\label{eq:power_scale}
\tilde{R}_{n}^\dagger \to \log_2 \big(1+E_s  \gamma_n \alpha \epsilon/\tau \big),\quad \text{as}\ N_r \to \infty.
\end{align}
%Note that \eqref{remark:L_scale} is similar to the uplink rate of low-resolution massive MIMO systems with Rayleigh channels \cite{fan2015uplink} for the power-scaling law, but the achievable rate is logarithmically scaled by $\epsilon/\tau$.
Thus, \eqref{eq:power_scale} shows that we can scale down the user transmit power $p_u$ proportionally to $1/N_r$ maintaining a desirable rate.
In addition, \eqref{eq:power_scale} can be improved by using more quantization bits (larger $\alpha$). 
This result is similar to that of the uplink rate of low-resolution massive MIMO systems with Rayleigh channels \cite{fan2015uplink} but different in that \eqref{eq:power_scale} includes the factor of $\epsilon/\tau$ due to the analog beamforming and the sparse nature of mmWave channels.
% in mmWave channels.
%If $\epsilon = 1$, which implies $L = N_r$, \eqref{remark:L_scale} aligns with the conclusion in \cite{fan2015uplink}.
%Considering the case with infinite resolution ADCs with $\epsilon = 1$, \eqref{remark:L_scale} further agrees with the result in \cite{ngo2013energy}.
\end{remark}
{\color{black} 
The proposed BA algorithms and rate analysis are for the single-cell assumption and can be derived similarly for a multi-cell assumption. 
Although changes in the algorithms and ergodic rate for the multi-cell scenario are minor and beyond the scope of this paper, the multi-cell analysis is included in Appendix \ref{appx:multi-cell} for completeness.} 
In the following section, we evaluate the performance of the proposed BA algorithms.
We also validate Theorem \ref{thm:SE_uniform} and Corollary \ref{col:path_large}, and confirm the observations made in this section.

% % TABLE
%\begin{table}[b]
%\centering
%\caption{{\color{black} Simulation Parameters}}\label{tb:param}
%\begin{tabular}{ l c }
%  \hline
%  \hline
%  \qquad  Parameters  & \quad Values\\
%  \hline
%\quad Carrier frequency   &  \quad  $28$ GHz \\
%\quad Bandwidth &  \quad $1$ GHz \\
%\quad Thermal noise density & \  \  $-174$ dBm/Hz \ \  \\
%\quad Noise figure at BS &  \quad 5 dB  \\
%\quad Cell radius &  \quad $200$ m  \\
%\quad Minimum distance &  \quad $30$ m\\
%  \hline
%\hline
%\end{tabular}
%\end{table}

% FIGURE
%\begin{figure}[!t]
%\centering
%$\begin{array}{c c}
%{\resizebox{0.87\columnwidth}{!}
%{\includegraphics{BA_pu_256Nr8Nu1G.png}}
%}\\
%\mbox{(a)} \\
%{\resizebox{0.9\columnwidth}{!}
%{\includegraphics{BA_Nr_20dBm8Nu1G.png}}
%}\\
%\mbox{(b)}
%\end{array}$
%\caption{
%Simulation results of average capapcity for $\bar b \in \{1, 2\}$ constraint bits and $N_u = 8$ users (a) with $N_r = 256$ BS antennas and (b) with $p_u = 20$ dBm.} 
%\label{fig:revBA}
%\end{figure}

%%%%%%%%%%%%%%%%%%%%%%%%%%%%
\section{Simulation Results}
\label{sec:Num}
%%%%%%%%%%%%%%%%%%%%%%%%%%%%

% \cite{zheng201510} 28 GHz 200MHz
% \cite{holma2009lte}
% \cite{3GPP}
%In this section, we evaluate the proposed architecture with the BA algorithms in terms of the sum rate $R$.
%, operating at the channel coherence time where the algorithms achieve their best performance.

%The revMMSQE-BA algorithmDifferent switching period
%Energy efficiency is also adopted as a performance measure to further evaluate the revMMSQE-BA in Section \ref{subsec:EE}. 
% Simulation Setting
We consider single cell with a radius of $200$ $m$ and $N_u=8$ users distributed randomly over the cell. 
The minimum distance between the BS and users is $30$ $m$, i.e., $30 \leq d_n \leq 200$ for $n = 1,\cdots, N_u$ where $d_n$ [$m$] is the distance between the BS and user $n$.
%User locations are fixed once they are dropped in the cell, which corresponds to the setting of the analytical derivations in Section \ref{sec:Rate}.
%The large-scale fading gain $\gamma_{n,\rm dB}$ in \eqref{eq:large_scale_fading_gain} is fixed once $N_u = 8$ users are dropped in the cell, which corresponds to the setting of the analytical derivations in Section \ref{sec:main}.
%{\color{black} Simulation parameters are listed in Table \ref{tb:param}.}
Considering that the system operates at a $28$ GHz carrier frequency, we adopt the mmWave pathloss model in \cite{akdeniz2014millimeter} given as $PL(d_n)\text{ [dB]} = \alpha_{\rm pl} + \beta_{\rm pl} 10\log_{10}d_n + \chi$
% for non-line-of-sight (NLOS) channels:
%\begin{align}
%	\nonumber
%	PL(d_n)\text{ [dB]} = \alpha_{\rm pl} + \beta_{\rm pl} 10\log_{10}d_n + \chi
%\end{align}
where $\chi\sim\mathcal{N}(0,\sigma^2_{s})$ is the lognormal shadowing with $\sigma_s^2 = 8.7$ dB. 
The least square fits 
%of floating intercept and slope over the measured distances
are $\alpha_{\rm pl}  = 72$ dB and $\beta_{\rm pl}  = 2.92$ dB \cite{akdeniz2014millimeter}.
Noise power is calculated as $P_{\rm noise}\text{ [dBm]} = -174 + 10\log_{10}W + n_f$
%\begin{align}
%\nonumber
%P_{\rm noise}\text{ [dBm]} = -174 + 10\log_{10}W + n_f
%\end{align}
where $W$ and $n_f$ are the transmission bandwidth and noise figure at the BS, respectively.
We assume $W = 1$ GHz so as $f_s = 1$ GHz in \eqref{eq:ADCpower}, and $n_f=5$ dB.
Since we assume the normalized noise variance in our system model \eqref{eq:rx_signal}, the large-scale fading gain incorporating the normalization is
\begin{align}
%\label{eq:large_scale_fading_gain}
\nonumber
 {\gamma}_{n,\text{dB}} \text{ [dB]} = -(PL(d_n)   + P_{\rm noise}).
\end{align}
%Simulation parameters are shown in Table \ref{tb:param}.
We consider the near average number of propagation paths $ \lambda_{\rm p} = \epsilon N_r$ and the number of RF chains $N_{\rm RF} = \tau N_r$ with $ \epsilon = 0.1$ and $\tau = 0.5$.
We assume that the slowly changing characteristics of mmWave channels are consistent over $100 \times the\ channel\ coherence\ time$, i.e., large-scale fading gains $\gamma_n$ and the sparse structure of $\bf G$ in \eqref{eq:RFchannel} are fixed over $100$ channel realizations but the complex gains in $\bf G$ change at every channel realization.
This simulation environment holds for the rest of this paper unless mentioned otherwise.
% throughout simulations.

We evaluate the proposed algorithms in terms of the capacity \eqref{eq:capacity}, uplink sum rate with MRC, and energy efficiency.
The uplink sum rate is defined as $ R = \sum_{n = 1}^{N_u} R_n$.
{\color{black} 
The ergodic rate of the $n$th user $R_n$ is computed as follows. 
Applying MRC ${\bf D}_\alpha {\bf H}_{\rm b}$ to the quantized signal vector ${\bf y}_{\rm q}$ in \eqref{eq:AQNM2}, the ergodic rate of user $n$ with ADC bit allocation ${\bf b}$ is given as
%the uplink sum spectral efficiency is $ R = \sum_{k = 1}^{N_u} R_k$, where
\begin{align}
	\label{eq:AchievableRate_BA}
	R_{n}({\bf b}) = \mathbb{E} \Bigg[ \log_2\bigg(1+\frac{p_u\gamma_n|\pmb \alpha^H{\bf v}_{n}|^2}{ \Psi_{\bf G}^{\rm BA}  }
	\bigg)\Bigg]
\end{align}
where 
%${\rm UI}_n = p_u\sum_{\substack{m = 1\\ m \neq n}}^{N_u}\gamma_m|{\bf g}_{n}^H {\bf W}^2_\alpha {\bf g}_m|^2$, ${\rm N}_n = {\bf g}_{n}^H{\bf W}^4_{\alpha}{\bf g}_{n}$, and ${\rm QN}_n = {\bf g}_{n}^H{\bf W}^H_\alpha \mathbf{R}_{\mathbf{n}_{\rm q}\mathbf{n}_{\rm q}}{\bf W}_\alpha {\bf g}_{n}$.
\small
\begin{align}
	\nonumber
	\Psi_{\bf G}^{\rm BA} =   p_u\sum_{\substack{m = 1\\ m \neq n}}^{N_u}\gamma_m|{\bf g}_{n}^H {\bf D}^2_\alpha {\bf g}_m|^2 + {\bf g}_{n}^H({\bf D}^4_{\alpha} + {\bf D}^H_\alpha \mathbf{R}_{\mathbf{n}_{\rm q}\mathbf{n}_{\rm q}} {\bf W}_\alpha ){\bf g}_{n}
\end{align}
with $\pmb \alpha = [\alpha^2_1, \cdots, \alpha^2_{N_{\rm RF}}]^\intercal 
\text{ and }{\bf v}_{n} = \left[|g_{1,{n}}|^2,\cdots,|g_{N_{\rm RF},{n}}|^2\right]^\intercal.$
%\begin{align}
%\nonumber
%\pmb \alpha = [\alpha^2_1, \cdots, \alpha^2_{N_{\rm RF}}]^\intercal 
%\text{ and }{\bf v}_{n} = \left[|g_{1,{n}}|^2,\cdots,|g_{N_{\rm RF},{n}}|^2\right]^\intercal.
%\end{align}}
\normalsize
Note that when quantization bits are same across ADCs, $b_i = b_j$, $\forall i, j$, \eqref{eq:AchievableRate_BA} reduces to \eqref{eq:AchievableRate}.}
%The other simulation setting is the same as in Section \ref{sec:Validation}. 

%FIGURE 
\begin{figure}[!t]\centering
\includegraphics[scale = 0.4]{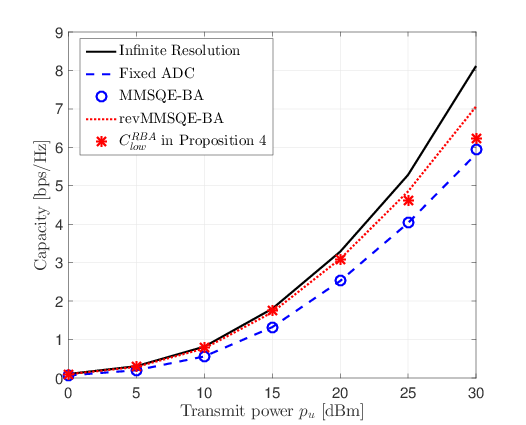}
\caption{Simulation results of the average capacity for $N_u = 8$ users and $N_r = 256$ BS antennas with $\bar b =1$ constraint bit.} 
\label{fig:capacity}
\end{figure}

{\color{black}
%%%%%%%%%%%%%%%%%%%%%%%%%%%%
\subsection{Average Capacity}
\label{subsec:num_capacity}
%%%%%%%%%%%%%%%%%%%%%%%%%%%%

We compare the proposed BA algorithms with the fixed-ADC case and include the infinite-resolution ADC case to indicate an upper bound.
In Fig. \ref{fig:capacity}, the BA algorithms are applied with ${\bar b} = 1$.
Recall that $\bar {b}$ is the number of ADC bits for a fixed-ADC system, which we use to give a reference total ADC power in the constraint for the MMSQE problem.
This indicates that the total ADC power consumption with the algorithms is equal or less than that of $N_{\rm RF}$ $1$-bit ADCs. 
%For the rest of this paper, we also use $\bar b$ to indicates the number of ADC bits for fixed ADCs.
In Fig. \ref{fig:capacity}, the revMMSQE-BA improves the average capacity compared to the fixed ADCs.
Moreover, it nearly achieves the capacity similar to the one with infinite-resolution ADCs in the low SNR regime, offering large energy saving from ADCs.
The MMSQE-BA, however, does not show capacity improvement because the large pathloss makes the noise dominant over the range of $p_u$ in Fig. \ref{fig:capacity}.
Consequently, the performance gap between the algorithms demonstrates the noise-robustness of the revMMSQE-BA.
Although the gap between the capacity with the revMMSQE-BA and its approximation $\tilde C^{RBA}_{low}$ in \eqref{eq:capacity_lowSNR_BA} increases as $p_u$ increases,
$\tilde C^{RBA}_{low}$ provides a good approximation of the capacity with the revMMSQE-BA algorithm in the low SNR regime.}

% FIGURE
\begin{figure}[!t]
\centering
$\begin{array}{c c}
{\resizebox{0.86\columnwidth}{!}
{\includegraphics{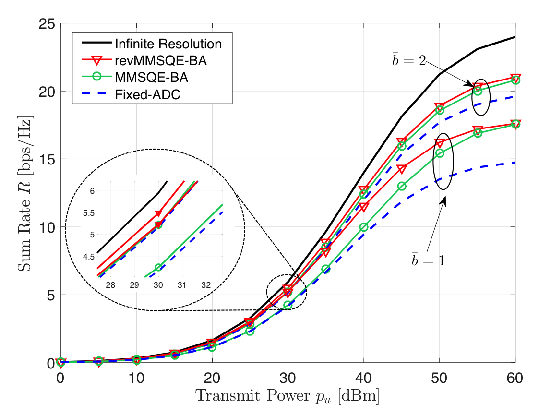}}
}\\
\mbox{(a)} \\
{\resizebox{0.88\columnwidth}{!}
{\includegraphics{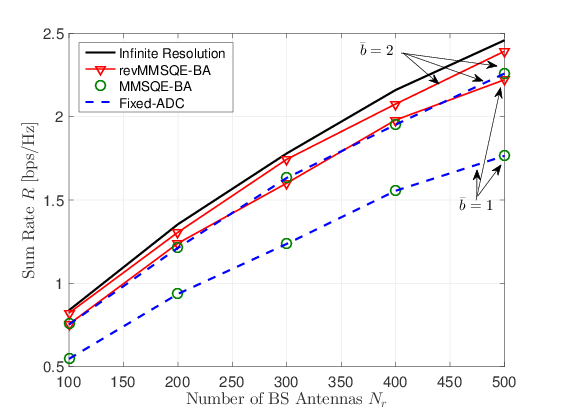}}
}\\
\mbox{(b)}
\end{array}$
\caption{
Simulation results of uplink sum rate for $\bar b \in \{1, 2\}$ constraint bits and $N_u = 8$ users (a) with $N_r = 256$ BS antennas and (b) with $p_u = 20$ dBm.} 
\label{fig:revBA}
\end{figure}

% FIGURE
\begin{figure}[!t]
\centering
$\begin{array}{c c}
{\resizebox{0.86\columnwidth}{!}
{\includegraphics{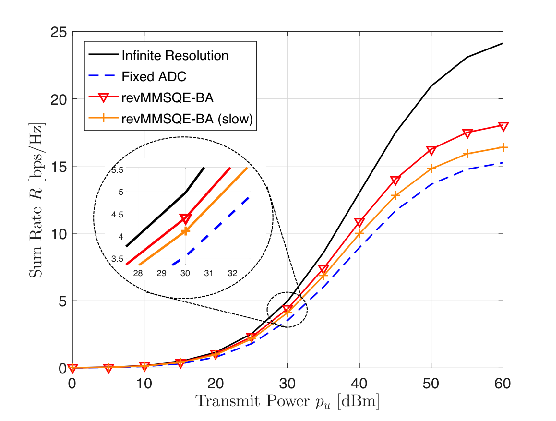}}
}\\
\mbox{(a)} \\
{\resizebox{0.86\columnwidth}{!}
{\includegraphics{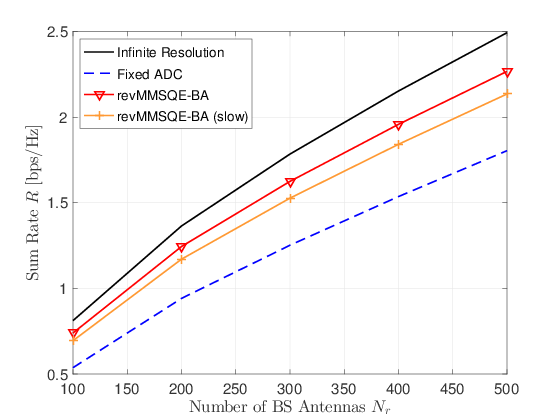}}
}\\ 
\mbox{(b)}
\end{array}$
\caption{
Simulation results of uplink sum rate for $\bar b = 1$ constraint bit and $N_u = 8$ users (a) with $N_r = 256$ BS antennas and (b) with $p_u = 20$ dBm, including the case of switching at slowly changing channel characteristics.} 
\label{fig:revBA_slow}
\end{figure}

%%%%%%%%%%%%%%%%%%%%%%%%%%%%
\subsection{Average Uplink Sum Rate}
\label{subsec:num_sol}
%%%%%%%%%%%%%%%%%%%%%%%%%%%%

Fig. \ref{fig:revBA} shows the uplink sum rate of the MMSQE-BA, revMMSQE-BA and fixed-ADC systems (a) over different transmit power $p_u$ with $N_r = 256$ antennas and $N_u = 8$ users and (b) over the different number of BS antennas $N_r$ with $p_u = 20$ dBm transmit power and $N_u = 8$ users.
%The users are randomly dropped for every channel realization, i.e., the large-scale fading gain $\gamma_{n, \rm dB}$ changes at every iteration.
%The infinite-resolution ADC case is also included to show an upper bound of the ergodic rate without any limitation of power consumption.
%Recall that $\bar b$ represents the number of quantization bits for the fixed-ADC system where all ADCs are equipped with $\bar b$ bits, and the proposed BA algorithms have the total ADC power of the fixed-ADC system as the power constraint.
In Fig. \ref{fig:revBA}(a), the MMSQE-BA and revMMSQE-BA achieve the higher sum rate than the fixed-ADC system for both cases of $\bar b = 1$ and  $\bar b = 2$. 
In particular, the revMMSQE-BA provides the best sum rate over the entire $p_u$ while the MMSQE-BA shows a similar sum rate to the fixed-ADC case in the low SNR regime due to additive noise. 
This demonstrates that the revMMSQE-BA is robust to the noise. 
Notably, the rate of the MMSQE-BA becomes close to that of the revMMSQE-BA in the high SNR regime, which corresponds to the intuition that the revMMSQE-BA performs similarly to the MMSQE-BA in the high SNR regime.

In Fig. \ref{fig:revBA}(b), the revMMSQE-BA also offers the best sum rate for all cases over the entire $N_r$. 
Notice that the sum rate of the revMMSQE-BA with $\bar b =1$ shows similar rate to the fixed-ADC system with $\bar b = 2$, thus implying that the revMMSQE-BA achieves about the 1-bit better sum rate than the fixed-ADC system for the considered system.
In contrast to the revMMSQE-BA, the MMSQE-BA shows no improvement for $p_u = 20$ dBm because the noise power is dominant when allocating quantization bits due to the large pathloss of mmWave channels.
This, again, validates the noise-robustness of the revMMSQE-BA.
{\color{black} Table \ref{tb:BAresult} shows the average ratio of ADCs for different resolutions after applying the revMMSQE-BA algorithm for $\bar b = 1, 2 $ and $3$ with $p_u = 20$ dBm, $N_u = 8$, $N_r = 256$, and $N_{\rm RF} = 128$. 
Intuitively, the number of ADCs with higher resolution increases while that with lower resolution decreases as the constraint bits $\bar b$ increases.
% allowing more ADC power consumption.
For example, the average number of $1$-bit ADCs decreases from $36.10\  (0.282 \times 128)$ to $9.55$ while that of $3$-bit ADCs increases from $5.89$ to $36.53$ as $\bar b$ increases from 1 to 3.}

% TABLE
\begin{table}[!b]
\centering
\caption{{\color {black} Average Ratio of ADCs after Bit Allocation ($\%$)}}\label{tb:BAresult}
\begin{tabular}{ c | c c c c c c c }
  \thickhline
  Constraint&\multicolumn{7}{c}{ADC Resolutions (bits)} \\ 
 \cline{2-8}
  Bits & 0 & 1 & 2 & 3 & 4 & 5 & 6 \\
  \hline
 $\bar b  = 1$  & 40.78 & 28.20  & 26.46 & 4.46 &  0.10 & 0 & 0 \\
 $\bar b  = 2$   & 32.10 & 16.32  & 25.54 & 19.36 &  6.54 & 0.14 & 0 \\
 $\bar b  = 3$   &19.40 & 7.46  & 18.42 & 28.54 &  22.58 & 3.48 & 0.12 \\
% $\bar b  = 1$   & 36.10  & 33.88 & 5.70 &  0.14 & 0 & 0 & 0\\
% $\bar b  = 2$   & 20.88  & 32.68 & 24.78 &  8.38 & 0.17 & 0 & 0\\
% $\bar b  = 3$   & 9.55  & 23.59 & 36.54 &  28.89 & 4.45 & 0.15 & 0\\
  \thickhline
\end{tabular}
\end{table}

%Considering the equivalent power consumption of ADCs, these results reveal the benefit of using the MMSQE-BA and the revMMSQE-BA in terms of the uplink sum rate compared to the homogeneous-ADC case.
In Fig. \ref{fig:revBA_slow}, to consider more realistic implementation of the proposed BA algorithms, we evaluate the revMMSQE-BA with two different switching periods: the channel coherence time and the time-scale of slowly changing channel characteristics (slow switching).
We observe that the slow switching results in small decrease of the sum rate from the coherence-time switching, while still achieving higher sum rate than the fixed-ADCs.
%We observe that the general trends of the ergodic achievable rate for the proposed BA algorithm and the fixed-ADC (worst-case scenario) are similar with the performance gap so that the derived ergodic rate of the fixed-ADC system in Theorem \ref{thm:SE_uniform} can give insight with respect to the system.
Accordingly, the simulation results imply that the proposed hybrid architecture with slow switching can achieve the sum rate in between the revMMSQE-BA with the coherence-time switching and fixed-ADC systems.
In addition, the general trends of the ergodic rate for the proposed BA algorithm and the fixed-ADC (worst-case scenario) are similar with the performance gap.
% and therefore, the derived ergodic rate of the fixed-ADC system in Theorem \ref{thm:SE_uniform} can explain the general tradeoffs of the proposed system.

%With the large transmission bandwidth of $1$ GHz, the increase of the sum rate in bps/Hz becomes much more essential. 
Regarding the total power consumption of the receiver, there will be an additional benefit of using the BA algorithms.
The power saving from turning off the RF process associated with 0-bit ADCs (deactivated ADCs) as a consequence of BA can be accomplished. 
In Section \ref{subsec:EE}, we provide energy efficiency for different ADC configurations to incorporate the additional advantage of the BA algorithms in performance evaluation.

%%%%%%%%%%%%%%%%%%%%%%%%%%%%
\subsection{Energy Efficiency}
\label{subsec:EE}
%%%%%%%%%%%%%%%%%%%%%%%%%%%%

In this subsection, we evaluate the revMMSQE-BA in terms of energy efficiency. 
Energy efficiency can be defined as \cite{mo2016hybrid}
\begin{align}
	\nonumber
	\eta_{\rm eff} = \frac{R\,W}{P_{\rm tot}}\ {\rm bits/Joule}
\end{align}
where $P_{\rm tot}$ is the receiver power consumption.
Recall that $R$ is the sum rate over a single cell, $W$ is the transmission bandwidth.
Let $P_{\rm LNA}$, $P_{\rm PS}$, $P_{\rm RFchain}$, and $P_{\rm BB}$ represent power consumption in the low-noise amplifier, phase shifter, RF chain, and baseband processor, respectively.
Applying an additional power consumption term due to the ADC resolution switching $P_{\rm SW}(b_i)$, the receiver power consumption of the considered system in Fig. \ref{fig:system} is given as
\begin{align}
	\nonumber
	P_{\rm tot} = &N_rP_{\rm LNA} + N_{\rm act}(N_rP_{\rm PS} + P_{\rm RFchain}) \\ \nonumber
	&+ 2\sum_{i=1}^{N_{\rm RF}}\Big(P_{\rm ADC}(b_i) + P_{\rm SW}(b_i)\Big)+ P_{\rm BB} 
\end{align}
where $N_{\rm act}$ is the number of activated ADC pairs ($b_i \neq 0$). 
We assume $P_{\rm LNA} = 20$ mW, $P_{\rm PS} = 10$ mW, $P_{\rm RFchain} = 40$ mW, and $P_{\rm BB} = 200$ mW \cite{mendez2016hybrid, mo2016hybrid}. 
We consider $c = 494$ fJ/conv-step \cite{orhan2015low,chung20097} for $P_{\rm ADC}(b_i)$ in \eqref{eq:ADCpower}.
%, which is a typical achievable value at 1 GHz
According to the measures in \cite{nahata2004high}, the switching power consumption $P_{\rm SW}(b_i)$ when switching from $b_i^{\rm p}$ bits to $b_i$ bits can be modeled as 
\begin{align}
	\label{eq:Psw_model}
	P_{\rm SW}(b_i) = c_{\rm sw}\big|2^{b_i} - 2^{b_i^{\rm p}}\big|, \quad i = 1,\cdots, N_{\rm RF}
\end{align} 
where $c_{\rm sw} = 3.47$ (or $0.94$) mW/conv-step if the resolution increases, $b_i > b^{\rm p}_i$ (or decreases, $b_i < b_i^{\rm p}$).
Notice that \eqref{eq:Psw_model} becomes zero when there is no change in resolution ($b_i = b_i^{\rm p}$).

In the simulation, we compare the following cases:  1) fixed-ADC, 2) revMMSQE-BA with coherence-time switching, 3) reMMSQE-BA with slow switching, and 4) mixed-ADC systems \cite{liang2016mixed}.
We also simulate the infinite-resolution ADC case for benchmarking, assuming $b_\infty = 12$ quantization bits for the case.
For the mixed-ADC system, we employ 1-bit and 7-bit ADCs, and assigns 7-bit ADCs to the RF chains with the strongest channel gains by satisfying the total ADC power constraint $N_{\rm RF}P_{\rm ADC}(\bar b)$.
Consequently, the number of 1-bit and 7-bit ADCs varies depending on the power constraint.
Note that, except for the revMMSQE-BA, the number of activated ADC pairs is equal to that of RF chains $N_{\rm act} = N_{\rm RF}$.
In addition, we impose two harsh simulation constraints on our algorithm.
First, we apply the switching power consumption $P_{\rm SW}(b_i)$ only to the revMMSQE-BA despite the fact that the mixed-ADC system also consumes ADC switching power.
Second, we assume that channel coherence time is equal to symbol duration, implying that if the switching operates at the channel coherence time, it occurs at every transmission.
%Third, we consider that each channel realization does not have any correlation, which increases the average switching power consumption.
%This is pessimistic as changes from one coherence time to another will not lead to an entirely independent channel.

% FIGURE (EVM with revBA)
\begin{figure}[!t]
\centering
$\begin{array}{c c}
{\resizebox{0.86\columnwidth}{!}
{\includegraphics{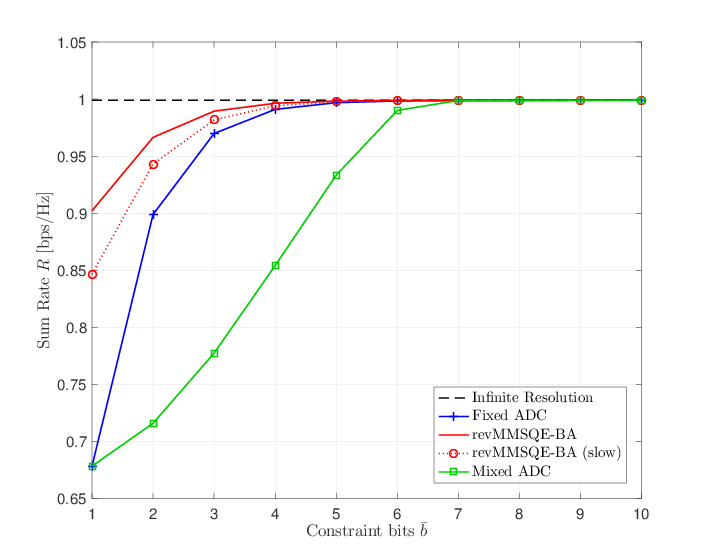}}
}\\ \mbox{(a)} \\
{\resizebox{0.86\columnwidth}{!}
{\includegraphics{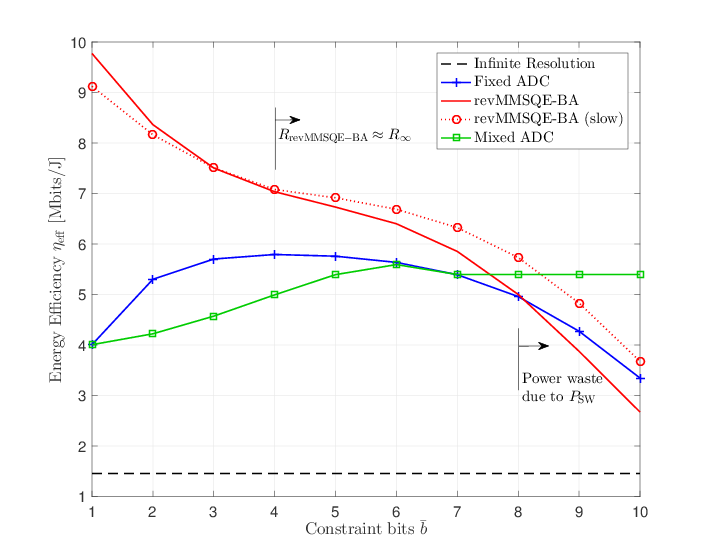}}
} \\
 \mbox{(b)}
\end{array}$
\caption{
Uplink (a) sum rate and (b) energy efficiency simulation results with $N_r = 256$ BS antennas, $N_u = 8$ users and $p_u = 20$ dBm transmit power.} 
\label{fig:REE}
\end{figure}

% FIGURE
\begin{figure}[t]
\centering
$\begin{array}{c c}
{\resizebox{0.88\columnwidth}{!}
{\includegraphics{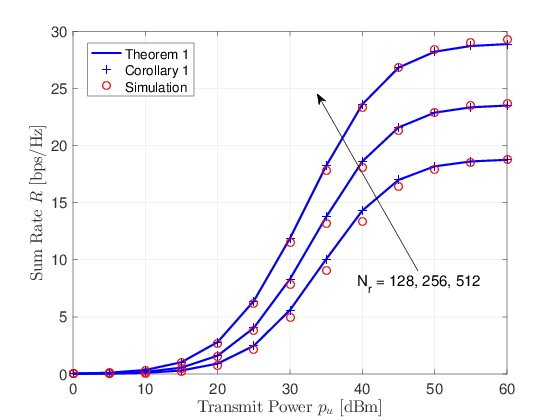}}
}\\
\mbox{(a)}\\
{\resizebox{0.88\columnwidth}{!}
{\includegraphics{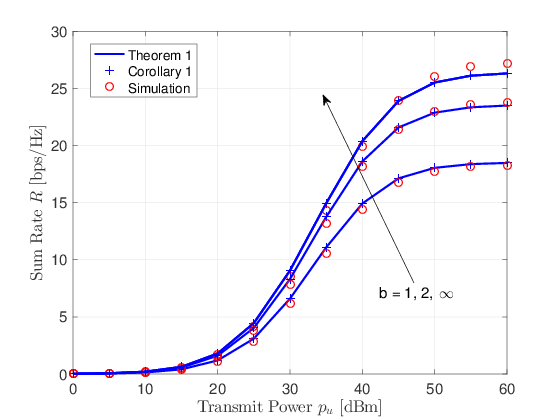}}
}\\ 
\mbox{(b)}
\end{array}$
\caption{
Uplink sum rate of the analytical approximations and the simulation results for $N_u = 8$ users with (a) $b = 2$ quantization bits and $N_r \in \{ 128, 256, 512\}$ BS antennas, and (b) $N_r = 256$ and $b \in \{1, 2, \infty\}$.} 
\label{fig:SumRate}
\end{figure}

In Fig. \ref{fig:REE}, the sum rate and energy efficiency are simulated over different constraint bits $\bar b$.
We note that the fixed-ADC, revMMSQE-BA, revMMSQE-BA (slow) and mixed-ADC system consume the similar total ADC power while the total power consumptions $P_{\rm tot}$ of the revMMSQE-BA and revMMSQE-BA (slow) are not equal to the other cases due to the deactivated (0-bit) ADCs and the switching power $P_{\rm SW}(i)$.
In Fig. \ref{fig:REE}(a), the revMMSQE-BA shows the higher sum rate than the fixed-ADC and mixed-ADC cases in the low-resolution regime ($\bar b\leq 4$), and it converges to the sum rate of the infinite-resolution case faster than the other two cases.
Since the slow switching cannot capture the channel fluctuations caused by small-scale fading, the revMMSQE-BA (slow) shows a lower sum rate than the revMMSQE-BA in the low-resolution regime.
The revMMSQE-BA (slow), however, achieves the higher sum rate than the fixed-ADC and mixed-ADC cases in the low-resolution regime ($\bar b\leq 4$). Given the same power constraint, the mixed-ADC system discloses the lowest sum rate due to the dominant ADC power consumption from the high-resolution ADCs.

In Fig. \ref{fig:REE}(b), the revMMSQE-BA provides the highest energy efficiency in the low-resolution regime, achieving the highest rate.
In the high-resolution regime ($\bar b \geq 8$), the energy efficiency of the revMMSQE-BA is lower than that of the fixed-ADC and mixed-ADC systems due to the dissipation of power consumption in resolution switching. 
Note that although the revMMSQE-BA (slow) shows a lower energy efficiency than the revMMSQE-BA when $\bar b < 4$, it achieves a higher energy efficiency as $\bar b$ increases.
This is because the slow switching accomplishes a better tradeoff between the rate and the switching power consumption than the coherence-time switching as $\bar b$ increases.
Regarding the sum rate and energy efficiency, it is not worthwhile to consider the number of constraint bits above $\bar b = 6$ because the sum rate of the revMMSQE-BA is already comparable with the infinite-resolution system around $\bar b = 4$ with $22$\% better energy efficiency than the fixed-ADC case. 
Therefore, the simulation results demonstrate that the revMMSQE-BA with coherence-time switching provides the best performance, and that the slow switching approach offers performance improvement concerning the implementation.
Fig. \ref{fig:REE} indeed, implies that the proposed BA algorithm eliminates most of the quantization distortion requiring the minimum power consumption.
Accordingly, we can employ existing digital beamformers to the power-constrained system when using the proposed BA algorithm in the low-resolution regime.
%as it makes the quantization error negligible

%FIGURE  (Remark5)
\begin{figure}[!t]\centering
\includegraphics[scale = 0.43]{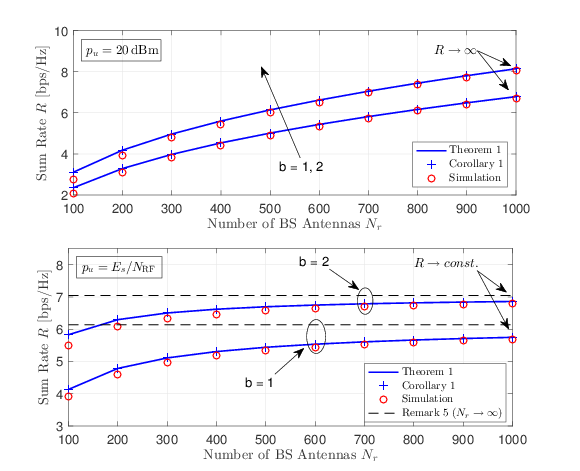}
\caption{Uplink sum rate of the analytical and simulation results for $b \in \{1, 2\}$ quantization bits with $N_u = 8$ users. 
Two different cases of the transmit power are considered: i) $p_u = 20$ dBm and ii) $p_u = E_s/N_{\rm RF}$ with $E_s = 45$ dBm.} 
\label{fig:remark4}
\end{figure}

%%%%%%%%%%%%%%%%%%%%%%%%%%%%
\subsection{Worst-Case Anaylsis Validation}
\label{sec:Validation}
%%%%%%%%%%%%%%%%%%%%%%%%%%%%

In this subsection, we validate Theorem \ref{thm:SE_uniform} and Corollary \ref{col:path_large}, and confirm the observations in Section \ref{sec:main}.
For simulation, user locations are fixed once they are dropped in the cell, which corresponds to the setting of the analytical derivations in Section \ref{sec:main}.
% Numerical Result 1 & 2
Fig. \ref{fig:SumRate} illustrates the sum rate for (a) $N_r \in \{128, 256,512\}$ BS antennas with $b = 2$ quantization bits and for (b) $N_r = 256$ with $b \in \{1, 2,\infty\}$ over different transmit power $p_u$.
The analytical results show accurate alignments with the simulation results in Fig. \ref{fig:SumRate}(a) and Fig. \ref{fig:SumRate}(b), which validates Theorem \ref{thm:SE_uniform} and Corollary \ref{col:path_large}.
The sum rates show the transition from the noise-limited regime to the interference-limited regime as $p_u$ increases.
% and then the growth becomes slower as the change from the noise-limited regime to the interference-limited regime. 
This observation corresponds to the convergence of the achievable rate with increasing transmit power in Remark \ref{remark:3}.
Notably, the sum rate with $b = \infty$ also tends to converge in Fig. \ref{fig:SumRate}(b) due to the interference power in \eqref{eq:rm3-2}.
% the non-uniform bit ADCs with $ b_i \in [1,3]$, $i = 1,\cdots, N_r$, are considered under the constraint of the equal ADC power consumption, i.e., $N_r P_{\rm ADC}(b) = \sum_{i = 1}^{N_r} P_{\rm ADC}(b_i)$.
%We call such non-uniform bit ADC case as $2$-bit non-uniform bit ADC case.
%Fig. \ref{fig:pu_ant10UE} presents the uplink sum rates for $N_r = 64, 128$ and $256$. 
%In all cases, the our analytical results are almost in accordance with the simulation results. 
%The rate scales proportionally with the number of BS antennas $N_r$ as more BS antennas capture more non-zero channel path gains and provide more aggregated channel gains.
%In addition, the sum rates for the homogeneous-ADC case are higher than the sum rates for the non-uniform bit ADC case, which agrees with our intuition for the optimality of ADC bit allocation without switching;  uniform ADC bit allocation is expected to be optimal for the system with fixed ADC bits due to the randomness of channels.
%Fig. \ref{fig:PuBit} shows the sum rates for different ADC bits in the uniform bit ADC scenario, and our analytical results precisely align with the simulation results.
%The gap between any two curves tends to increase with increasing $p_u$ since quantization error becomes more dominant than the error from noise.

% Numerical Result 3
We also evaluate our analytical results over the different number of BS antennas.
The fixed transmit power of 20 dBm ($p_u = 20$ dBm) and the power-scaling law ($p_u = E_s/N_{\rm RF}$) with $E_s = 45$ dBm are considered in Fig. \ref{fig:remark4}.
It is observed that Fig. \ref{fig:remark4} validates the derived approximations of the achievable rate for the different power assumptions and offers intuitions discussed in Remark \ref{remark:3} and \ref{remark:L_scale}:
%Since we assume that the number of propagation paths is proportional to the number of BS antennas, $\lambda_{\rm p} = \epsilon N_{\rm RF}$, 
the uplink sum rate with $p_u = 20$ dBm keeps increasing as $N_r$ increases, and 
%, which coincides with the discussion in Remark \ref{remark:3}.
we can maintain the sum rate by decreasing the transmit power $p_u$ proportionally to $1/N_r$ ($N_{\rm RF}  = \tau N_r$).
% as explained in Remark \ref{remark:L_scale}. 
In addition, the sum rate with the power-scaling law converges to \eqref{eq:power_scale} and can be improved by increasing the number of quantization bits (larger $\alpha$) as illustrated in Fig. \ref{fig:remark4}.
In Section \ref{subsec:num_sol}, the fixed-ADC approach serves the lower bound of the sum rate in the proposed architecture showing the similar trend to the BA strategies with performance gap.
% as it does not exploit the flexible ADC architecture. 
Therefore, the derived ergodic rate in Theorem \ref{thm:SE_uniform} explains general tradeoffs of the proposed system serving the lower bound of the sum rate.

\section{Conclusion}
\label{sec:Con}
%%%%%%%%%%%%%%%%%%%%%%%%%%%%

This paper proposes the hybrid MIMO receiver architecture with resolution-adaptive ADCs for mmWave communications.
Employing array response vectors for analog beamforming, we investigate the ADC bit-allocation problem to minimize the quantization distortion of received signals by leveraging the flexibility of ADC resolutions.
%As a closed-solution minimizing the quantization error under the power constraint, 
%One key finding is that it is beneficial to allocate more bits to the ADC associated with the stronger channel gain. 
One key finding is that the optimal number of ADC bits increases logarithmically proportional to the RF chain's SNR raised to the $1/3$ power.
%We further revised the BA algorithm to be noise-robust.
%We evaluated the proposed BA algorithms with switching at the channel coherence time and the time-scale of slowly changing channel characteristics.
%In both cases of ideal and more realistic switching time, 
{\color{black} The proposed algorithms outperform the conventional fixed ADCs in the proposed architecture in the low-resolution regime. 
In particular, the revised algorithm shows a higher capacity, sum rate
%The revised BA algorithm showed the higher ergodic sum rate than the fixed-ADC and mixed-ADC systems with consistent improvement regardless of the noise power.
and energy efficiency in any communication environment.}
% that of the fixed-ADC and mixed-ADC systems also with the higher sum rate 
Furthermore, the revised algorithm makes the quantization error of desired signals negligible while achieving higher energy efficiency than the fixed-ADC system.
Having negligible quantization distortion allows existing state-of-the-art digital beamforming techniques to be readily applied to the proposed system.
%Noting that the fixed-ADC system is the worst-case scenario of the proposed architecture as it does not exploit the dynamic ADCs, 
{\color{black} The approximated capacity expression captures the capacity that the proposed flexible ADC architecture can achieve adaptively for a given channel by using the revised algorithm.}
The derived ergodic rate from the worst-case analysis explains the tradeoffs of the proposed system in terms of system parameters, serving as the lower performance bound of the proposed system.
Therefore, for a future mmWave base station, this paper provides a spectrum- and energy-efficient mmWave receiver architecture with analysis and novel performance-increasing algorithms by means of adapting quantization resolution.

\begin{appendices}

\section{Proof of Proposition 1}
\label{appx:BA_power}
%%%%%%%%%%%%%%%%%%%%%%%%%%%%

By defining $z_i = 2^{-2b_i}$, $\bar z = 2^{-2\bar b}$ and $c_i = \sigma_{y_i}^2$ where $\sigma^2_{y_i}= p_u\|[\mathbf{H}]_{i,:}\|^2 + 1$, we can convert~\eqref{eq:opt_power} into a simpler form given as
%EQUATION
\begin{gather}
\label{eq:trans_problem}
{\bf \hat z} = \argmin_{{\bf z}  > {\bf 0}_{N_{\rm RF}}} {\bf c^\intercal z}  \quad \text{s.t.} \quad \sum_{i=1}^{N_{\rm RF} }{z_i^{-\frac{1}{2}}} \leq N_{\rm RF} \bar z^{-\frac{1}{2}}
\end{gather}
% where ${\bf v} \in \mathbb{R}^{N}$. 
% The cost function in~\eqref{eq:opt_power} becomes a linear function and the feasible set is a convext set.
%$\mathcal{S} = \{{\bf x}\in \mathbb{R}^N |\sum_{i=1}^{N}x_i^{-1/2}\leq N\bar x^{-1/2}, {\bf x > 0}_N\}$.
where ${\bf 0}_{N_{\rm RF}}$ is a $N_{\rm RF} \times 1$ zero vector. 
Note that~\eqref{eq:trans_problem} is the equivalent problem to~\eqref{eq:opt_power} and is a convex optimization problem. 
The global optimal solution of~\eqref{eq:opt_power} can be achieved by the KKT conditions for~\eqref{eq:trans_problem}.

By relaxing ${\bf z >0}_{N_{\rm RF}}$ to ${\bf z \geq 0}_{N_{\rm RF}}$ and defining ${\bf v}$ as ${\bf v} = \begin{bmatrix}
\sum_{i=1}^{N_{\rm RF}}{z_i^{-\frac{1}{2}}}- N_{\rm RF} \bar z^{-\frac{1}{2}} \\ 
 -{\bf z}
\end{bmatrix},$
%\begin{align}
%\nonumber
%{\bf v} = \begin{bmatrix}
%\sum_{i=1}^{N_{\rm RF}}{z_i^{-\frac{1}{2}}}- N_{\rm RF} \bar z^{-\frac{1}{2}} \\ 
% -{\bf z}
%\end{bmatrix},
%\end{align}
KKT conditions become
% EQUATION
\begin{align} 
\label{eq:kkt1}
{\bf c} + J_{\bf v}({\bf z})^\intercal{\pmb \mu} &= {\bf 0}_{N_{\rm RF}} 
\\ \label{eq:kkt2}
\mu_i\,v_i &= 0, \quad \forall \,i %\quad  i \in \{ 1,\cdots, N_{\rm RF}+1\}
\\ \label{eq:kkt3}
\mathbf{v} &\leq \mathbf{0}_{(N_{\rm RF}+1)}
\\ \label{eq:kkt4}
{\pmb \mu} &\geq \mathbf{0}_{(N_{\rm RF}+1)}
\end{align}
where the Jacobian matrix of $\mathbf{v}$ is 
$J_{\bf v}(\mathbf{z}) = \begin{bmatrix}
       \mathbf{p}
 &      -\mathbf{I}_{N_{\rm RF}}
\end{bmatrix}^\intercal$ with ${\bf p} = \left[-\frac{1}{2}z_1^{-\frac{3}{2}}, \cdots,-\frac{1}{2}z_{N_{\rm RF}}^{-\frac{3}{2}} \right]^\intercal$, and ${\pmb \mu} \in \mathbb{R}^{(N_{\rm RF}+1)}$ is the vector of the Lagrangian multipliers. 
Since $z_i \neq 0$, $i = 1, \cdots, N_{\rm RF}$, the Lagrangian multipliers become $\mu_j = 0$, $j = 2, \cdots, N_{\rm RF}+1$, from~\eqref{eq:kkt2}.
Hence, \eqref{eq:kkt1} guarantees $\mu_1 \neq 0$ as ${\bf c \neq 0}_{N_{\rm RF}}$, and~\eqref{eq:kkt2} gives $v_1 = 0$ meaning that the equality holds for the power constraint.
% Relating this to~\eqref{eq:kkt2}, we have $h_1 = 0$, which means that the equality holds for the power constraint.
From~\eqref{eq:kkt1} and~\eqref{eq:kkt2}, we have
% \begin{align}
% \label{eq:v}
$c_i = \frac{1}{2}z_i^{-\frac{3}{2}}\mu_1 \text{ and }\sum_{i=1}^{N_{\rm RF}}{z_i^{-\frac{1}{2}}}= N_{\rm RF} \bar z^{-\frac{1}{2}}$,
% \quad \forall i \in \{1,\cdots,N\} \\
%\label{eq:h} 
%&
% \end{align} 
which gives $\mu_1 =\left\{ \frac{\bar z^{\frac{1}{2}}}{N_{\rm RF}}\sum_{j=1}^{N_{\rm RF}}(2c_j)^{\frac{1}{3}}\right\}^3 > 0$.
%  $\mu_1^{\frac{1}{3}} =\frac{\bar x^{\frac{1}{2}}}{N}\sum_{j=1}^{N}(2v_j)^{\frac{1}{3}}$.
% Note that $\mu_1 >0$.  
Putting $\mu_1 =\left\{ \frac{\bar z^{\frac{1}{2}}}{N_{\rm RF}}\sum_{j=1}^{N_{\rm RF}}(2c_j)^{\frac{1}{3}}\right\}^3 $ into $c_i = \frac{1}{2}z_i^{-\frac{3}{2}}\mu_1$, we have 
% EQUATION
\begin{align}
	\label{eq:BA_mid}
	\hat z_i = \bar z \bigg\{ 1/{N_{\rm RF}}\cdot \sum_{j = 1}^{N_{\rm RF}}{\left({c_j}/{c_i}\right)}^{\frac{1}{3}}\bigg\}^2.
\end{align}
%\begin{align}
%	\label{eq:BA_mid}
%	\hat z_i = \bar z \left\{ \frac{1}{N_{\rm RF}}\sum_{j = 1}^{N_{\rm RF}}{\left(\frac{c_j}{c_i}\right)}^{\frac{1}{3}}\right\}^2.
%\end{align}
% \begin{align}
% x_i = \frac{\bar x}{N^2} \left(1+ \sum_{j \neq i}^{N}\frac{v_j}{v_i}\right)^2.
% \end{align}
Since $\hat z_i > 0$, the solution $\hat{ \mathbf{z}}$ meets the KKT conditions.
Using the definitions of $z_i, \bar z \text{ and }c_i$, we obtain \eqref{eq:opt_BA} from~\eqref{eq:BA_mid}.
\qed

{\color{black} 
%%%%%%%%%%%%%%%%%%%%%%%%%%%
\section{Proof of Proposition \ref{pr:maxGMI}}
\label{appx:maxGMI}
%%%%%%%%%%%%%%%%%%%%%%%%%%%

With the optimal combiner ${\bf w}^{\rm opt}_{n} = {\bf R}^{-1}_{{\bf y}_{\rm q}{\bf y}_{\rm q}}{\bf R}_{{\bf y}_{\rm q} s_n}$ \cite{liang2016mixed}, \eqref{eq:kappa} becomes
\begin{align}
	\label{eq:kappa_opt}
	\kappa({\bf w}^{\rm opt}_{n}\ {\bf b}) = {\bf R}^H_{{\bf y}_{\rm q}s_n}{\bf R}^{-1}_{{\bf y}_{\rm q}{\bf y}_{\rm q}}{\bf R}_{{\bf y}_{\rm q}s_n}. 
\end{align}
In the low SNR regime, ${\bf R}^H_{{\bf y}_{\rm q}{\bf y}_{\rm q}}$ is computed as
\begin{align}
	\nonumber
	\lim_{p_u \to 0}{\bf R}_{{\bf y}_{\rm q}{\bf y}_{\rm q}} &= \lim_{p_u \to 0}\Big(p_u {\bf D}_{\alpha}{\bf H}_{\rm b}{\bf H}_{\rm b}^H {\bf D}_{\alpha}^H +  {\bf D}_{\alpha}^2 + {\bf R}_{{\bf n}_{\rm q}{\bf n}_{\rm q}}\Big)
	\\ \nonumber
	& = \lim_{p_u \to 0}\Big({\bf D}_{\alpha} + {\bf D}_\alpha {\bf D}_\beta \,{\rm diag}\big(p_u{\bf H_{\rm b}}{\bf H}_{\rm b}^H\big)\Big)
	\\ \label{eq:Ryqyq}
	& = {\bf D}_{\alpha}.
\end{align}
The correlation vector ${\bf R}_{{\bf y}_{\rm q}s_n}$ is computed as
\begin{align}
	\label{eq:R_yqs}
	{\bf R}_{{\bf y}_{\rm q}s_n} = \mathbb{E}[{\bf y}_{\rm q}s_n] = \sqrt{p_u}{\bf D}_{\alpha}{\bf h}_{{\rm b},n}.
\end{align}
Using \eqref{eq:Ryqyq} and \eqref{eq:R_yqs}, $\kappa({\bf w}^{\rm opt}_{n}\ {\bf b})$ \eqref{eq:kappa_opt} becomes
\begin{align}
	\label{eq:kappa_opt2}
	\kappa({\bf w}^{\rm opt}_{n},\ {\bf b}) 
	%&= p_u {\bf h}_{{\rm b},n}^H{\bf W}_{\alpha}{\bf h}_{{\rm b},n} 
	&= p_u\sum_{i = 1}^{N_{\rm RF}} \Big(1-\frac{\pi \sqrt{3}}{2}2^{-2b_i}\Big)|h_{{\rm b},n,i}|^2
\end{align}
where $h_{{\rm b},n,i}$ is the $i$th element of ${\bf h}_{{\rm b},n}$.
Since we have $\log(1+x/(1-x)) = x + o(x)$ as $x \to 0$, the GMI becomes $I_{n}^{\rm GMI}({\bf w}_n^{\rm opt},{\bf b}) = \kappa({\bf w}_n^{\rm opt}, {\bf b}) + o(\kappa({\bf w}_n^{\rm opt}, {\bf b}))$ in the low SNR regime, where $o(\cdot)$ is little-o.
Thus, the objective function in the GMI maximization problem \eqref{eq:maxGMI} with the low SNR approximation becomes
\begin{align}
   	\nonumber
        \hat {\mathbf{{b}}}^{GMI} 
        &\simeq \argmax_{{\bf b}} \sum_{n=1}^{N_{u}} \kappa({\bf w}_n^{\rm opt}, {\bf b})
	\\ \label{eq:bGMI}
	&= \argmin_{{\bf b}} \sum_{n=1}^{N_{u}} \sum_{i=1}^{N_{\rm RF}} p_u \frac{\pi \sqrt{3}}{2}2^{-2b_i}|h_{{\rm b},n,i}|^2.
\end{align}
Note that \eqref{eq:bGMI} is equal to the objective function in \eqref{eq:opt_power_rev}.
This completes the proof.
\qed
}

%%%%%%%%%%%%%%%%%%%%%%%%%%%%
\section{Proof of Theorem \ref{thm:SE_uniform}}
\label{appx:SE_uniform}
%%%%%%%%%%%%%%%%%%%%%%%%%%%%

% EQUATION
%Applying Lemma \ref{lem1}, we approximate the achievable rate $R_{n}$ in \eqref{eq:AchievableRate} as
%\begin{align}
%\label{eq:rate_apprx}
%\tilde R_{n} = \log_2\left(1+ \frac{p_u\alpha \gamma_n \mathbb{E}\big[\|{\bf g}_{n}\|^4\big]}{\mathbb{E}\big[\tilde \Psi_{\bf G}\big]} \right)
%\end{align}
%where 
%\begin{align}
%\label{eq:E_psi}
%\mathbb{E}\big[\tilde \Psi_{\bf G}\big] =   \mathbb{E}\Bigg[p_u\alpha \sum\limits_{\substack{k = 1\\ k \neq n}}^{N_u}\gamma_k\big|{\bf g}_{n}^H {\bf g}_k\big|^2 + \alpha \|{\bf g}_{n}\|^2 +  (1-\alpha) {\bf g}_{n}^H {\rm diag}\big(p_u{\bf G} {\bf D} {\bf G}^H + {\bf I}_{N_{\rm RF}}\big) {\bf g}_{n}\Bigg].
%\end{align}
Since the beamspace channels $\bf g$ are sparse, we use an indicator function to characterize the sparsity. 
The indicator function $\mathds{1}_{\{i\in \mathcal{A}\}}$ is defined by
\begin{align}
\nonumber
\mathds{1}_{\{i\in \mathcal{A}\}} = 
\begin{cases}
   1 &\quad \text{if } i \in \mathcal{A}\\
    0 &\quad \text{else.}
\end{cases}
\end{align} 
Utilizing the function $\mathds{1}_{\{\cdot\}}$, we first model the $i$th complex path gain of the $n$th user $g_{i,n}$ as
\begin{align}
\label{eq:g}
g_{i,n} = \mathds{1}_{\{i \in \mathcal{P}_n\}}\xi_{i,n},\quad n = 1,\cdots,N_u
\end{align}
where $\mathcal{P}_n = \big\{i\, \big|\, g_{i,n} \neq 0, i = 1,\cdots, N_{\rm RF}\big\}$ and $\xi_{i,n}$ is an IID complex Gaussian random variable which follows $\mathcal{CN}(0,1)$. 
We compute the expectation of the number of propagation paths $\mathbb{E} [ L ]$. 
Since we assume $L \sim \max\{Q, 1\}$ with $Q \sim Poisson(\lambda_{\rm p})$, the expectation $\mathbb{E}[L]$ is derived as 
\begin{align}
	\label{eq:E_L}
	\mathbb{E}\big[L\big] &= e^{-\lambda_{\rm p}} + \sum_{\ell= 1}^{\infty}\ell \,\frac{\lambda_{\rm p}^{\ell}e^{-\lambda_{\rm p}}}{\ell!} = e^{-\lambda_{\rm p}} + \lambda_{\rm p}.
\end{align}
Similarly, $\mathbb{E}\big[L^2\big]$ can be given as
\begin{align}
	\label{eq:E_L^2}
	\mathbb{E}\big[L^2\big] &= e^{-\lambda_{\rm p}} + \sum_{\ell= 1}^{\infty}\ell^2 \,\frac{\lambda_{\rm p}^{\ell}e^{-\lambda_{\rm p}}}{\ell!} \stackrel{(a)}= e^{-\lambda_{\rm p}} + \lambda_{\rm p} + \lambda_{\rm p}^2
\end{align}
where (a) comes from $\mathbb{E}\big[Q^2\big] = \sum_{\ell= 1}^{\infty}\ell^2 \,\frac{\lambda_{\rm p}^{\ell}e^{-\lambda_{\rm p}}}{\ell!}$ and $\mathbb{E}\big[Q^2\big] = \text{Var}[Q] +\big\{\mathbb{E}[Q]\big\}^2$.

Now, we solve the expectations in Lemma \ref{lem1}. 
We have $|g_{i,n}|^2 = \mathds{1}_{\{i \in \mathcal{P}_n\}}|\xi_{i,n}|^2$ and $|\xi_{i,n}|^2$ is distributed as exponential random variable with mean of the value $1$, i.e., $|\xi_{i,n}|^2 \sim \exp(1)$.
Despite the fact that the dimension of ${\bf g}_{i,n}$ is $N_{\rm RF}$, $\|{\bf g}_{i,n}\|^2$ follows the chi-square distribution of $2L_n$ degrees of freedom  $\|{\bf g}_{n}\|^2 \sim \chi^2_{2L_n}$ due to the channel sparsity, where $L_n$ is the number of propagation paths for the $n$th user. 
Then, we derive the expectation of $\|{\bf g}_n\|^2$ for the AWGN noise power in \eqref{eq:E_psi} as
\begin{align}
	\label{eq:E_N2}
	\mathbb{E}\big[\|{\bf g}_n\|^2\big] & = \mathbb{E}\Big[\mathbb{E}\big[\|{\bf g}_n\|^2 \big | L_n\big]\Big] \stackrel{(a)} = e^{-\lambda_{\rm p}} + \lambda_{\rm p}
\end{align}
where (a) comes from $\|{\bf g}_{n}\|^2 \sim \chi^2_{2L_n}$ and \eqref{eq:E_L}.
Similarly, the expectation of the desired signal power in \eqref{eq:rate_apprx} is
\begin{align}
%\label{eq:E_sig_pow}
\nonumber
\mathbb{E}\Big[\|{\bf g}_{n}\|^4\Big]  &= \mathbb{E}\Big[\mathbb{E}\big[\|{\bf g}_{n}\|^4\,\big |\, L_n\big]\Big]
\\ \nonumber
&=  \mathbb{E}\bigg[\text{Var}\big[\|{\bf g}_{n}\|^2 |\, L_n\big]  + \Big\{\mathbb{E}\big[\|{\bf g}_{n}\|^2 |\, L_n \big]\Big\}^2\bigg]
\\ \label{eq:E_D}
&\stackrel{(a)}= \lambda_{\rm p}^2 + 2\lambda_{\rm p} + 2e^{-\lambda_{\rm p}}
\end{align}
%From \ref{eq:g}, the variance of $\|{\bf g}_{i,n}\|^2$ in \ref{eq:E_sig_pow} is expressed as
%\begin{align}
%\text{Var}\Big[\|{\bf g}_{i,n}\|^2\Big] &= \text{Var}\bigg[\sum_{i = 1}^{N_{\rm RF}} \mathds{1}_{\{i \in \mathcal{P}_n\}}|\xi_{i,n}|^2 \bigg]
%\\
%&\stackrel{(a)}{=}\sum_{i = 1}^{N_{\rm RF}} \text{Var}\Big[ \mathds{1}_{\{i \in \mathcal{P}_n\}}|\xi_{i,n}|^2 \Big]
%\end{align}
%where (a) is from the independence between the complex path coefficients $\xi_{i,n}$. 
%Since the indicator function $\mathds{1}_{\{i\in \mathcal{P}_{n}\}}$ is independent from the channel path coefficient $\xi_{i,n}$, $\text{Var}\big[ \mathds{1}_{\{i \in \mathcal{P}_n\}}|\xi_{i,n}|^2 \big]$ can be written as
%\begin{align}
% \text{Var}&\Big[ \mathds{1}_{\{i \in \mathcal{P}_n\}}|\xi_{i,n}|^2 \Big] \\
%& = \Big\{\mathbb{E}\big[\mathds{1}_{\{i \in \mathcal{P}_n\}}\big]\Big\}^2\text{Var}\big[|\xi_{i,n}|^2\big] + \Big\{\mathbb{E}\big[|\xi_{i,n}|^2\big]\Big\}^2\text{Var}\big[\mathds{1}_{\{i \in \mathcal{P}_n\}}\big]  + \text{Var}\big[\mathds{1}_{\{i \in \mathcal{P}_n\}}\big] \text{Var}\big[|\xi_{i,n}|^2\big] 
%\end{align}
where (a) comes from $\|{\bf g}_{n}\|^2 \sim \chi^2_{2L_n}$, \eqref{eq:E_L} and \eqref{eq:E_L^2}.

To further derive $\mathbb{E}\big[\tilde \Psi_{\bf G}\big]$ in \eqref{eq:E_psi}, we solve the inter-user interference power $\mathbb{E}\big[|{\bf g}_{n}^H {\bf g}_k|^2\big]$ for $k \neq n$, which is given as
\begin{align}
	\nonumber
	\mathbb{E}&\Big[|{\bf g}_{n}^H  {\bf g}_k|^2\Big] 
	= \mathbb{E}\left[\left(\sum_{i = 1}^{N_{\rm RF}} g_{i,n}^{*}g_{i,k}\right)\left(\sum_{j = 1}^{N_{\rm RF}} g_{j,n} g_{j,k}^*\right)\right]
	\\\label{eq:E_N1_b}
	&= \sum_{i=1}^{N_{\rm RF}}\mathbb{E}\Big[ |g_{i,n}|^2|g_{i,k}|^2\Big]\stackrel{(a)}= \sum_{i=1}^{N_{\rm RF}}\mathbb{E}\Big[\mathds{1}_{\{i \in \mathcal{P}_{n},i \in \mathcal{P}_{k} \}}\Big].
\end{align}
%(a) is from the independence between $g_{i,k}$ and $g_{j,k}$ for $ i \neq j$ and
Note that (a) comes from $g_{i,n} = \mathds{1}_{\{i \in \mathcal{P}_n\}}\xi_{i,n}$ defined in \eqref{eq:g} and the independence between $\xi_{i,n}$ and $\xi_{i,k}$ when $k \neq n$.
Furthermore, $\mathbb{E}\Big[\mathds{1}_{\{i \in \mathcal{P}_{n},i \in \mathcal{P}_{k} \}}\Big]$ in \eqref{eq:E_N1_b} can be computed as
\begin{align}
	\nonumber
	\mathbb{E}\Big[\mathds{1}_{\{i \in \mathcal{P}_{n},i \in \mathcal{P}_{k} \}}\Big] & \stackrel{(a)}= \bigg \{ \mathbb{E}\Big[\mathbb{E}\big[\mathds{1}_{\{i \in \mathcal{P}_{n}\}} \big | L_n\big] \Big]\bigg\}^2
	\\ \label{eq:E_1_pp}
	&=\Bigg(\frac{\mathbb{E}\big[L_n\big]}{N_{\rm RF}}\Bigg)^2 \stackrel{(b)}=\left(\frac{\lambda_{\rm p} + e^{-\lambda_{\rm p}}}{N_{\rm RF}}\right)^2
\end{align}
where (a) is from the IID of $L_n$ and the independence between the two events: $\{i \in \mathcal{P}_{n}\}$ and $\{i \in \mathcal{P}_{k} \}$, and (b) comes from \eqref{eq:E_L}.
%From \eqref{eq:Noise_E1_pre}, we have
%\begin{align}
%\label{eq:Noise_E1}
%\sum_{\substack{k = 1\\ k \neq n}}^{N_u}\gamma_k \mathbb{E}\Big[|{\bf g}_{n}^H {\bf W}_\alpha {\bf g}_k|^2\Big]
%= \left(\frac{L}{N_r}\right)^2\sum_{\substack{k=1\\k \neq n}}^{N_u}\sum_{i=1}^{N_r}\gamma_k\alpha_i^4.
%\end{align}
%Finally, we compute the last term $\mathbb{E}\big[{\bf g}_{n}^H{\bf W}_\alpha\mathbf{R}_{\mathbf{n}_{\rm q}\mathbf{n}_{\rm q}}{\bf W}_\alpha{\bf g}_{n}\big]$ in \eqref{eq:E_psi} as follows. 
%Noting that the covariance matrix $\mathbf{R}_{\mathbf{n}_{\rm q}\mathbf{n}_{\rm q}}= {\bf W}_\alpha {\bf W}_\beta \,{\rm diag}(p_u{\bf H_{\rm b}}{\bf H}_{\rm b}^H + {\mathbf{I}_{N_r}})$ \eqref{eq:cov2}, we express the $i$th diagonal entry of ${\rm diag}(p_u{\bf H_{\rm b}}{\bf H}_{\rm b}^H + {\mathbf{I}_{N_r}})$ as
%\begin{align}
%\label{eq:diagonal}
%\big[{\rm diag}(p_u{\bf H_{\rm b}}{\bf H}_{\rm b}^H + {\mathbf{I}_{N_r}})\big]_{i,i} =1+p_u\sum_{k = 1}^{N_u}\gamma_k |g_{i,k}|^2.
%\end{align}
Putting \eqref{eq:E_1_pp} into \eqref{eq:E_N1_b}, $\mathbb{E}\big[|{\bf g}_{n}^H {\bf g}_k|^2\big]$ finally becomes 
\begin{align}
	\label{eq:E_N1}
	\mathbb{E}\Big[|{\bf g}_{n}^H {\bf g}_k|^2\Big] = \frac{(\lambda_{\rm p} + e^{-\lambda_{\rm p}})^2}{N_{\rm RF}}.
\end{align}

%Similarly, $\mathbb{E}\big[\mathds{1}_{i \in \mathcal{P}_n}\big]$ is solved as 
%\begin{align}
%\mathbb{E}\big[\mathds{1}_{i \in \mathcal{P}_n}\big] & = \mathbb{E}\Big[\Big]
%\end{align}
Lastly, we compute the quantization noise power in \eqref{eq:E_psi} as
\small
\begin{align}
\nonumber
&\mathbb{E}\Big[{\bf g}_{n}^H {\rm diag}\big(p_u{\bf G} {\bf D}_\gamma {\bf G}^H + {\bf I}_{N_{\rm RF}}\big) {\bf g}_{n}\Big]
\\ \nonumber
&= \mathbb{E}\Bigg[\sum_{i=1}^{N_{\rm RF}}|g_{i,n}|^2 \bigg(p_u\sum_{\substack{k = 1\\ k \neq n}}^{N_u}\gamma_k |g_{i,k}|^2 + p_u \gamma_n |g_{i,n}|^2 + 1\bigg)\Bigg]
\\ \nonumber
&=\sum_{i=1}^{N_{\rm RF}}\bigg( p_u\sum_{\substack{k=1 \\ k\neq n}}^{N_u}\gamma_k \mathbb{E}\Big[|g_{i,k}|^2|g_{i,n}|^2\Big]+ \mathbb{E}\Big[ p_u\gamma_n |g_{i,n}|^4+ |g_{i,n}|^2\Big]\bigg)
\\ \nonumber
&\stackrel{(a)} = \sum_{i = 1}^{N_{\rm RF}}\bigg( p_u \sum_{\substack{k=1\\k \neq n}}^{N_u}\gamma_k \mathbb{E}\Big[\mathds{1}_{\{i \in \mathcal{P}_k, i \in \mathcal{P}_n\}}\Big]+ \big(2 p_u\gamma_n +1 \big)\mathbb{E}\Big[\mathds{1}_{\{i \in \mathcal{P}_n\}}\Big]\bigg)
\\ \label{eq:E_N3}
&\stackrel{(b)} = p_u\frac{(\lambda_{\rm p} + e^{-\lambda_{\rm p}})^2}{N_{\rm RF}}\sum_{\substack{k=1\\k \neq n}}^{N_u}\gamma_k + (\lambda_{\rm p} + e^{-\lambda_{\rm p}})(2p_u\gamma_n + 1)
\end{align}
\normalsize
where (a) and (b) are from \eqref{eq:g} and \eqref{eq:E_1_pp}, respectively.
Substituting \eqref{eq:E_N2}, \eqref{eq:E_D}, \eqref{eq:E_N1} and \eqref{eq:E_N3} into \eqref{eq:rate_apprx} and simplifying the equations, we derive the final result \eqref{eq:SE_uniform}.
\qed

{\color{black}
%%%%%%%%%%%%%%%%%%%%%%%%%%%
\section{Multi-Cell Scenario}
\label{appx:multi-cell}
%%%%%%%%%%%%%%%%%%%%%%%%%%%

\subsection{Bit-Allocation Problem}
Assuming a multi-cell MIMO receiver, let $\bar{\bf y}$ be the received signal after analog beamforming at the BS with out-of-cell interference.
Then, the signal is given as
\begin{align}
	\label{eq:multicell_y}
	\bar {\bf y} = \sqrt{p_u}{\bf H}_{\rm b} {\bf s} + \sqrt{p_u} \sum_{i = 1}^{N_c} {\bf H}^{i}_{\rm b}{\bf s}^{i} + {\bf n} = \sqrt{p_u} \bar{\bf H}_{\rm b} \bar {\bf s} + {\bf n}
\end{align}
with $\bar{\bf H}_{\rm b} = [{\bf H}_{\rm b}, {\bf H}_{\rm b}^{1}\dots,{\bf H}^{N_c}_{\rm b}] \text{, }  \bar{\bf s}^{\intercal} = [{\bf s}^{\intercal},({\bf s}^{1})^{\intercal},\dots,({\bf s}^{N_c})^{\intercal}]$ .
%
%\begin{align}
%	\nonumber 
%	&\bar{\bf H}_{\rm b} = [{\bf H}_{\rm b}, {\bf H}_{\rm b}^{1}\dots,{\bf H}^{N_c}_{\rm b}]
%	\\ \nonumber
%	&\bar{\bf s}^{\intercal} = [{\bf s}^{\intercal},({\bf s}^{1})^{\intercal},\dots,({\bf s}^{N_c})^{\intercal}].
%\end{align}
Here, $N_c$ denotes the number of interference cells, ${\bf H}_{\rm b}^i$ is the beamspace channel matrix from $N_u$ users in the cell $i$ to the BS in interest, ${\bf s}^i$ is the signals of users in the cell $i$, and ${\bf n}$ is the zero mean and unit variance complex Gaussian noise.
Note that $\sqrt{p_u} \sum_{i = 1}^{N_c} {\bf H}^{i}_{\rm b}{\bf s}^{i}$ is the out-of-cell interference from $N_c$ cells with $N_u$ users in each cell, and $\bar {\bf H}_{\rm b}$ is decomposed as $\bar {\bf H}_{\rm b} = \bar {\bf G} \bar {\bf D}_\gamma^{1/2}$ where $\bar {\bf G} = [ {\bf G}, {\bf G}^1, \dots, {\bf G}^{N_c}]$ and $\bar {\bf D}_\gamma = {\rm diag}({\bf D}_\gamma, {\bf D}_\gamma^1, \dots,{\bf D}_\gamma^{N_c} )$.
%is the matrix of the beamspace complex channel gains including the out-of-cell channels and $\bar {\bf D} = {\rm diag}({\bf D}, {\bf D}^1, \dots,{\bf D}^{N_c} )$ is the diagonal matrix of the large scale fading gains including the out-of-cell channels.  

Under the multi-cell assumption, the BA problem is also formulated as same as \eqref{eq:opt_power}
% EQUATION
%\begin{gather}
%	\nonumber
%	\hat {\mathbf{{b}}}^{\rm } 
%	= \argmin_{\mathbf{b}=[b_1,\cdots,b_{N_{\rm RF}}]^\intercal} \sum_{i=1}^{N_{\rm RF}}\mathcal{E}_{\bar y_i}(b_i)
%	\\ \nonumber
%	\text{s.t.} \quad \sum_{i=1}^{N_{\rm RF}} P_{\rm ADC}(b_i) \leq N_{\rm RF}P_{\rm ADC}(\bar b),\ {\bf b}\in\mathbb{R}^{N_{\rm RF}}.
%\end{gather}
with $\mathcal{E}_{\bar y_i}(b_i)  = \frac{\pi\sqrt{3}}{2}\sigma_{\bar y_i}^2\,2^{-2b_i}$. 
In this case, the difference of the BA problem for a single-cell and a multi-cell assumption is the variance, $\sigma^2_{\bar y_i}$.
%Accordingly, the change of the BA solution is shown to be trivial.
From \eqref{eq:multicell_y}, the variance of $\bar y_i$ becomes $\sigma^2_{\bar y_i} = p_u\|[\bar{\bf H}_{\rm b}]_{i,:}\|^2+1$. 
Then, following the same steps of the proof in Appendix \ref{appx:BA_power},  the MMSQE-BA solution is given as \eqref{eq:opt_BA} with ${\rm SNR}_i^{\rm rf} = p_u\|[\bar {\bf H}_{\rm b}]_{i,:}\|^2$.
%\begin{align}
%	\label{eq:opt_BA_multi}
%	\hat  b^{\rm}_i = \bar b +  \log_2\left(\frac{N_{\rm RF}\big(1+{\rm SNR}^{\rm rf}_i\big)^{\frac{1}{3}}}{\sum_{j = 1}^{N_{\rm RF}} \big(1+{\rm SNR}^{\rm rf}_j\big)^{\frac{1}{3}}}\right),& 
%	\\ \nonumber  i = 1,\cdots,N_{\rm RF}&
%\end{align}
%where ${\rm SNR}_i^{\rm rf} = p_u\|[\bar {\bf H}_{\rm b}]_{i,:}\|^2$.
% the change in the MMSQE-BA solution is trivial and 
Accordingly, the MMSQE-BA algorithm for the multi-cell MIMO additionally requires CSI for all interference channels.

There will be no difference in the revMMSQE-BA problem as it will be formulated by ignoring the noise and out-of-cell interference.
Consequently, the problem remains unchanged leading to the same solution as \eqref{eq:opt_BA_rev} even with the multi-cell assumption, and \eqref{eq:opt_BA_rev} for the multi-cell MIMO does not require interference CSI.
Therefore, the revMMSQE-BA is still effective in multi-cell MIMO.

\subsection{Ergodic Achievable Rate}
Assuming fixed-ADC systems with same quantization bits across all ADCs under the multi-cell MIMO, the quantized received signal becomes $\bar {\bf y}_{\rm q} = \alpha \bar {\bf y} + {\bf n}_{\rm q}$ where $\bar {\bf y}$ is in \eqref{eq:multicell_y}.
Then, after applying the MRC receiver, 
%the signal vector becomes 
%\begin{align}
%	\nonumber
%	&\bar {\bf y}_{\rm q}^{\rm mrc}  = {\bf H}_{\rm b}^H \bar {\bf y}_{\rm q}
%	\\ \nonumber
%	&= \alpha\sqrt{p_u}{\bf H}_{\rm b}^H {\bf H}_{\rm b} {\bf s} + \alpha \sqrt{p_u} \sum_{i = 1}^{N_c} {\bf H}_{\rm b}^H{\bf H}^{i}_{\rm b}{\bf s}^{i} + {\bf H}_{\rm b}^H{\bf n}  + {\bf H}_{\rm b}^H{\bf n}_{\rm q},
%\end{align}
the $n$th element of the signal vector is expressed as
\begin{align}
	\nonumber
	\bar y_{{\rm q},n}^{\rm mrc}  = &\, \alpha  \sqrt{p_u}{\bf h}_{{\rm b},n}^H {\bf h}_{{\rm b},n} s_{n} + \alpha  \sqrt{p_u}\sum_{\substack{k = 1\\ k \neq n}}^{N_u}{\bf h}_{{\rm b},n}^H {\bf h}_{{\rm b},k} s_k + 
	\\ \label{eq:y_multi}
	& \alpha \sqrt{p_u} \sum_{i=1}^{N_c}\sum_{k=1}^{N_u}{\bf h}_{{\rm b},n}^H {\bf h}_{{\rm b},k}^i s_k^i +\alpha {\bf h}_{{\rm b},n}^H   {\bf n} + {\bf h}_{{\rm b},n}^H  \bar{\bf n}_{\rm q}.
\end{align}
Accordingly, the noise-plus-interference power is given by
\begin{align}
	\nonumber
	& \bar \Psi_{\bf G} = p_u\alpha^2 \gamma_n \sum_{\substack{k = 1\\ k \neq n}}^{N_u}\gamma_k|{\bf g}_{n}^H {\bf g}_k|^2 +  p_u\alpha^2 \gamma_n\sum_{i=1}^{N_c}\sum_{k=1}^{N_u}\gamma_k^i |{\bf g}_{n}^H {\bf g}_k^i|^2 
	\\ \nonumber
	& + \alpha^2 \gamma_n \|{\bf g}_{n}\|^2 +\alpha (1-\alpha)\gamma_n {\bf g}_{n}^H {\rm diag}(p_u \bar {\bf G} \bar {\bf D}_\gamma \bar {\bf G}^H + {\bf I}_{N_{\rm RF}}) {\bf g}_{n}.
\end{align}
The approximated ergodic achievable rate for \eqref{eq:y_multi} can be derived in closed form by following similar steps of the proof of Theorem \ref{thm:SE_uniform} and is shown in Theorem \ref{thm:SE_uniform_multi}.
\begin{theorem}
	\label{thm:SE_uniform_multi} 
	Under the multi-cell assumption, the uplink ergodic achievable rate of the user $n$ in the considered system with fixed ADCs is derived in a closed-form approximation as
	% EQUATION
	\begin{align}
	\label{eq:SE_uniform_multi} 
	\tilde R_{n}^{\rm mc} = \log_2\left( 1+\frac{p_u \gamma_n \alpha \big(\lambda_{\rm p}^2  + 2\lambda_{\rm p} + 2e^{-\lambda_{\rm p}}\big)}{\bar\eta}\right)
\end{align}
where 
\footnotesize
\begin{align}
\nonumber
\bar \eta  = & \Big(\lambda_{\rm p}+ e^{-\lambda_{\rm p}}\Big)\cdot
\\ \nonumber
&\Bigg( \frac{p_u\big(\lambda_{\rm p} + e^{-\lambda_{\rm p}}\big)}{N_{\rm RF}} \Bigg(\sum_{\substack{k=1\\k\neq n}}^{N_u}\gamma_k + \sum_{j=1}^{N_c}\sum_{k=1}^{N_u}\gamma_k^j \Bigg) + 2p_u\gamma_n (1-\alpha) + 1 \Bigg).
\end{align}
\end{theorem}
\normalsize
\begin{proof}
	The expectation of  $\sum_{i=1}^{N_c}\sum_{k=1}^{N_u}\gamma_k^i |{\bf g}_{n}^H {\bf g}_k^i|^2$ can be handled similarly to that of $\sum_{k \neq n}^{N_u}\gamma_k|{\bf g}_{n}^H {\bf g}_k|^2 $, and we have
	\begin{align}
		\label{eq:interferecen_multi}
		\mathbb{E}\Bigg[\sum_{i=1}^{N_c}\sum_{k=1}^{N_u}\gamma_k^i|{\bf g}_{n}^H{\bf g}_{k}^i|\Bigg] 
		\stackrel{(a)} = \sum_{i=1}^{N_c}\sum_{k=1}^{N_u}\gamma_k^i\frac{(\lambda_{\rm p}+e^{-\lambda_{\rm p}})^2}{N_{\rm RF}} 
	\end{align}
	where (a) comes from \eqref{eq:E_N1}. 
	The quantization noise power ${\bf g}_{n}^H {\rm diag}(p_u \bar {\bf G} \bar {\bf D}_\gamma \bar {\bf G}^H + {\bf I}_{N_{\rm RF}}){\bf g}_{n}$ in $\bar \Psi_{\bf G}$ has similar structure as the one without out-of-cell interference and becomes

	\vspace{-1em}
	\footnotesize
	\begin{align}
			\nonumber
		\sum_{i = 1}^{N_{\rm RF}}|g_{i,n}|^2\bigg(p_u\sum_{\substack{k = 1\\ k \neq n}}^{N_u}\gamma_k |g_{i,k}|^2 + p_u \gamma_n |g_{i,n}|^2 + p_u\sum_{j, k}^{N_c, N_u}\gamma_k^j|g_{i,k}^j|^2 + 1\bigg).
	\end{align}
	\normalsize
	Consequently, we can derive the expectation of the quantization noise power by calculating the additional term associated with the out-of-cell interference channels as

	\vspace{-1em}
	\footnotesize
	\begin{align}
		\label{eq:quantization_noise_power_multi}
		\mathbb{E}\Bigg[\sum_{i = 1}^{N_{\rm RF}}|g_{i,n}|^2& \sum_{j=1}^{N_c}\sum_{k=1}^{N_u}\gamma_k^j|g_{i,k}^j|^2\Bigg] 
		\stackrel{(b)} =  \frac{(\lambda_{\rm p} + e^{-\lambda_{\rm p}})^2}{N_{\rm RF}} 	\sum_{j=1}^{N_c}\sum_{k=1}^{N_u}\gamma_k^j
	\end{align}
	\normalsize
	where (b) comes from $\mathbb{E}\big[|g_{i,n}|^2|g_{i,k}^j|^2\big] = \Big(\frac{\lambda_{\rm p} + e^{-\lambda_{\rm p}}}{N_{\rm RF}}\Big)^2$. 
	Combining \eqref{eq:interferecen_multi} and \eqref{eq:quantization_noise_power_multi} with the proof in Appendix \ref{appx:SE_uniform}, the ergodic uplink achievable rate can be approximated as \eqref{eq:SE_uniform_multi}.
\end{proof}
%Note that the ergodic rate for the multi-cell MIMO has the additional interference term from the users in the other cells.
}

\end{appendices}

\bibliographystyle{IEEEtran}
\bibliography{tsp17.bib}

% Generated by IEEEtran.bst, version: 1.13 (2008/09/30)
\begin{thebibliography}{10}
\providecommand{\url}[1]{#1}
\csname url@samestyle\endcsname
\providecommand{\newblock}{\relax}
\providecommand{\bibinfo}[2]{#2}
\providecommand{\BIBentrySTDinterwordspacing}{\spaceskip=0pt\relax}
\providecommand{\BIBentryALTinterwordstretchfactor}{4}
\providecommand{\BIBentryALTinterwordspacing}{\spaceskip=\fontdimen2\font plus
\BIBentryALTinterwordstretchfactor\fontdimen3\font minus
  \fontdimen4\font\relax}
\providecommand{\BIBforeignlanguage}[2]{{%
\expandafter\ifx\csname l@#1\endcsname\relax
\typeout{** WARNING: IEEEtran.bst: No hyphenation pattern has been}%
\typeout{** loaded for the language `#1'. Using the pattern for}%
\typeout{** the default language instead.}%
\else
\language=\csname l@#1\endcsname
\fi
#2}}
\providecommand{\BIBdecl}{\relax}
\BIBdecl

\bibitem{pi2012millimeter}
Z.~Pi and F.~Khan, ``{A millimeter-wave massive MIMO system for next generation
  mobile broadband},'' in \emph{Proc. Asilomar Conf. Signals, Systems and
  Comp.}, Nov. 2012, pp. 693--698.

\bibitem{swindlehurst2014millimeter}
A.~L. Swindlehurst, E.~Ayanoglu, P.~Heydari, and F.~Capolino,
  ``{Millimeter-wave massive MIMO: The next wireless revolution?}'' \emph{IEEE
  Comm. Mag.}, vol.~52, no.~9, pp. 56--62, Sep. 2014.

\bibitem{bai2015coverage}
T.~Bai and R.~W. Heath, ``{Coverage and rate analysis for millimeter-wave
  cellular networks},'' \emph{IEEE Trans. Wireless Comm.}, vol.~14, no.~2, pp.
  1100--1114, 2015.

\bibitem{han2015large}
S.~Han, I.~Chih-Lin, Z.~Xu, and C.~Rowell, ``{Large-scale antenna systems with
  hybrid analog and digital beamforming for millimeter wave 5G},'' \emph{IEEE
  Comm. Mag.}, vol.~53, no.~1, pp. 186--194, Jan. 2015.

\bibitem{mo2015capacity}
J.~Mo and R.~W. Heath, ``{Capacity analysis of one-bit quantized MIMO systems
  with transmitter channel state information},'' \emph{IEEE Trans. Signal
  Process.}, vol.~63, no.~20, pp. 5498--5512, Jul. 2015.

\bibitem{el2014spatially}
O.~El~Ayach, S.~Rajagopal, S.~Abu-Surra, Z.~Pi, and R.~W. Heath, ``{Spatially
  sparse precoding in millimeter wave MIMO systems},'' \emph{IEEE Trans.
  Wireless Comm.}, vol.~13, no.~3, pp. 1499--1513, 2014.

\bibitem{alkhateeb2014channel}
A.~Alkhateeb, O.~El~Ayach, G.~Leus, and R.~W. Heath, ``{Channel estimation and
  hybrid precoding for millimeter wave cellular systems},'' \emph{IEEE Journal
  Sel. Topics in Signal Process.}, vol.~8, no.~5, pp. 831--846, 2014.

\bibitem{el2012capacity}
O.~El~Ayach, R.~W. Heath, S.~Abu-Surra, S.~Rajagopal, and Z.~Pi, ``{The
  capacity optimality of beam steering in large millimeter wave MIMO
  systems},'' in \emph{IEEE Int. Work. Signal Process. Advances in Wireless
  Comm.}, 2012, pp. 100--104.

\bibitem{el2012low}
------, ``{Low complexity precoding for large millimeter wave MIMO systems},''
  in \emph{IEEE Int. Conf. Comm.}, 2012, pp. 3724--3729.

\bibitem{alkhateeb2013hybrid}
A.~Alkhateeb, O.~El~Ayach, G.~Leus, and R.~W. Heath, ``{Hybrid precoding for
  millimeter wave cellular systems with partial channel knowledge},'' in
  \emph{IEEE Info. Theory and App. Work.}, 2013, pp. 1--5.

\bibitem{liang2014low}
L.~Liang, W.~Xu, and X.~Dong, ``{Low-complexity hybrid precoding in massive
  multiuser MIMO systems},'' \emph{IEEE Wireless Comm. Letters}, vol.~3, no.~6,
  pp. 653--656, 2014.

\bibitem{alkhateeb2015limited}
A.~Alkhateeb, G.~Leus, and R.~W. Heath, ``{Limited feedback hybrid precoding
  for multi-user millimeter wave systems},'' \emph{IEEE Trans. Wireless Comm.},
  vol.~14, no.~11, pp. 6481--6494, 2015.

\bibitem{zhang2005variable}
X.~Zhang, A.~F. Molisch, and S.-Y. Kung, ``{Variable-phase-shift-based
  RF-baseband codesign for MIMO antenna selection},'' \emph{IEEE Trans. Signal
  Process.}, vol.~53, no.~11, pp. 4091--4103, 2005.

\bibitem{venkateswaran2010analog}
V.~Venkateswaran and A.-J. van~der Veen, ``{Analog beamforming in MIMO
  communications with phase shift networks and online channel estimation},''
  \emph{IEEE Trans. Signal Process.}, vol.~58, no.~8, pp. 4131--4143, 2010.

\bibitem{lee2008analog}
H.-S. Lee and C.~G. Sodini, ``{Analog-to-digital converters: Digitizing the
  analog world},'' \emph{Proc. of the IEEE}, vol.~96, no.~2, pp. 323--334,
  2008.

\bibitem{mezghani2007ultra}
A.~Mezghani and J.~A. Nossek, ``{On ultra-wideband MIMO systems with 1-bit
  quantized outputs: Performance analysis and input optimization},'' in
  \emph{IEEE Int. Symposium Info. Theory}, 2007, pp. 1286--1289.

\bibitem{risi2014massive}
C.~Risi, D.~Persson, and E.~G. Larsson, ``{Massive MIMO with 1-bit ADC},''
  \emph{arXiv preprint arXiv:1404.7736}, Apr. 2014.

\bibitem{jacobsson2015one}
S.~Jacobsson, G.~Durisi, M.~Coldrey, U.~Gustavsson, and C.~Studer, ``{One-bit
  massive MIMO: Channel estimation and high-order modulations},'' in \emph{IEEE
  Int. Conf. Comm. Work.}

\bibitem{mo2014channel}
J.~Mo, P.~Schniter, N.~G. Prelcic, and R.~W. Heath, ``{Channel estimation in
  millimeter wave MIMO systems with one-bit quantization},'' in \emph{Proc.
  Asilomar Conf. Signals, Systems and Comp.}, Nov. 2014, pp. 957--961.

\bibitem{wang2014multiuser}
S.~Wang, Y.~Li, and J.~Wang, ``{Multiuser detection for uplink large-scale MIMO
  under one-bit quantization},'' in \emph{IEEE Int. Conf. Comm.}, 2014, pp.
  4460--4465.

\bibitem{wen2016bayes}
C.-K. Wen, C.-J. Wang, S.~Jin, K.-K. Wong, and P.~Ting, ``{Bayes-optimal joint
  channel-and-data estimation for massive MIMO with low-precision ADCs},''
  \emph{IEEE Trans. Signal Process.}, vol.~64, no.~10, pp. 2541--2556, 2016.

\bibitem{mo2016channel}
J.~Mo, P.~Schniter, and R.~W. Heath~Jr, ``{Channel estimation in broadband
  millimeter wave MIMO systems with few-bit ADCs},'' \emph{arXiv preprint
  1610.02735, submitted IEEE Trans. Signal Process}, 2016.

\bibitem{mezghani2012capacity}
A.~Mezghani and J.~A. Nossek, ``{Capacity lower bound of MIMO channels with
  output quantization and correlated noise},'' in \emph{IEEE Int. Symposium
  Info. Theory}, 2012.

\bibitem{li2016channel}
Y.~Li, C.~Tao, L.~Liu, G.~Seco-Granados, and A.~L. Swindlehurst, ``{Channel
  estimation and uplink achievable rates in one-bit massive MIMO systems},'' in
  \emph{IEEE Sensor Array and Multichannel Signal Process Work.}, 2016, pp.
  1--5.

\bibitem{orhan2015low}
O.~Orhan, E.~Erkip, and S.~Rangan, ``{Low power analog-to-digital conversion in
  millimeter wave systems: Impact of resolution and bandwidth on
  performance},'' in \emph{IEEE Info. Theory and App. Work.}, Feb. 2015, pp.
  191--198.

\bibitem{fan2015uplink}
L.~Fan, S.~Jin, C.-K. Wen, and H.~Zhang, ``{Uplink achievable rate for massive
  MIMO systems with low-resolution ADC},'' \emph{IEEE Comm. Letters}, vol.~19,
  no.~12, pp. 2186--2189, Oct. 2015.

\bibitem{zhang2016spectral}
J.~Zhang, L.~Dai, S.~Sun, and Z.~Wang, ``{On the spectral efficiency of massive
  MIMO systems with low-resolution ADCs},'' \emph{IEEE Comm. Letters.},
  vol.~20, no.~5, pp. 842--845, Feb. 2016.

\bibitem{zhang2017performance}
J.~Zhang, L.~Dai, Z.~He, S.~Jin, and X.~Li, ``{Performance Analysis of
  Mixed-ADC Massive MIMO Systems over Rician Fading Channels},'' \emph{IEEE
  Journal Sel. Areas in Comm.}, vol.~35, no.~6, pp. 1327--1338, Jun. 2017.

\bibitem{mo2016achievable}
J.~Mo, A.~Alkhateeb, S.~Abu-Surra, and R.~W. Heath, ``{Achievable rates of
  hybrid architectures with few-bit ADC receivers},'' in \emph{VDE Int. ITG
  Work. Smart Antennas}, 2016, pp. 1--8.

\bibitem{liang2016mixed}
N.~Liang and W.~Zhang, ``{Mixed-ADC massive MIMO},'' \emph{IEEE Journal Sel.
  Areas in Comm.}, vol.~34, no.~4, pp. 983--997, Mar. 2016.

\bibitem{zhang2016mixed}
T.-C. Zhang, C.-K. Wen, S.~Jin, and T.~Jiang, ``{Mixed-ADC massive MIMO
  detectors: Performance analysis and design optimization},'' \emph{IEEE Trans.
  Wireless Comm.}, vol.~15, no.~11, pp. 7738--7752, 2016.

\bibitem{choi2016adc}
J.~Choi, B.~L. Evans, and A.~Gatherer, ``{ADC Bit Allocation under a Power
  Constraint for MmWave Massive MIMO Communication Receivers},'' in \emph{IEEE
  Int. Conf. Acoustics, Speech and Signal Process.}, 2017.

\bibitem{khan2012millimeter}
F.~Khan, Z.~Pi, and S.~Rajagopal, ``{Millimeter-wave mobile broadband with
  large scale spatial processing for 5G mobile communication},'' in \emph{IEEE
  Annual Allerton Conf. Comm., Control, and Computing}, 2012, pp. 1517--1523.

\bibitem{lozano2012yesterday}
A.~Lozano and N.~Jindal, ``{Are yesterday-s information-theoretic fading models
  and performance metrics adequate for the analysis of today's wireless
  systems?}'' \emph{IEEE Comm. Mag.}, vol.~50, no.~11, 2012.

\bibitem{yoo2002power}
J.~Yoo, D.~Lee, K.~Choi, and J.~Kim, ``{A power and resolution adaptive flash
  analog-to-digital converter},'' in \emph{ACM Int. Symposium on Low Power
  Electronics and Design}, 2002, pp. 233--236.

\bibitem{nahata2004high}
S.~Nahata, K.~Choi, and J.~Yoo, ``{A high-speed power and resolution adaptive
  flash analog-to-digital converter},'' in \emph{IEEE Int. System-on-Chip
  Conf.}, 2004, pp. 33--36.

\bibitem{rajashekar2008design}
G.~Rajashekar and M.~Bhat, ``{Design of Resolution Adaptive TIQ Flash ADC using
  AMS 0.35$\mu$m technology},'' in \emph{IEEE Int. Conf. Electronic Design},
  2008, pp. 1--6.

\bibitem{le2005analog}
B.~Le, T.~W. Rondeau, J.~H. Reed, and C.~W. Bostian, ``{Analog-to-digital
  converters},'' \emph{IEEE Signal Process. Mag.}, vol.~22, no.~6, pp. 69--77,
  2005.

\bibitem{sayeed2007maximizing}
A.~M. Sayeed and V.~Raghavan, ``{Maximizing MIMO capacity in sparse multipath
  with reconfigurable antenna arrays},'' \emph{IEEE Journal Sel. Topics in
  Signal Process.}, vol.~1, no.~1, pp. 156--166, 2007.

\bibitem{heath2016overview}
R.~W. Heath, N.~Gonzalez-Prelcic, S.~Rangan, W.~Roh, and A.~M. Sayeed, ``{An
  overview of signal processing techniques for millimeter wave MIMO systems},''
  \emph{IEEE Journal Sel. Topics in Signal Process.}, vol.~10, no.~3, pp.
  436--453, Feb. 2016.

\bibitem{akdeniz2014millimeter}
M.~R. Akdeniz, Y.~Liu, M.~K. Samimi, S.~Sun, S.~Rangan, T.~S. Rappaport, and
  E.~Erkip, ``{Millimeter wave channel modeling and cellular capacity
  evaluation},'' \emph{IEEE Journal Sel. Areas in Comm.}, vol.~32, no.~6, pp.
  1164--1179, 2014.

\bibitem{sayeed2002deconstructing}
A.~M. Sayeed, ``{Deconstructing multiantenna fading channels},'' \emph{IEEE
  Trans. Signal Process.}, vol.~50, no.~10, pp. 2563--2579, Nov. 2002.

\bibitem{mendez2016hybrid}
R.~M{\'e}ndez-Rial, C.~Rusu, N.~Gonz{\'a}lez-Prelcic, A.~Alkhateeb, and R.~W.
  Heath, ``{Hybrid MIMO architectures for millimeter wave communications: Phase
  shifters or switches?}'' \emph{IEEE Access}, vol.~4, pp. 247--267, Jan. 2016.

\bibitem{kim2015virtual}
T.~Kim and D.~J. Love, ``{Virtual AoA and AoD estimation for sparse millimeter
  wave MIMO channels},'' in \emph{IEEE Int. Work. Signal Process. Advances in
  Wireless Comm., 2015}, 2015, pp. 146--150.

\bibitem{fletcher2007robust}
A.~K. Fletcher, S.~Rangan, V.~K. Goyal, and K.~Ramchandran, ``{Robust
  predictive quantization: Analysis and design via convex optimization},''
  \emph{IEEE Journal Sel. Topics in Signal Process.}, vol.~1, no.~4, pp.
  618--632, 2007.

\bibitem{lee2014exploiting}
J.~Lee, G.-T. Gil, and Y.~H. Lee, ``{Exploiting spatial sparsity for estimating
  channels of hybrid MIMO systems in millimeter wave communications},'' in
  \emph{IEEE Global Comm. Conf.}, 2014, pp. 3326--3331.

\bibitem{gao2016channel}
Z.~Gao, C.~Hu, L.~Dai, and Z.~Wang, ``{Channel estimation for millimeter-wave
  massive MIMO with hybrid precoding over frequency-selective fading
  channels},'' \emph{IEEE Comm. Letters}, vol.~20, no.~6, pp. 1259--1262, 2016.

\bibitem{choi2016space}
J.~Choi, B.~L. Evans, and A.~Gatherer, ``{Space-time fronthaul compression of
  complex baseband uplink LTE signals},'' in \emph{in Proc. IEEE Int. Conf.
  Comm.}, July. 2016, pp. 1--6.

\bibitem{zhang2012general}
W.~Zhang, ``{A general framework for transmission with transceiver distortion
  and some applications},'' \emph{IEEE Trans. Comm.}, vol.~60, no.~2, pp.
  384--399, 2012.

\bibitem{zhang2016remark}
------, ``{A remark on channels with transceiver distortion},'' in \emph{IEEE
  Info. Theory and App. Work.,}, 2016, pp. 1--4.

\bibitem{zhang2014power}
Q.~Zhang, S.~Jin, K.-K. Wong, H.~Zhu, and M.~Matthaiou, ``{Power scaling of
  uplink massive MIMO systems with arbitrary-rank channel means},'' \emph{IEEE
  Journal Sel. Topics in Signal Process.}, vol.~8, no.~5, pp. 966--981, 2014.

\bibitem{raghavan2011sublinear}
V.~Raghavan and A.~M. Sayeed, ``{Sublinear capacity scaling laws for sparse
  MIMO channels},'' \emph{IEEE Trans. Info. Theory}, vol.~57, no.~1, pp.
  345--364, 2011.

\bibitem{mo2016hybrid}
J.~Mo, A.~Alkhateeb, S.~Abu-Surra, and R.~W. Heath~Jr, ``{Hybrid architectures
  with few-bit ADC receivers: Achievable rates and energy-rate tradeoffs},''
  \emph{arXiv preprint arXiv:1605.00668, submitted to IEEE Trans. on Wireless
  Comm.}, 2016.

\bibitem{chung20097}
H.~Chung, A.~Rylyakov, Z.~T. Deniz, J.~Bulzacchelli, G.-Y. Wei, and
  D.~Friedman, ``{A 7.5-GS/s 3.8-ENOB 52-mW flash ADC with clock duty cycle
  control in 65nm CMOS},'' in \emph{Symposium VLSI Circuits}, 2009, pp.
  268--269.

\end{thebibliography}

% Jinseok Choi
\begin{IEEEbiography}
[{\includegraphics[width=1in,height=1.25in,clip,keepaspectratio]{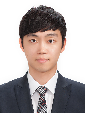}}]
{Jinseok Choi}  (S'14) received his BS (2014) degree in Electrical and Electronic Engineering from Yonsei University, South Korea, and his MS (2016) degree in Electrical and Computer Engineering, University of Texas at Austin, TX, USA. 
He is currently working toward the Ph.D. degree in Electrical and
Computer Engineering, University of Texas at Austin. 
He had worked as a graduate intern in Wireless Access Lab., Huawei, Plano, TX, USA, in 2016 and 2017. 
His research interest is developing and analyzing future wireless communication systems.
% using tools of multi-antenna theory and information theory.
\end{IEEEbiography}
% Brian L. Evans
\begin{IEEEbiography}
[{\includegraphics[width=1in,height=1.25in,clip,keepaspectratio]{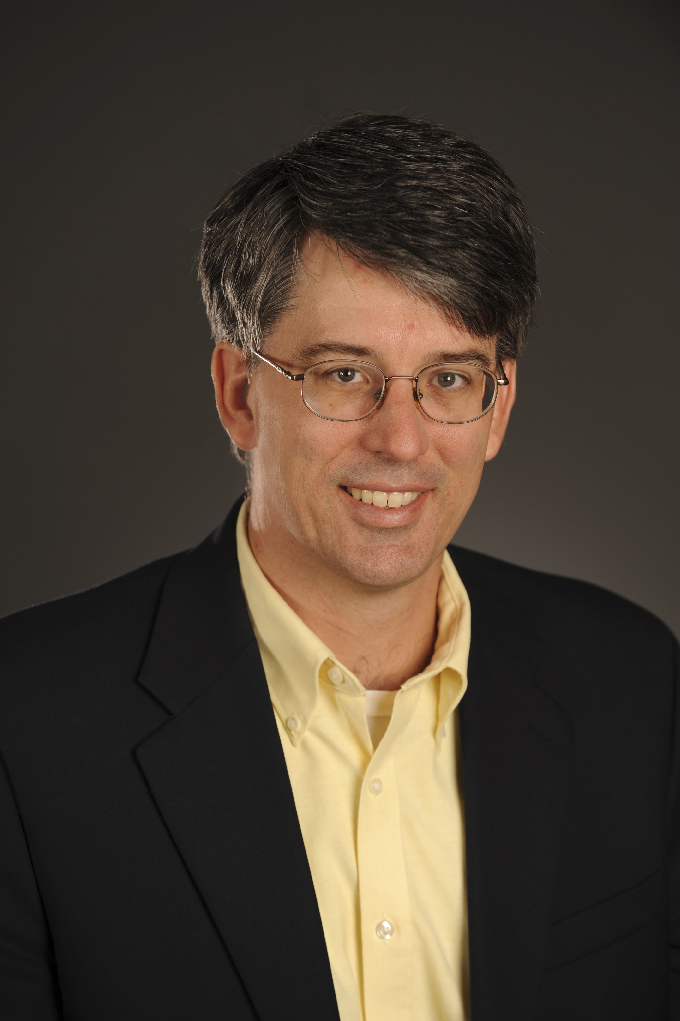}}]
{Brian L. Evans} (M'87-SM'97-F'09) received his BSEECS (1987) degree from the Rose-Hulman Institute of Technology, and his MSEE (1988) and PhDEE (1993) degrees from the Georgia Institute of Technology. From 1993 to 1996, he was a post-doctoral researcher at the University of California, Berkeley. Since 1996, he has been on the faculty at UT Austin.  Prof. Evans has published 250 refereed conference and journal papers, and graduated 27 PhD and 11 MS students. He has received teaching awards, paper awards, and an NSF CAREER Award.
\end{IEEEbiography}
% Alan Gatherer
\begin{IEEEbiography}
[{\includegraphics[width=1in,height=1.25in,clip,keepaspectratio]{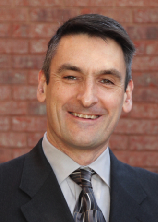}}]
{Alan Gatherer}  received his BEng (1988) degree from Strathclyde University in Scotland, and his MSEE (1989) and PhD (1993) degrees from Stanford University. From 1993 to 2010, he worked for Texas Instruments in Dallas Texas where he became a TI Fellow. In 2010 he Joined Huawei in Plano Texas. Dr. Gatherer has published more than 40 refereed conference and journal papers and has 75 awarded patents. He is the author of one book on the use of DSP in telecommunications and is currently Editor in Chief for the IEEE Communications Society Online Technology Newsfeed. He became an IEEE Fellow in 2015.
\end{IEEEbiography}

\vfill

\end{document}